\documentclass[10pt,conference]{IEEEtran}

\ifCLASSINFOpdf
\else
\fi

\usepackage{amsmath}
\usepackage{amssymb}
\usepackage{amsthm}
\usepackage{latexsym}
\usepackage{graphicx}
\usepackage[usenames,dvipsnames]{color}
\usepackage{amsfonts}
\usepackage{stmaryrd}
\usepackage{url}
\usepackage{ulem}
\usepackage{mathpartir}
\usepackage{enumerate}
\usepackage{prooftree}
\usepackage{framed}

\usepackage{verbatim}

\newcommand{\paper}[1]{#1}
\newcommand{\techrep}[1]{\ignore{#1}}

\usepackage{mathrsfs}
\usepackage{tikz}

\newtheorem{theorem}{Theorem}
\newtheorem*{theorem*}{Theorem}
\newtheorem{lemma}{Lemma}
\newtheorem{corollary}{Corollary}

\newtheorem{definition}{Definition}
\newtheorem{proposition}{Proposition}
\newtheorem{remark}{Remark}

\long\def\ignore#1{\relax}
\newcommand{\pierre}[1]{\red{#1}}

\newcommand{\fnsz}{\footnotesize}

\newcommand{\ie}{\textit{i.e.}\ }
\newcommand{\eg}{\textit{e.g.},\ }
\newcommand{\wrt}{w.r.t.\ }
\newcommand{\cf}{\textit{cf.}\ }
\newcommand{\resp}{resp.\ }

\newcommand{\leqs}{\leqslant}
\newcommand{\geqs}{\geqslant}

\newcommand{\fublainv}[1]{}
\newcommand{\fubla}[1]{\pierre{blahblahblah}}
\newcommand{\furef}[1]{\red{ref}}
\newcommand{\fucite}[1]{\red{cite}}

\renewcommand{\l}{\ell}

\colorlet{RED}{red} 

\newcommand{\ct}{\mathord{\cdot}}
\newcommand{\set}[1]{\{ #1 \}}
\newcommand{\mult}[1]{ [ #1 ]}
\newcommand{\ovl}[1]{\overline{#1}}

\newcommand{\eset}{\emptyset}

\newcommand{\emul}{\mult{\,}}
\newcommand{\tv}{o}
\newcommand{\inter}{\wedge}

\newcommand{\rstr}[1]{|_{#1}}
\newcommand{\omin}{\ominus}

\newcommand{\phd}{\phantom{.}}

\newcommand{\subeq}{\subseteq}

\newcommand{\bbA}{\mathbb{A}}
\newcommand{\bbAa}{\bbA_{a}}
\newcommand{\bbB}{\mathbb{B}}
\newcommand{\bbN}{\mathbb{N}}

\newcommand{\ttB}{\mathtt{B}}
\newcommand{\ttC}{\mathtt{C}}

\newcommand{\ttS}{\mathtt{S}}

\newcommand{\ttT}{\mathtt{T}}

\newcommand{\tti}{\mathtt{i}}
\newcommand{\ttj}{\mathtt{j}}

\newcommand{\ttw}{\mathtt{w}}

\newcommand{\scrD}{\mathscr{D}}

\newcommand{\scrR}{\mathscr{R}}
\newcommand{\scrRo}{\scrR_0}

\newcommand{\calC}{\mathcal{C}}
\newcommand{\calD}{\mathcal{D}}

\newcommand{\calT}{\mathcal{T}}

\newcommand{\secu}{^{\prime \prime}}

\newcommand{\ovla}{\ovl{a}}

\newcommand{\argN}{\bbN\setminus\set{0,\,1}} 
\newcommand{\Nmzo}{\bbN\setminus\set{0,1}}

\newcommand{\ttdom}{\mathtt{dom}}
\newcommand{\ttcodom}{\mathtt{codom}}

\newcommand{\tttr}{\mathtt{tr}}

\newcommand{\ttpos}{\mathtt{pos}}

\newcommand{\hlam}{h_{\lam}}
\newcommand{\ttax}{\mathtt{ax}}
\newcommand{\ttabs}{\mathtt{abs}}
\newcommand{\ttapp}{\mathtt{app}}

\newcommand{\ttAx}{\mathtt{Ax}}
\newcommand{\ttPol}{\mathtt{Pol}}

\newcommand{\ttTg}{\mathtt{Tg}}

\newcommand{\scrDw}{ \scrD_{\ttw}}

\newcommand{\scrRw}{ \scrR_{\ttw}}

\newcommand{\sysDOm}{\calD \Om} 

\newcommand{\iI}{{i \in I}}

\newcommand{\kK}{{k \in K}}

\newcommand{\jJi}{{j\in J(i)}}

\newcommand{\lam}{\lambda}
\newcommand{\Lam}{\Lambda}
\newcommand{\rew}{\rightarrow}
\newcommand{\rewsh}{\mathord{\rightarrow}}
\newcommand{\bred}{\rightarrow_{\!\beta}}

\newcommand{\beq}{\equiv_{\!\beta}}
\newcommand{\lx}{\lam x}
\newcommand{\ly}{\lam y}
\newcommand{\lz}{\lam z}
\newcommand{\lxy}{\lx y}

\newcommand{\rsx}{r\subx{s}}

\newcommand{\est}{(\,)}
\newcommand{\al}{\alpha}

\newcommand{\Gam}{\Gamma}
\newcommand{\Gami}{\Gam_i}

\newcommand{\Del}{\Delta}
\newcommand{\Deli}{\Del_i}

\newcommand{\sig}{\sigma}
\newcommand{\sigi}{\sig_i}

\newcommand{\Sig}{\Sigma}

\newcommand{\epsi}{\varepsilon}
\newcommand{\fom}{f^{\omega}}

\newcommand{\erewa}{\emul\rew \tv}

\newcommand{\om}{\omega}
\newcommand{\Om}{\Omega}

\newcommand{\arob}{\symbol{64}}

\newcommand{\TermV}{\mathscr{V}}
\newcommand{\TypeV}{\mathscr{O}}

\newcommand{\ju}[2]{#1 \vdash #2}

\newcommand{\mtv}{\mult{\tv}}
\newcommand{\msig}{\mult{\sig}}
\newcommand{\mtau}{\mult{\tau}}
\newcommand{\msigi}{\mult{\sigi}_{\iI}}

\newcommand{\sSk}{(S_k)_{\kK}}

\newcommand{\sSksh}{(\!S_k\!)_{\!k\!\in\!K}} 

\newcommand{\Rom}{R_{\tv}}
\newcommand{\rhoo}{\rho_o}
\newcommand{\mrhooom}{\mult{\rhoo}_\om}

\newcommand{\Delf}{\Del_f}
\newcommand{\cu}{\mathtt{Y}}
\newcommand{\cuf}{\cu_f}
\newcommand{\cul}{\cu_{\lam}}

\newcommand{\zhnfo}{x\,t_1\ldots t_q}

\newcommand{\hnfo}{\lx_1\ldots x_p.\zhnfo}

\newcommand{\Pex}{P_{\mathtt{ex}}}
\newcommand{\Piex}{\Pi_{\mathtt{ex}}}
\newcommand{\Sex}{S_{\mathtt{ex}}}

\newcommand{\loplus}{\!\,^{\oplus}}
\newcommand{\roplus}{\!\!\,^{\oplus}}
\newcommand{\lomin}{\!\,^{\omin}}
\newcommand{\romin}{\!\!\!\,^{\omin}\,}
\newcommand{\loast}{\!\,^{\oast}}

\newcommand{\Trl}[1]{\mathtt{Tr}_{\lam}(#1)}

\newcommand{\dom}[1]{\ttdom(#1)}
\newcommand{\codom}[1]{\ttcodom(#1)}

\newcommand{\Ax}{\ttAx}
\newcommand{\Hd}{\ttTg}

\newcommand{\Pol}[1]{\ttPol(#1)}
\newcommand{\Polp}[1]{\ttPol'(#1)}
\newcommand{\pos}[1]{\ttpos(#1)}
\newcommand{\posp}[1]{\ttpos'(#1)}

\newcommand{\AxP}{\Ax^P}
\newcommand{\AxPa}{\AxP_a}

\newcommand{\trP}[1]{\tttr^P(#1)}  

\DeclareMathOperator{\Rt}{\mathtt{Rt}}

\newcommand{\Res}{\mathtt{Res}}

\newcommand{\Lamuuu}{\Lam^{111}}
\newcommand{\Lamzzu}{\Lam^{001}}

\newcommand{\subb}[2]{[#1/#2]}
\newcommand{\subx}[1]{\subb{#1}{x}}

\newcommand{\Lves}[1]{\mathtt{Lves}(#1)}
\newcommand{\LB}{\Lves{B}}  

\newcommand{\rewa}{\!\rightarrow_{\! \mathtt{asc}}\!} 
\newcommand{\rewp}{\!\rightarrow_{\! \mathtt{pi}\!}}
\newcommand{\rewto}{\!\rightarrow_{ \mathtt{t1}\!}}  
\newcommand{\rewtt}{\!\rightarrow_{ \mathtt{t2}\!}}
\newcommand{\rewr}{\!\rightarrow_{\!\abs}\!}  

\newcommand{\rewdown}{\!\rightarrow_{\! \mathtt{down}\!}}

\newcommand{\rewbullet}{\!\rightarrow_{\!\bullet}\!}

\newcommand{\wer}{\leftarrow}
\newcommand{\awer}{\,\!_{\asc}\!\!\!\wer}
\newcommand{\pwer}{\,\!_{\mathtt{pi}}\!\!\!\wer}

\newcommand{\iden}{\equiv}
\newcommand{\idena}{\iden_{\asc}}

\newcommand{\denot}[1]{{[} \! {[} #1 {]}  \! {]} }

\newcommand{\cdeg}{\mathtt{cdeg}}

\newcommand{\p}{\mathtt{p}}
\newcommand{\pepsi}{\p_{\epsi}}
\newcommand{\bbot}{\p_{\bot}}

\newcommand{\ax}{\ttax}
\newcommand{\abs}{\ttabs}
\newcommand{\app}{\ttapp}

\newcommand{\axw}{\ttax_{\mathtt{w}}}

\DeclareMathOperator{\QRes}{\mathtt{QRes}}

\newcommand{\red}[1]{\textcolor{red}{#1}}
\newcommand{\blue}[1]{\textcolor{blue}{#1}}
\newcommand{\green}[1]{\textcolor{green}{#1}}
\newcommand{\purple}[1]{\textcolor{purple}{#1}}
\definecolor{greenone}{RGB}{0,200,140}

\definecolor{darkred}{RGB}{220,00,00}

\newcommand{\rewb}[1]{\stackrel{ #1}{\rightarrow}}
\newcommand{\code}[1]{\lfloor #1 \rfloor}
\newcommand{\codep}[1]{\code{#1}'}

\newcommand{\trck}[1]{\,\red{[#1]}}

\newcommand{\posPr}[1]{~ \red{\langle\!#1 \!\rangle}}

\newcommand{\rs}{\mathtt{rs}}

\newcommand{\rstra}{\rstr{a}}

\newcommand{\tra}{t\rstra}

\newcommand{\tri}{\rhd}

\newcommand{\asc}{\mathtt{asc}}

\newcommand{\Asc}[1]{\mathtt{Asc}(#1)}
\newcommand{\Ascp}[1]{\mathtt{Asc}'(#1)}

\newcommand{\juRax}[2]{
\ju{#1:\mult{#2}}{#1:#2}
  }
\newcommand{\axxRo}[1]{\juRax{x}{#1}}

\newcommand{\juaxtt}{\axxRo{\tau}}

\newcommand{\juGtt}{\ju{\Gam}{t:\tau}}

\newcommand{\juxkT}{\ju{x:(k\cdot T)}{x:T}}

\newcommand{\juCtt}{\ju{C}{t:T}}

\newcommand{\juCtpt}{\ju{C}{t':T}}

\newcommand{\ttsupp}{\mathtt{supp}}
\newcommand{\ttbisupp}{\mathtt{bisupp}}
\newcommand{\supp}[1]{\ttsupp(#1)}
\newcommand{\bisupp}[1]{\ttbisupp(#1)}

\newcommand{\ttThr}{\mathtt{Thr}}

\newcommand{\Thr}[1]{\ttThr(#1)}
     
\newcommand{\ttthr}{\mathtt{thr}}
\newcommand{\thr}[1]{\ttthr(#1)}    
\newcommand{\thrp}[1]{\ttthr'(#1)}  

\newcommand{\thbot}{\theta_{\bot}}
\newcommand{\thepsi}{\theta_{\epsi}}
\newcommand{\bbBm}{\bbB_{\min}}

\newcommand{\rewc}[1]{\stackrel{#1}{\rew}}    

\newcommand{\rrewc}[1]{\stackrel{#1}{\tilde{\rew}}} 
\newcommand{\rrew}{\tilde{\rightarrow}} 

\newcommand{\werr}{\tilde{\wer}}

\newcommand{\rrewt}{\rrew_{\mathtt{t}}}  
\newcommand{\rrewto}{\rrew_{\mathtt{t1}}}  
\newcommand{\rrewtt}{\rrew_{\mathtt{t2}}}
\newcommand{\rrewr}{\rrew_{\abs}}

\newcommand{\rrewdown}{\rrew_{\!\mathtt{down}}}
\newcommand{\rrewb}{\tilde{\rew}_{\bullet}}

\newcommand{\IM}{\mathcal{I}}
\newcommand{\IMi}{\IM_i}

\newcommand{\Gamu}{\Gam_u}

\newcommand{\Gamt}{\Gam_t}

\newenvironment{proofsketch}
        {\bgroup \noindent \textit{Proof sketch.}}{\hfill \qedsymbol \egroup\medskip}

\newcommand{\trans}[3]{ 
  \begin{scope}[xshift=#1cm,yshift=#2cm]
     #3
  \end{scope}
}

\newcommand{\transh}[2]{ 
  \begin{scope}[xshift=#1cm]
     #2
  \end{scope}
}

\newcommand{\drawlabnode}[3]{  
     \trans{ #1 }{ #2 }{ 
       \draw (0,0) node {\small #3 };
       \draw (0,0) circle (0.23);
     }
}

\newcommand{\drawarob}[2]{   
  \trans{#1}{#2}{
       \draw (0,0) node {\small $\arob$};
       \draw (0,0) circle (0.23);
    }
}

\newcommand{\drawtri}[3]{   
  \trans{#1}{#2}{  
   \draw (0,0) --++ (0.55,0.9) --++ (-1.1,0) -- cycle ;
    \draw (0.07,0.6) node{\small#3} ; 
  }
}

\newcommand{\drawsmalltriin}[3]{   
  \trans{#1}{#2}{  
    \draw (0,0) --++ (0.35,0.65) --++ (-0.7,0) --++ (0.35,-0.65) ;
    \draw (0,0.38) node{\fnsz #3} ; 
  }
}

\newcommand{\blockunary}[3]{ 
  \drawlabnode{#1}{#2}{#3}
  \trans{#1}{#2}{
  \draw (0,0.23) --++ (0,0.44);
}
  }

\newcommand{\blocka}[2]{  
  \trans{#1}{#2}{  
    \drawarob{0}{0}
    \draw (-0.12,0.18) -- (-0.53,0.72); 
    \draw (0.12,0.18) -- (0.53,0.72);
  }
}

\newcommand{\blockappl}[2]{ 
  \trans{#1}{#2}{  
    \drawarob{0}{0}
    \draw (-0.12,0.18) -- (-0.53,0.72); 
 } 
}

\newcommand{\inputtypeleft}[3]{
  \trans{#1}{#2}{
 \draw (-0.9,0.7) node{\red{#3}}; 
 \red{\draw [>=stealth, ->] (-0.7,0.5)--(-0.25,0.25) ; } }
}

\newcommand{\inputtypeleftn}[3]{ 
  \trans{#1}{#2}{
     \draw (-0.9,0.7) node{#3}; 
     \draw [>=stealth, ->] (-0.7,0.5)--(-0.25,0.25) ; } 
}

\newcommand{\inputtyperightn}[3]{  
  \trans{#1}{#2}{
 \draw (0.45,0.7) node{#3}; 
 \draw [>=stealth,->] (0.25,0.55)--(0.15,0.25) ; }
}

\newcommand{\outputtypeleft}[3]{
  \trans{#1}{#2}{
  \draw (0,-0.35) node [left] {#3} ; 
}}

\newcommand{\outputtyperight}[3]{
  \trans{#1}{#2}{
  \draw (-0.2,-0.25) node [right] {#3} ; 
}}

\begin{document}
\title{All the $\lam$-Terms are Meaningful for the Infinitary Relational Model}
\author{\IEEEauthorblockN{Pierre Vial}
\IEEEauthorblockA{
INRIA (LS2N CNRS)\\
Email: pierre.vial@inria.fr}
}
\maketitle

\renewcommand{\em}{\it}

\begin{abstract}
Infinite types and formulas are known to have really curious and unsound behaviors. For instance, they allow to type $\Om$, the auto-autoapplication and they thus do not ensure any form of normalization/productivity. Moreover, in most infinitary frameworks, it is not difficult to define a type $R$ that can be assigned to \textit{every} $\lam$-term. However, these observations do not say much about what coinductive (\ie infinitary) type grammars are able to provide: it is for instance very difficult to know what types (besides $R$) can be assigned to a term in this setting. We begin with a discussion on the expressivity of different forms of infinite types.
Using the resource-awareness of sequential intersection types (system $\ttS$) and tracking, we then prove that infinite types are able to characterize the \textit{order} (arity) of every $\lam$-terms and that, in the infinitary extension of the relational model, every term has a ``meaning'' \ie a non-empty denotation. From the technical point of view, we must deal with the total lack of productivity guarantee for typable terms: we do so by importing methods inspired by first order model theory.
\ignore{ 

Certain type assignment systems are known to ensure or characterize normalization. The grammar of the types they feature is usually {\it inductive}. It is easy to see that, when types are {\it coinductively} generated, we obtain unsound type systems (meaning here that they are able to type some {\it mute} terms). Even more, for most of those systems, it is not difficult to find an argument proving that {\it every} term is typable ({\it complete unsoundness}). However, this argument does not hold for {\it relevant} intersection type systems (ITS), that are more restrictive because they forbid {\it weakening}. Thus, the question remains: are relevant ITS featuring coinductive types -- despite being unsound -- still able to characterize some bigger class of terms? We show that it is actually not the case: every term is typable in a standard coinductive ITS called $\scrD$. Moreover, we prove that semantical information can be extracted from the typing derivations of $\scrD$, as the {\it order} of the typed terms. Our work also implicitly provides a new {\it non sensible} relational model for the pure $\lam$-calculus.}

\end{abstract}

\newtheorem{goal}{Goal}

\section{Introduction (Infinite types)}

\subsection{Some semantical aspects of infinite types}
\label{ss:intro-sem-aspects-infty-types}
It is well-known that the mere fact of allowing infinite formulas gives birth to unsound/contradictory proof systems. For instance, let $A$ be \textit{any} formula. We then define the infinite formula $F_A:=(((\ldots) \rew A)\rew A)\rew A$ as \ie $F_A=F_A\rew A$ (the letter ``F'' stands for ``fixpoint'')
. Using $F_A$, we may  write the proof of $A$ given by Fig.\;\ref{fig:infinite-formulas-unsound}.
\begin{figure}[b]
\vspace{-0.8cm}
\[\infer{
\infer{
\infer{
\infer{\phd}{\ju{F_A}{F_A\ \ie F_A\rew A}}\\
\infer{\phd}{\ju{F_A}{F_A}}}{
  \ju{F_A}{A}}}{
\ju{}{F_A\rew A\ \ie F_A}
}\\
\infer{
\infer{
\infer{\phd}{\ju{F_A}{F_A}}\\
\infer{\phd}{\ju{F_A}{F_A}}}{
  \ju{F_A}{A}}}{
\ju{}{F_A}
}}{
  \ju{}{A}}\]
  \caption{The unsoundness of infinite formulas}
  \label{fig:infinite-formulas-unsound}
  \end{figure}

Now, by the Curry-Howard correspondence identifying proofs with programs -- here, $\lam$-terms --, this does not come as a surprise that all types $A$ are inhabited when one allows infinite types. Indeed, the proof above is a simple typing of $\Om:=\Del$, with $\Del=\lx.x\,x$: see Fig.\;\ref{fig:infinite-simple-types-unsound}.

Thus, every type $A$ is inhabited by $\Om$. But given a $\lam$-term $t$, what types $A$ does $t$ inhabit? A first observation is that every $\lam$-term can be easily typed when infinite types are allowed: let us just define $R$ (standing for ``reflexive'') by $R=R\rew R$. Thus, $R=(R\rew R)\rew( R\rew R)=\ldots$ 
Then, it is very easy to inductively type every term with $R$. In the inductive steps below, $\Gam$ denotes a context that assigns $R$ to every variable of its domain: \label{inpage:comp-unsound-inf-sts}
\[\infer{\phd}{\ju{\Gam;x:R}{x:R} } \hspace{1cm} \infer{\ju{\Gam;x:R}{t:R}}{\ju{\Gam}{\lx.t:R \rew R\ \ (=R)} }\]
\[\infer{\ju{\Gam}{t:R\  (=R\rew R)} \hspace{0.7cm} \ju{\Gam}{u:R}}{\ju{\Gam}{t\,u:R}}\]
Thus, every $\lam$-term inhabits the type $R$, but note that this does not answer the former question: what types does a term $t$ inhabit? As we will see, this question is extremely complex and we will chiefly focus on one aspect of this problem, namely, the typing constraints caused by the \textbf{order} of the $\lam$-terms. Intuitively, the order of $\lam$-term $t$ is its arity \ie it is the supremal $n$ such that $t\bred \lx_1\ldots x_n.u$ (for some term $u$): the order of $t$ is the number of abstractions that one can ouput from $t$. For instance, $\Om$ is of order 0 (it is a \textbf{zero term}), the HNF $\lx_1x_2.x\,u_1\,u_2\,u_3$ (with $u_1,u_2,u_3$ terms) is of order 3 and the term $\cul:= (\lx.\ly.xx)\lx.\ly.xx$ (satisfying $\cul\bred \ly.\cul$ and thus, $\cul \bred^n \underbrace{\ly.\ldots.y}_{n}.\cul$) is of infinite order.

The \textbf{order of a type} is the number of its top-level arrows \eg if $\tv_1,\tv_2$ are \textit{type atoms} (or \textit{type variables}), $\tv_1\rew \tv_1$, $\tv_1$, $\tv_1\rew \tv_2\rew \tv_1$ and $(\tv_1 \rew \tv_2)\rew \tv_1$ are of respective orders  1, 0, 2, 1. Via Curry-Howard, the abstraction constructor $\lx$ corresponds to an implication introduction (\textit{modus tollens}). Thus, in most type systems, a typed term of the form $\lx_1\ldots \lx_n.u$ is typed with an arrow of order $\geqslant n$. For instance, if $\lx.\ly.u$ is typed, then it is so with a type of the form $A\rew B\rew C$.

\begin{figure}[b]
\vspace{-0.8cm}
\[\infer{
\infer{
\infer{
\infer{\phd}{\ju{x\!:\!F_A\!}{\!x\!:\!F_A}}\hspace{-0.15cm}
\infer{\phd}{\ju{x\!:\!F_A\!}{\!x\!:\!F_A}}}{
  \ju{x:F_A}{x\,x:A}}}{
\ju{}{\lx.x\,x:F_A\rew A}
}
\infer{
\infer{
\infer{\phd}{\ju{x\!:\!F_A\!}{\!x\!:\!F_A}}\hspace{-0.15cm}
\infer{\phd}{\ju{x\!:\!F_A\!}{\!x\!:\!F_A}}}{
  \ju{x:F_A}{x\,x:A}}}{
\ju{}{\lx.x\,x:F_A}
}}{
  \ju{}{\Om:A}}\]
  \caption{Typing $\Om$ using Fig.\;\ref{fig:infinite-formulas-unsound}.}
  \label{fig:infinite-simple-types-unsound}
  \end{figure}

Now, if a type system satisfies \textbf{subject reduction} (as every type system should), meaning that typing is stable under reduction, the above observation entails that, if a typable term is of order $n$, then it is typable only with types of order $\geqslant n$ (the order of a term is a lower bound for the order of its possible types). Equivalently, if $t$ is typed with $B$, then the order of $B$ \textit{statically} (\ie without reduction) gives an upper bound to the order of $t$. A finite type has a finite order (whereas the finite \textit{term} $\cul$ has an infinite order), but an infinite type may have an infinite order \eg $R$ defined by $R=R\rew R$ above. This shows that the typing of any term $t$ with $R$ is trivial and does not bring any information, since the order of a term is $\leqs \infty$\ldots However, the facts that, by allowing infinite types, (1) one can type every term with $R=R\rew R$  and (2) one can type $\Om$ with any type $A$, do not mean that finding the types that can be assigned to a given term $t$ is an easy problem in this setting:  actually, this turns out to be  very difficult (see \S\;\ref{ss:intro-difficulty-comp-unsound}).

More interestingly, intersection type systems (i.t.s.), introduced by Coppo-Dezani\ \cite{CoppoD80}, generally satisfy \textbf{subject expansion}, meaning that typing is stable under anti-reduction. Those systems are designed to ensure equivalences of the form ``$t$ is typable iff $t$ is normalizing'' and also provide \textit{semantical} proofs  of \fublainv{en remettre une couche sur combien c intéressant} non-type-theoretic properties such as ``$t$ is weakly normalizing iff the leftmost-outermost reduction strategy terminates on $t$''\;\cite{Krivine93}. From subject expansion and the typing of \textbf{normal forms (NF)} \ie terminal states, i.t.s. are actually able to \textit{capture} the order of \textit{some} $\lam$-terms. For instance, if an i.t.s.  characterizes \textbf{head normalization}, then, every HN term $t$ of order $n$ is typable with a type whose order is also \textit{equal} to $n$ (and not only bounded below by $n$).

Let us informally understand why the order of the typable terms is usually captured by i.t.s. For instance,  i.t.s. characterizing HN usually allow arrow types to have an empty source\footnote{In this article, we put aside \textit{non-strict} types systems  featuring $\underline{\Om}$-rules allowing every term $t$ to be typed with a special type constant $\underline{\Om}$ (\eg system $\sysDOm$~\cite{Krivine93,CoppoD80}), in which types are less constrained by the order of terms.
}
(that we generically denote by $\eset$), meaning that the underlying functions do not look at their argument \ie if $t:\eset \rew B$, then $t\,u$ is typable with $B$ for \textit{any} term $u$. Then, with such arrow types, it is indeed very easy to type any \textbf{head normal form (HNF)} while capturing their order. It is done on Fig.\;\ref{fig:typing-hnf-its}: one just assigns $\eset \rew \ldots \ldots \eset \rew \tv$ (order $q$) to the head variable $x$, so that $\zhnfo$ is typed with $\tv$ and the HNF $\hnfo$, whose order is $p$, is typed with an arrow type of order $p$.  

Then, by subject expansion, one concludes that every HN term of order $n$ is typable with a type of order $n$. The same kind of argument can be adapted to i.t.s. characterizing weak or strong normalization, which also usually capture the order of their typable terms.

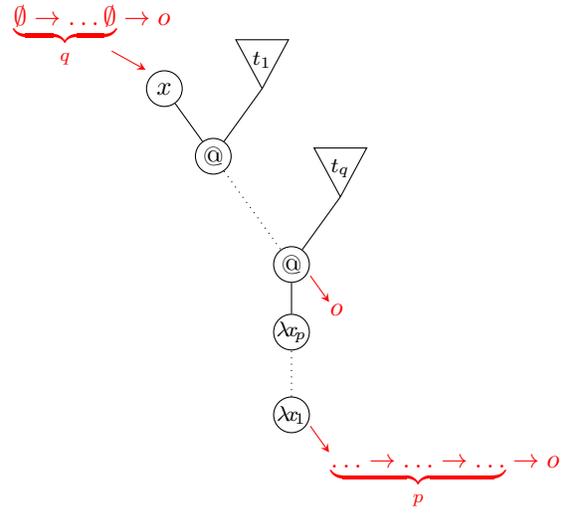
\begin{figure}[t]
\begin{tikzpicture}

\inputtypeleft{-1.69}{2.34}{$\underbrace{\eset \rew \ldots \eset}_q \rew \tv$ }

\node[draw,circle,minimum width=3ex] (X) at (-1.69,2.34) {};
\node at (X) {$x$};

\coordinate (T1) at (-0.39,2.34) ;
\drawsmalltriin{-0.39}{2.34}{$t_1$};

\node[draw,circle,minimum width=3ex] (A1) at (-1.04,1.44) {};
\node at (A1) {$\arob$};
\draw (A1) -- (T1);

\coordinate (Tq) at (0.65,0.9) ;
\drawsmalltriin{0.65}{0.9}{$t_q$};

\node[draw,circle,minimum width=3ex] (Aq) at (0,0) {};
\node at (Aq) {$\arob$};
\draw (X) -- (A1);
\draw [dotted] (A1) -- (Aq);
\draw (Aq) -- (Tq);

\red{\draw [>=stealth, -> ] (0.25,-0.15)-- (0.5,-0.5) ; 
       \draw (0.6,-0.6) node{$\tv$} ;
    }

\node[draw,circle,minimum width=3ex] (Lp) at (0,-0.9) {};
\node at (Lp) {\fnsz $\lam \!x_{\!p}$};
\draw (Aq) -- (Lp);

\node[draw,circle,minimum width=3ex] (L1) at (0,-2) {};
\node at (L1) {\fnsz $\lam \!x_{\!1}$};
\draw [dotted](L1) -- (Lp);

\red{\draw [>=stealth, -> ] (0.25,-2.15)-- (0.5,-2.5) ; 
       \draw (0.4,-2.9) node[right] {$\underbrace{\ldots \rew \ldots \rew \ldots}_{p}\rew \tv$} ;
    }

\end{tikzpicture}
\caption{Typing a Head Normal Form in an I.T.S.}
\label{fig:typing-hnf-its}
\vspace{-0.7cm}
\end{figure}


\subsection{In Search for Infinite Denotations}
\label{ss:intro-search-inf-denot}
Independently from normalization properties, another important facet of intersection type systems is that they also provide \textbf{denotational models} for the $\lam$-calculus \ie they associate to each $\lam$-term a denotation\footnote{
Models usally interpret the $\lam$-calculus in a categorical way. This aspect is only marginal in this paper. 
} $\denot{t}$, meaning that $\denot{t}$ is invariant under $\beta$-conversion (\ie $t_1\beq t_2$ implies $\denot{t_1}=\denot{t_2}$). 
For instance, the typing judgments of  Gardner-de Carvalho's system $\scrRo$ (that we will shortly discuss in \S\;\ref{ss:system-R-LICS18}) correspond to the points of the relational model~\cite{BucciarelliEM07} in that, for any term $t$, $\set{\tri_{\scrRo} \juGtt\;|\; \Gam,\tau}=\denot{t}_{\mathtt{rel}}$, where the left-hand side corresponds to the set of derivable judgments of system $\scrRo$ typing $t$ and the right-hand side, the denotation of $t$ in the relational model. Thus, infinite types, which are the concern of this article, are a mean to study the infinitary extension of the relational model. 

The relational model only gives a (non-empty) denotation to HN terms. This reflects the fact that non-HN (equivalently, non-\textbf{solvable}) terms have an infinitary behavior 
\wrt head reduction: it is thus natural to seek whether such terms have an infinitary semantics, since infinitary models bring information on asymptotic aspects of terms \eg in the recent work of Grellois-Melli\`es~\cite{GrelloisM15,GrelloisM15a}. The first main contribution of this article is to prove that every term has a non-empty interpretation in the infinitary relational model. 
In the perspective of Illative Systems, introduced by Curry, to find more and more ``meanings'' to $\lam$-terms (see \eg~\cite{CoppoD80} or 5.4. in \cite{CardoneH09}), one may thus consider that for the infinitary relational model, every $\lam$-term is meaningful.


An interesting aspects of models is that they allow us to statically discriminate between two terms, meaning that if $\denot{t_1}\neq\denot{t_2}$ then $t_1\not\beq t_2$ \ie two terms that do not have the same denotation do not represent two different states of a same program. For instance, the i.t.s.  that are  able to assign a type of order $n$ to any (\eg HN) term of order $n$ (but not to a term of order $n+1$) can be regarded as \textbf{order-discriminating} for HN terms. This holds for system $\scrRo$. The infinitary extension of system $\scrRo$, that we denote $\scrR$, is order-discriminating for all $\lam$-terms (not just the HN ones). This result, which extends a feature of system $\scrRo$ concerning HN terms to the whole $\lam$-calculus is the second main contribution of the paper (Theorem\;\ref{th:order}).

\subsection{Stable positions  (and the Difficulty of Infinitary Typing)}
\label{ss:intro-difficulty-comp-unsound}
Now, a few words on the difficulty of describing the infinitary typing that can be assigned to a given term. Let us have another look at Fig.\;\ref{fig:typing-hnf-its}. The fact that we can describe the expressivity of a given i.t.s. usually comes from the fact that it is very easy to describe the typings of ``partial'' normal forms (\eg HNF or $\beta$-NF). Intuitively, the nodes of a partial normal forms corresponding to normalized parts 
 cannot be affected by reduction \eg the nodes labelled with $\lx_i$, $\arob$ and $x$ in Fig.\;\ref{fig:typing-hnf-its}  (the ``spine'' of the HNF). Such nodes can be called \textbf{stable}. In contrast, some terms do not ever give rise to stable position and they are known as \textbf{mute terms} in the literature~\cite{BerarducciI93}. Formally, $t$ is mute iff  $t$ is mute if any reduct of $t$ may be reduced to a redex. Thus, mute terms are \textit{persisting (latent)} redexes and they are regarded as the \textit{``most undefined $\lam$-terms"} [ibid].
Note that $\Omega$ is mute: the infinite reduction sequence $\Om \bred \Om \bred \ldots$ only seems to be constant, but, actually, a new redex is \textit{created} at each reduction step.

Thus, there is no clear way to capture the order of any typable term, since the example of $\Om$ shows that the case of totally unstabilizable term must be handled when considering infinite types. Note that $\Om$ is just an example of mute terms, that occurs to satisfy the nice fixpoint equation $\Om \bred \Om$ ($\Om$ is indeed an instance of Curry fixpoint) and has a simple parsing tree. This partially explains why $\Om$ was easily typable. In general, this is not the case of mute terms.

\subsection{Infinitary Typing and Klop's Problem}
\label{ss:intro-infty-relev-typing}


Let us say a few words about the questions raised by infinitary typing in the non-idempotent 
type framework \ie by the interpretation of terms in the infinitary relational model.
\begin{itemize}
\item One of the fundamental interests with non-idempotent intersection is that, in this setting, a type is a \textbf{resource} that cannot be duplicated or merged/contracted and is possibly consumed under reduction.
\item Moreover, non-idempotent i.t.s. are often \textbf{relevant}, meaning that weakening is not allowed.
\end{itemize}
An i.t.s. that forbids duplication and weakening can be qualified as \textbf{linear}, which is the case of system $\scrR$ that we hinted at in \S\;\ref{ss:intro-search-inf-denot}. As we will see (see \S\;\ref{ss:system-R-LICS18}), relevance disables the argument proving that every term is typable with $\rho$, the non-idempotent counterpart of the type $R$ considered above. However, while trying to characterize a form of infinitary weak normalization, we noticed in~\cite{VialLICS17} that $\Om$ is also typable $\scrR$. We recovered soundness by defining a validity criterion, discarding degenerate typing derivations, which was possible by introducing a \textbf{rigid} variant of system $\scrR$, namely system $\ttS$. System $\ttS$ has many nice features \eg tracking (see \S\;\ref{ss:system-S-LICS18}).

Still, these observations rise the question of the set of typable terms (without validity criterion) in the coinductive relevant and non-idempotent framework: is every term $\scrR$-typable? Or is there a term $t$ that is not $\scrR$-typable? Observe that such a term $t$ would not be linearizable, even in an infinite way (since systems $\scrR$ and $\ttS$ are \textit{linear}).  The existence of non-linearizable terms would be very suprising and must be therefore investigated. Since our main theorem\;\ref{th:comp-unsoundness-R-S} states that system $\scrR$ actually types every term, we actually prove that they do not exist, as expected.

Note again that the method described at the end of \S\;\ref{ss:intro-sem-aspects-infty-types} does not work for non-normal term: naively, when $x\,u$ occurs in $t$, we would like to assign to $x$ a type of the form $A\rew B$, where $A$ is the type of $u$, and proceed by induction. However, $x$ may be substituted in the course of a reduction sequence, and so, typing constraints on $x$ are not easily readable. Thus, in the finite or productive case, in the purpose of proving that the terms of a given set (we considered the set of HN terms earlier) are typable, we escape this problem by typing normal forms (\eg HNF) and then proceeding by expansion. But, as noted above, in the coinductive case, no form of normalizability is ensured by typability.

To sum up, due to full resource-awareness (including relevance), typing in the coinductive systems $\scrR$ and $\ttS$ is intrisically non-trivial. 
But for the same reason (full resource-awareness), linear intersection type systems consitute the good framework to study the expressive power of infinite types and to capture the order of every $\lam$-term. Thus, besides our first goal\ldots

\begin{goal}
\label{goal:order}
Capturing the order of every $\lam$-term with infinite types.
\end{goal}

\ldots we have now a second one, narrowly related to the first:

\begin{goal}
\label{goal:R-typability}
Proving that every term is $\scrR$-typable.
\end{goal}





\subsection{A Technical Contribution}

We have still not addressed the way we can study unproductive/mute terms, despite the fact that all the known techniques of intersection type theory fail: we propose a solution to overcome this difficulty, inspired by (a simplified form of) first order model theory, that we mix with techniques specific to the $\lam$-calculus. This is our main technical innovation, since it allows us to study \textit{unproductive} reduction. Thus, beyond the relational model, this work proposes the first use of first order model theory to study an infinitary extension of a finitary model of the $\lam$-calculus and to generalize properties coming from the finite model to every $\lam$-term (\eg capturing their order).
 The proof that every term is $\scrR$-typable has three main stages: (1) reducing the problem to a set of \textit{stability relations} (\S\;\ref{s:candidate-bisupp}) (2) describing the possible \textit{interactions} between these relations (3) describing a procedure of partial (but finite) normalization (\S\;\ref{s:normalizing-chains}). The same ingredients allow us to capture the order.
 The italicized words, as well as the tools of the proof and the reasons why they arise, are gradually explained in the paper but we refer to \S\;\ref{ss:candidate-supp-types}, \ref{ss:toward-candidate-bisupp} and the introduction of \S\;\ref{s:nihilating-chains} for some high-level input.

\newcommand{\inferRaxx}[1]{\infer{\phd}{\ju{x:\mult{#1}}{x:#1}}}
\newcommand{\inferRaxxsh}[1]{\infer{\phd}{\ju{x\!:\!\mult{#1}\!}{x\!:\!#1}}}

\section{Infinitary Relevant and Non-Idempotent Intersection}
\label{s:non-idem-its-LICS18}

\subsection{System $\scrR$}
\label{ss:system-R-LICS18}

We now define more formally system $\scrR$, the \textit{coinductive} version of the finite system $\scrRo$, independently introduced by Gardner and de Carvalho~\cite{Gardner,Carvalho07}. See \cite{BucciarelliKV17} for a general presentation of $\scrRo$. \paper{System $\scrR$ is of good help to understand relevant intersection 
  but, as we shall see in  \S\;\ref{ss:system-S-LICS18} et \ref{ss:candidate-supp-types}, it 
}\techrep{System $\scrR$ is of two uses here:
\begin{itemize}
\item Understanding relevant intersection.
\item Informally presenting some key tools of our main proofs.
\end{itemize}
However, as we shall see in \S\;\ref{ss:system-S-LICS18} and \ref{ss:candidate-supp-types}, system $\scrR$} is unfit to express the techniques yielding Theorems\;\ref{th:comp-unsoundness-R-S} and \;\ref{th:order}, and we refer to it only for heuristic purposes.

The set of $\scrR$-types is coinductively defined by.
\[\sigma,\,\tau :: = \tv\in \TypeV~ |~ \msigi \rew \tau \]
We call $\IM:=\msigi$ a \textbf{multiset type}. The multiset types represent intersection in system $\scrRo$ and the intersection operator $\inter$ is the multiset-theoretic sum: $\inter_\iI \IMi=+_\iI \IMi$ (\ie $\inter_{\iI} \mult{\sig^i_j}_\jJi  := +_{\iI} \mult{\sig^i_j}_{\jJi}$). We assume $I$ to be countable, the empty multiset type is denoted by $\emul$ and $\mult{\sig}_\om$ is the multiset type in which $\sig$ occurs infinitely many times.

A $\scrR$-\textbf{context} (metavariables $\Gam,\Del$) is a \textit{total} function from $\TermV$ to the set of multiset types. The \textbf{domain of $\Gam$} is given by $\set{x\,|\,\Gam(x)\neq \emul}$. The intersection of contexts $+_{\iI} \Gami$ is defined point-wise.
We may write $\Gam;\Del$ instead of $\Gam + \Del$ when $\dom{\Gam}\cap \dom{\Del}=\eset$. Given a multiset type $\msigi$, we write $x:\msigi$ for the context $\Gam$ s.t. $\Gam(x)=\msigi$ and $\Gam(y)=\emul$ for all $y\neq x$.
 A $\scrR$-\textbf{judgment} is a triple $\ju{\Gam}{t:\sig}$ where $\Gam$ is a $\scrR$-context, $t$ a term and $\sig$ a $\scrR$-type.

The set of $\scrR$-derivations is defined \textit{inductively} by\techrep{ the following rules}: 
  $$
  \begin{array}[t]{c}\infer[\ax]{ \phd}{ \ju{x:\mult{\tau} }{x:\tau}}  \hspace{1.4cm} \infer[\abs]{ \ju{ \Gam;x:\msigi}{t:\tau}}{\ju{\Gam}{\lx.t:\msigi\rew \tau}} \\[4ex]
    \infer[\app]{\ju{\Gam}{t:\msigi\rew \tau}\hspace{0.3cm} (\ju{\Deli }{u:\sigi})_{\iI} }{\ju{\Gam+(+_{\iI} \Deli )}{t\,u:\tau}}
    \end{array}
$$

As announced in \S\;\ref{ss:intro-infty-relev-typing}, system $\scrR$ is not only non-idempotent, but also \textit{relevant}. 
For instance, the $K$-term $\lx.y$ can only be assigned types of the form $\emul \rew \tau$. Indeed, these are below the only possible typings of $\lx.y$.
$$
  \infer[\abs]{
\infer[\ax]{\phd}{\juaxtt}
  }{\ju{x:\mult{\tau}}{\ly.x:\emul\rew \tau}}
$$
This comes from the fact that $y$ does not occur in $x$, and thus, by relevance, the constructor $\ly$ cannot produce a type on the left-hand side of the arrow type.

Fig.\;\ref{fig:infinite-simple-types-unsound} can adapted\techrep{, yielding Fig.\;\ref{fig:infinite-R-om-typ}, } and we can type $\Om$ with of $\tau$ for all $\scrR$-types $\tau$, by just defining $\phi_\tau$ by $\phi_\tau =\mult{\phi_\tau}_\om \rew \tau$. However, defining $\rho$ by $\rho=\mult{\rho}_\om \rew \rho$ does not allow to type every term with $\rho$ in system $\scrR$. We explain why now.
\techrep{
\begin{figure*}
\[\infer{
\infer{
\infer{
\infer{\phd}{\ju{x:\mult{\rhoo}\!}{x:\rhoo}}\\
(\infer{\phd}{\ju{x:\mult{\rhoo}}{x:\rhoo}})_{\om}}{
  \ju{x:\mrhooom}{x\,x:\tv}}}{
\ju{}{\lx.x\,x:\mrhooom\rew \tv}
}\\
(\infer{
\infer{
\infer{\phd}{\ju{x:\mult{\rhoo}}{x:\rhoo}}\\
(\infer{\phd}{\ju{x:\mult{\rhoo}}{x:\rhoo}})_{\om}}{
  \ju{x:\mrhooom}{x\,x:\tv}}}{
\ju{}{\lx.x\,x:\rhoo}
})_\om
}{
  \ju{}{\Om:\tv}}\]
  \caption{Typing $\Om$ in $\scrR$.}
  \label{fig:infinite-R-om-typ}
  \end{figure*}
}

Relevance\techrep{ of system $\scrR$} can be disabled, by replacing $\ax$ par $\axw$:
\[
\infer[\axw]{ i_0\in I}{ \ju{\Gam;x:\msigi }{x:\sig_{i_0}}}\]
We call $\scrRw$ the type system thus obtained. Then, the proof on p.\;\pageref{inpage:comp-unsound-inf-sts} can be adapted to $\scrRw$, by considering only contexts $\Gam,\Gamt,\Gamu$ assigning $\mult{\rho}_\om$ to all the variables in their domains:
\[\infer[\ax]{\phd}{\ju{x:\mult{\rho}_\om}{x:\rho}} \hspace{1.6cm} \infer[\abs]{\ju{\Gam;x:\mult{\rho}_\om}{t:\rho}}{\ju{\Gam}{\lx.t:\mult{\rho}_\om \rewsh \rho\ \ (=\rho)} }\]
\[\infer[\app]{\ju{\Gamt}{t:\rho\  (\mathord{=}\mult{\rho}_\om \rewsh \rho)} \hspace{0.2cm} (\ju{\Gamu}{u:\rho})_{\om}}{\ju{\Gam  }{t\,u:\rho}}\hspace{0.4cm} {\small
\begin{array}{l}
\text{with}\\ \Gam\mathord{=}\Gamt+\om \times \Gamu \end{array}}\]
Thus, every term is $\scrRw$-typable. But note that, by relevance, this proof by induction on the structure of $t$ fails for $\scrR$. For instance, if $x$ is not in $t$, $\ju{\Gam}{t:\rho}$ yields $\ju{\Gam}{\lx.t:\emul\rew \rho\ (\neq \rho!)}$ and since the empty multiset type may occur in unpredictable places in a term, finding a $\scrR$-typing of any term $t$ is non-trivial, as announced in \S\;\ref{ss:intro-difficulty-comp-unsound} and \ref{ss:intro-infty-relev-typing}. In some sense $\scrR$-typability is about capturing the way relevance constrains emptiness to occur! But since a variable $x$ can be substituted under reduction as observed in \S\;\ref{ss:intro-infty-relev-typing},  $\emul$ may occur in unpredictable places.

\subsection{Towards Tracking and Sequential Intersection}
\label{ss:system-S-LICS18}

Unfortunately, resource-awareness of system $\scrR$ is not enough to process the proof techniques to be developped here:  we also need \textbf{tracking}. Let us just give an example to show what the impossibility of tracking means:
\[
\infer{
\inferRaxxsh{\mult{\sig,\sig}\rewsh \tau}\hspace{0.2cm}
\inferRaxxsh{\sig} \hspace{0.1cm} \inferRaxxsh{\sig}}{
\ju{x:\mult{\mult{\sig,\sig}\rew \sig,\sig,\sig}}{x\,x:\tau}
}
\]
In the derivation above, in the context $x:\mult{\msig\rew \sig,\red{\sig},\red{\sig}}$ of the conclusion, one cannot know  which particular axiom rule, each red occurrence of $\sig$ comes from: there is no possible notion of \textbf{pointer} with multiset intersections. But as it will turn out in \S\ref{ss:candidate-supp-types}, this is one thing that we absolutely need to capture the key notion of \textbf{support candidate}.

Tracking can be retrieved while keeping most of system $\scrRo$'s nice features (\eg syntax-direction) by considering system $\ttS$, that we introduced in~\cite{VialLICS17}. System $\ttS$ uses \textbf{sequence types} instead of \textbf{multiset types} to represent intersection. For instance, instead of having a cardinal 3 intersection $\mult{\tv,\tv',\tv}$, system $\ttS$ considera cardinal 3 sequence $(2\cdot \tv,5\cdot \tv',8\cdot \tv)$. Sequences come along with a disjoint union operator \eg $(2\cdot \tv,5\cdot \tv',8\cdot \tv)=(2\cdot \tv,5\cdot \tv') \uplus (8\cdot \tv)$: in this equality, the occurrence of $\tv$ in the left-hand side annotated with 2 unambiguously comes from that which is also annotated with 2 in the right-hand side. We call these annotations \textbf{tracks}.

In contrast, $\mult{\tv,\tv',\tv}=\mult{\tv,\tv'}+\mtv$, but there is no way to unambiguously associate to an occurrence of $\tv$ in the left-hand side the one of $\mult{\tv,\tv'}$ or the one of $\mtv$ in the right-hand side. 

Formally, the set of $\ttS$-types is defined coinductively by:
\newcommand{\skSk}{(k\cdot S_k)_{\kK}}
\begin{center}$ 
\begin{array}{cll} 
  T,\,S_k & ::= &  \tv ~ \|~ F \rightarrow T\\
  F & ::=& \skSk\ (K\subeq \Nmzo)
\end{array}$\\[-2ex] 
\end{center}
The empty sequence type is denoted $\est$ and we often write $\sSk$ instead of $\skSk$.
The set of top-level tracks of a sequence type is called its set of \textbf{roots} and we write \eg $\Rt(F)=\set{2,5,8}$ when $F=(2\cdot \tv,5\cdot \tv',8\cdot \tv)$. Note that the disjoint union operator can lead to \textbf{track conflict} \eg, if $F_1=(2\cdot \tv,3\cdot \tv')$ and  $F_2=(3\cdot \tv',8\cdot \tv)$, the union $F_1\uplus F_2$ is not defined,  since $\Rt(F_1)\cap \Rt(F_2)=\set{3}\neq \eset$. 

A $\ttS$-context $C$ (or $D$) is a total function from $\TermV$ to the set of $\ttS$ types. The operator $\uplus$ is extended point-wise.
A $\ttS$-judgment is a triple $\juCtt$, where $C$, $t$ and $T$ are respectively a $\ttS$-context, a term and $T$ a $\ttS$-type. 
A \textbf{sequence judgment} is a sequence of judgments  $(k\cdot (C_k\vdash t:T_k))_{k\in K}$ with $K\subeq \Nmzo$, often just written $(\ju{C_j}{t:T})_{\kK}$. For instance, if $5 \in K$, then the judgment on track 5 is $\ju{C_5}{t:S_5}$.

The set of $\ttS$-derivations is defined inductively by:
\[\infer{\phd}{\juxkT} \ax \hspace{1.4cm}
  \infer{\ju{C;x:\sSk}{t:T}}{\ju{C}{\lx.t:\sSk \rew T}}\abs\]
  \[\infer{\ju{C}{t:\sSk\rew T} \hspace{0.7cm} (\ju{D_k}{u:S_k})_{\kK}  }
        {\ju{C\uplus(\uplus_{\kK} D_k)}{t\,u:T}}\app\]
The $\app$-rule can be applied only if there is no track conflict in the context $C\uplus(\uplus_{\kK} D_k)$. In an $\ax$-rule concluding with $\juxkT$, the track $k$ is called the \textbf{axiom track} of this axiom rule. We refer to \S\;III and IV of \cite{VialLICS17} for additional examples and figures for all what concerns the basics of system $\ttS$, which are presented in this \S\;\ref{s:non-idem-its-LICS18}. As expected:

\begin{proposition}
  \label{s:S-sr-se}
System $\ttS$  enjoys  subject reduction and expansion (as system $\scrR$ does).
  \end{proposition}

Note that, if we forget about track annotations, a sequence becomes just a multiset and a $\ttS$-derivation collapses on a $\scrR$-derivation, so that a $\ttS$-typable term is also $\scrR$-typable. We may thus replace Goal\;\ref{goal:R-typability} by Goal\;\ref{goal:S-typability}.

\begin{goal}
\label{goal:S-typability}
Proving that every term is $\ttS$-typable.
\end{goal}

We reduce the problem (\ie proving that a term $t$ is typable in system $\ttS$) into a first order theory, that we call $\calT_t$.
We actually prove that $\calT_t$ indeed captures the problem, by means of a proposition that can be interpreted as a (simplified) completeness theorem (see Corollary\;\ref{corol:typability-and-thepsi}): we show that if $\calT_t$ is coherent, then we are able. Then we prove that $\calT_t$ is coherent for all term $t$. 
Go to the introduction of \S\;\ref{ss:toward-candidate-bisupp} and \ref{s:nihilating-chains} to have a closer descriptions of the proof of the coherence of $\calT_t$ and of its main stages.\ignore{\\

  \noindent \textbf{Choice function for axiom tracks}
  Observe now that, to define a $\ttS$-derivation typing a term $t$ (thus fulfilling Goal\;\ref{goal:S-typability}), one must  one choose an axiom track in every axiom rule so that no conflict arise. In this short article, let us just say that we can escape this problem by resorting to an \textit{arbitrary} injection from $\bbN^*$ to $\Nmzo$, that chooses axiom tracks for us: we say that 
a $\ttS$-derivation $P$ is a $\code{\cdot}$-derivation if $P(a)=(\ju{x:(k\cdot T)}{x:T})$ (\ie $a$ is the position of an $\ax$-rule typing $x$ in $P$), then $k=\code{a}$.  If $t$ is $\code{\cdot}$-typable, then $t$ is in particular $\ttS$-typable, 
so that we now replace Goal\;\ref{goal:S-typability} with Goal\;\ref{goal:code-S-typability}:

\begin{goal}
\label{goal:code-S-typability}
Given $\code{\cdot}:\bbN^*\rew \Nmzo$, proving that every term is $\code{\cdot}$-typable.
\end{goal}

The function $\code{\cdot}$ is assumed to be an injection to avoid us bothering  about track conflict any further while achieving Goal\;\ref{goal:code-S-typability}.  It is also \wrt the funciton $\code{\cdot}$ that we will capture the occurrences of emptiness (see\;\S\;\ref{ss:toward-candidate-bisupp}). \techrep{\pierre{Let us present next the key notion of \textit{bisupport candidate}, that will help us achieve Goal\;\ref{goal:code-S-typability}.}}}

\subsection{Parsing, Pointing}
\label{ss:parsing-pointing}

In this technical section, we explain how we may point inside a $\ttS$-type or a $\ttS$-derivation, thanks to tracking. We define the support of a $\ttS$-type and a $\ttS$-derivation, and also the key notions of biposition and bisupports. Let $\bbN^*$ denote the set of the finite words on $\bbN^*$, the operator $\cdot$ denotes concatenation, $\epsi$ the empty word and $\leqs$ the prefix order \eg $2\cdot 1\cdot 3\cdot 7\in \bbN^*$, $2\cdot 1\leqs 2\cdot 1\cdot 3\cdot 7$. Moreover, the \textbf{collapse} $\ovl{k}$ of a track $k$ is defined by $\ovl{k}=\min(k,\,2)$. This notation is extended letter-wise on $\bbN^*$  \eg $\ovl{0 \ct 5 \ct 1 \ct 3\ct 2}=0\ct 2 \ct 1\ct 2\ct 2$. The \textbf{support of term} is defined by induction as expected: $\supp{x}=\set{\epsi}$, $\supp{\lx.t}=\set{\epsi}\cup 0\cdot \supp{t}$ and $\supp{t\,u}=\set{\epsi}\cup 1\cdot \supp{t} \cup 2\cdot \supp{u}$. If $a\in \bbN^*$ and $\ovla\in \supp{t}$, we denote by $\tra$ the subterm of $t$ rooted at position $\ovla$ whereas $t(a)$ is the constructor ($\arob$, $x$ or $\lx$) of $t$ at position $a$ \eg $t\rstr{0}=y\,x$ and $t(0\cdot 1)=y$ with $t=\lx.y\,x$.
\\

Let $\Sex=(2\cdot \tv,7\cdot \tv')\rew \tv\secu$. 
To gain space, we write $\ju{k}{x:T}$ (with $k\geqs 2$, $x\in \TermV$, $T$ $\ttS$-type) instead of $\ju{x:(k\cdot T)}{x:T}$ in $\ax$-rules.
We also indicate the track of argument derivations between brackets \eg $x:(3\cdot \tv)\trck{5}$ means that the argument judgment $\ju{x:(3\cdot \tv)}{x:\tv}$ is on track 5):
\begin{center}
$
\infer[\ax]{
\infer[\app]{
\infer[\ax]{\phd}{\ju{3}{y:\Sex}}\hspace{0.6cm}
\infer[\ax]{\phd}{\ju{3}{x: \tv}\trck{5}}\hspace{0.3cm}
\infer[\ax]{\phd}{\ju{9}{x:\tv'}\trck{6}}}{
\ju{x:(3\cdot \tv,9\cdot \tv'),\; y:(3\cdot \Sex)}{y\,x:\tv\secu}
  }}{\ju{y:(3\cdot \Sex)}{\lx.y\,x: (3\cdot \tv,9\cdot \tv')\rew \tv\secu}}
$
\end{center}
The \textbf{support of a type} (resp. \textbf{a sequence type}), which is a tree of $\bbN^*$ (resp. a forest),  is defined by mutual coinduction: $\supp{\tv}=\set{\epsi},~ \supp{F\rew T}=\set{\epsi}\cup \supp{F}\cup 1\cdot \supp{T}$ and $\supp{(T_k)_{k\in K}}= \cup_{k\in K} k\cdot \supp{T_k}$ \eg $\supp{\Sex}=\set{\epsi,1,2,7}$. We can define the \textbf{support} of a derivation $P\rhd C\vdash t:T$: $\supp{P}=\epsi$ if $P$ is an axiom rule, $\supp{P}=\{\epsi\}\cup 0\cdot \supp{P_0}$ if $t=\lambda x.t_0$ and $P_0$ is the subderivation typing $t_0$, $\supp{P}=\{\epsi\}\cup 1\cdot \supp{P_1} \cup_{\kK} k\cdot \supp{P_k}$ if $t=t_1\,t_2,~ P_1$ is the left subderivation typing $t_1$ and $P_k$ the subderivation typing $t_2$ on track $k$. The $P_k$ ($k\in K$) are called \textbf{argument derivations}. For instance, $\supp{\Pex}=\set{\epsi,0,0\ct 1,0\ct 5,0\ct 6}$, $P(0\cdot 1)=\ju{y:(3\ct \Sex)}{y:\Sex}$ and $P(0\ct 6)=\ju{x:(9\ct \tv')}{x:\tv'}$.

\noindent \textbf{Choice function for axiom tracks}
Meanwhile, note that, to define a $\ttS$-derivation typing a term $t$ (thus fulfilling Goal\;\ref{goal:S-typability}), one must  one choose an axiom track in every axiom rule so that no conflict arise. In this short article, let us just say that we can escape this problem by resorting to an \textit{arbitrary} injection from $\bbN^*$ to $\Nmzo$, that chooses axiom tracks for us: we say that 
a $\ttS$-derivation $P$ is a $\code{\cdot}$-derivation if $P(a)=(\ju{x:(k\cdot T)}{x:T})$ (\ie $a$ is the position of an $\ax$-rule typing $x$ in $P$), then $k=\code{a}$.  If $t$ is $\code{\cdot}$-typable, then $t$ is in particular $\ttS$-typable, 
so that we now replace Goal\;\ref{goal:S-typability} with Goal\;\ref{goal:code-S-typability}:

\begin{goal}
\label{goal:code-S-typability}
Given $\code{\cdot}:\bbN^*\rew \Nmzo$, proving that every term is $\code{\cdot}$-typable.
\end{goal}

The function $\code{\cdot}$ is assumed to be an injection to avoid us bothering  about track conflict any further while achieving Goal\;\ref{goal:code-S-typability}.  It is also \wrt the funciton $\code{\cdot}$ that we will capture the occurrences of emptiness (see\;\S\;\ref{ss:toward-candidate-bisupp}). \techrep{\pierre{Let us present next the key notion of \textit{bisupport candidate}, that will help us achieve Goal\;\ref{goal:code-S-typability}.}}\\ 

We explain now how to point \textit{inside} types \textit{nested} in $\ttS$-derivation, or to axioms typing a given variable, and formalize the associated pointers. 
If $P$ is a $\ttS$-derivation and $a\in \supp{P}$, then the judgment at position $a$ is denoted $\ju{\ttC^P(a)}{\tra:\ttT^P(a)}$ \eg $\ttC^{\Pex}(0\cdot 6)=x:(9\cdot \tv')$ and $\ttT^{\Pex}(0\cdot 6)=\tv'$. Let $P$ be a $\ttS$-derivation. A \textbf{right position} is a pair of the form $(a,c)$, where $a\in \supp{P}$ and $c\in \supp{\ttT^P(a)}$, we write $\bisupp{P}$ for the \textbf{(right) bisupport} of $P$ \ie its set of (right) bipositions. If $(a,c)\in \bisupp{P}$, then $P(a,c)$ denotes $\ttT^P(a,c)$ \eg $\Pex(0\ct 6,\epsi)=\tv'$ and $\Pex(0\cdot 1,\epsi)=\tv$, $\Pex(0\cdot 1,1)=\tv\secu$ and $\Pex(\epsi,9)=\tv'$. Note that, contrary to \cite{VialLICS17}, we only consider right bipositions. For this article, we think a biposition as type symbol ($\tv \in \TypeV$ or $\rew$) nested in a given $\ttS$-derivation $P$ and we often use this heuristic identification implicitly.

\techrep{Some other notations are useful to handle $\ttS$-derivations: assume}\paper{Assume} that $P$ types $t$. We set $A=\supp{P}$ and $B=\bisupp{P}$.
If $x\in \TermV,~a\in A$, we set $\AxPa(x)=\set{a_0 \in A\,|\,a\leqslant a_0,\,t(a)=x,\nexists a_0',\,a\leqslant a_0'\leqslant a_0,\,t(a'_0)=\lx }$ (occurrences of $x$ in $P$ above $a$, that are not bound \wrt $a$) \eg $\Ax^{\Pex}_\epsi(x)=\eset$ ($x$ is bound at the root), but $\Ax^{\Pex}_0(x)=\set{0\cdot 5,0\cdot 6}$ ($x$ is not bound at position 1). Technically, this notation is crucial to harness relevance 
(see\techrep{ the definition of} polar inversion, \S\;\ref{ss:asc-pi-c-u}) but the important thing to remember is that, thanks to tracking, in system $\ttS$, one can unambiguously designate the axiom rules typing the variable of a $\lx$, \eg in a redex.
\techrep{ 
If $P(a)=\ju{x:k\cdot T}{x:T}$, then we set $\trP{a}=k$ \ie $\trP{a}$ is axiom track  at  position $a$ \eg $\tttr^{\Pex}(0\ct 6)=9$.
Since $t$ is finite, we have $C(a)(x)=\uplus_{a_0\in \AxPa(x)} (\trP{a_0}\ct T(a_0))$. This indicates that in a $\ttS$-derivation, contexts and types can be computed from the support $\supp{P}$ and the types created in axiom rules.}


\subsection{Typing some Notable Terms in System $\scrR$}
\label{ss:typing-examples}

We now use system $\scrR$ to type a few terms  satisfying fixpoint equations. Some of them are not head normalizing.  Let $\Delf=\lx.f(x\,x),\,\cu=\lambda f.\Delf  \Delf$ ($\cu$ is \textit{Curry fixpoint combinator}). Moreover, if $I=\lx.x$ and $K=\lxy.x$, then  $\cu\,I\rew \Omega$ (satisfying $\Om \bred \Om$), $\cu\,f\rew \cuf:=\Delf\,\Delf$ (satisfying $\cuf\bred f(\cuf)$) and $\cu\,K\bred \cul:= (\lx.\ly.xx)\lx.\ly.xx$ (satisfying $\cul\bred \ly.\cul$). 
\techrep{
We recall from the Introduction\techrep{ the definition of the \textit{order} of a $\lam$-term}:

\begin{definition}
  \label{def:order-comp-unsound}
Let $t$ be a $\lam$-term. The \textbf{order} of $t$ is $\sup \set{n\in \bbN\,|\, \exists x_1,\ldots,x_n,u\ \text{s.t.}\ t\bred^* \lx_1\ldots \lx_n.u}$.
\end{definition}
}

Iterating reduction on $\cuf$ and $\cul$ infinitely many times, we see that $\cuf$ (resp. $\cul$) \textit{strongly converges} to the infinitary term  $\fom:=f(f(...))$ (resp. $\ly.\ly....$) in the sense of \cite{KennawayKSV97,Czajka14}. Thus, $\Om$ and $\cuf$ are both \textbf{zero terms} (terms of order 0)  and $\cul$ a term of infinite order. The term $\Om$ is actually a \textit{mute term} (see \S\;\ref{ss:intro-difficulty-comp-unsound}) and $\cuf$ is a term whose B\"ohm tree~\cite{Barendregt85} $\fom$ does not contain $\bot$.


Because of rule $\abs$ and subject reduction (that is satisfied in $\scrR$), a term of order $n$ may only be typed with a type of order $\geqslant n$, as explained in \S\;\ref{ss:intro-sem-aspects-infty-types}. 
However, some $\scrR$-derivations can capture more precisely the order of terms.
For all $\scrR$-type $\tau$, we define coinductively $\phi_\tau$ by $\phi_\tau=\mult{\phi_\tau}_\om\rew \tau$. For instance, we consider the following typing of $\cu$ (omitting left-hand sides of $\ax$-rules):\\[-3.4ex]
\[\Pi_{\Delf}= \infer[\abs]{\infer[\app]{ \infer[\ax]{\phd}{f:\mtau \rew \tau} \hspace{0.3cm}
    \infer[\app]{  \infer[\ax]{\phd}{x:\phi_\tau}
        \hspace{0.3cm} \big( \infer[\ax]{\phd }{x:\phi_\tau}\big)_{\om} }{ \ju{x:\mult{\phi_\tau}_{\om}}{x\,x:\tau}}}   {
      \ju{f:\mult{\mtau\rew \tau};x:\mult{\phi_\tau}_{\om}}{f(x\,x):\tau}}}{
  \ju{f:\mult{\mtau\rew \tau}}{\Delf:\phi_\tau\ (=\mult{\phi_\tau}_\om\rew \tau)}}\]
\[    \Pi_{\cu}=
    \infer[\abs]{
\infer[\app]{\Pi_{\Delf} \hspace{0.9cm} ( \Pi_{\Delf})_\om}{\ju{f:\mult{\mtau\rew \tau}_\om}{\Delf\,\Delf:\tau} }}{
      \ju{}{\cu:\mult{\mtau\rew \tau}_\om\rew \tau}
    }
  \]
Thus, in system $\scrR$, $\cu$ is typable with $\mult{\mtau\rew \tau}_\om\rew \tau$ for \textit{any} type $\tau$. Notice that we also have derived $\ju{f:\mult{\mtau\rew \tau}_\om}{\cuf:\tau}$ for any $\scrR$-type $\tau$.

 Using suitable instances or variants of $\Pi_{\cu}$, we can build $\Pi_{\Om}\tri \ju{}{\Om:\tau}$ (for any $\tau$) and $\Pi_{\lam}\tri \ju{}{\cul:\emul\rew \emul \rew\ldots}$ 
 By instantiating $\tau$ with a type variable $o$, we get $ \ju{}{\Om:\tv}$ and $\ju{}{\cuf:\tv}$. Thus, the zero terms $\Om$ and $\cuf$ are typed\footnote{
Note that $\cuf \rew f(\cuf)$ ($\cuf$ is HN) and $\cuf$ is typable with $\tv$ in the finite system $\scrRo$.} with types of order 0  whereas $\cul$ (whose order is infinite) is typed with a type of infinite order, as it was constrained to be.

We will generalize this result (not only for terms built from fixpoint combinators like $\Om$ or $\lx.\Om$) and show that, for all \textit{pure} terms $t$ of order $n$, there is a $\scrR$-derivation\techrep{ (or a $\scrDw$-derivation)} typing $t$ with a type of order $n$ (Theorem~\ref{th:order}).

\newcommand{\ascarrowleft}[2]{
  \trans{#1}{#2}{
    \draw [->,dotted,>=stealth] (0,0) --++ (-0.35,0.5) ;
  }
}

\newcommand{\ascarrowup}[2]{
  \trans{#1}{#2}{
    \draw [->,dotted,>=stealth] (0,0) --++ (0,0.5) ;
}}

\newcommand{\inBis}{\mathtt{inBis}}
\newcommand{\orange}[1]{\textcolor{orange}{#1}}

\section{Bisupport Candidates}
\label{s:candidate-bisupp}

%
%

  In this section, we characterize, for a given term $t$, the \textit{bisupport candidate} \ie the (potential) forms of a derivation typing $t$. {By ``form'', we intuitively mean a set of unlabelled positions (that must be stable under some suitable relations). We make explicit that idea by studying first the possible forms of a $\ttS$-type in \S\;\ref{ss:candidate-supp-types}. The notion of unlabelled position has a meaning only because tracks of $\ttS$ allow us to define suitable pointers. This would be impossible in system $\scrR$.

\subsection{A Toy Example: Support Candidates for Types}
\label{ss:candidate-supp-types}

In this section, we explain how the notion of ``form'' of a support can be formalized by giving a characterization of the supports of $\ttS$-types in terms of stability conditions.

The definition of a particular $\ttS$-type $T$ can be understood as a two-step process: first, we choose the support $C:=\supp{T}$, next, we choose the type labels $T(c)$ (in  the signature $\TypeV\cup \set{\rew}$) given to the positions $c\in C$. However, not all the subsets $C\subseteq \bbN^*$ are fit to be the support of a type, and not all the possible decorations of a suitable set $C$ yield a correct type. 

For instance, let us consider the two sets of positions $C_1$ and $C_2$ below. Do they define the supports of some types $T_1$ and $T_2$?
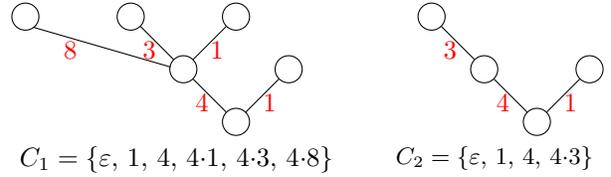
\begin{center}
  \begin{tikzpicture}
  \draw (0,0) circle (0.18) ;

  \draw (0.13,0.13) -- (0.57,0.57) ;
  \draw (0.45,0.25) node {$\tiny \red{1}$ };
  
  \draw (0.7,0.7) circle (0.18) ;
  
  \draw (-0.13,0.13) -- (-0.57,0.57) ;
  \draw (-0.45,0.25) node {$\tiny \red{4}$};
  \draw (-0.7,0.7) circle (0.18) ;

  \draw (-0.57,0.83) -- (-0.13,1.27);
  \draw (-0.25,0.95) node { $\tiny \red{1} $ } ;
  \draw (0,1.4) circle (0.18);

  \draw (-0.83,0.83) -- (-1.27,1.27) ;
  \draw (-1.15,0.95) node {$\tiny \red{3}$};
  \draw (-1.4,1.4) circle (0.18) ;

  \draw (-2.8,1.4) circle (0.18) ;
  \draw (-2.2 ,0.95)  node {$\tiny \red{8}$};
  \draw (-0.87,0.73) -- (-2.7,1.26) ;

  
  \draw (-3,-0.2) node [below right] {
    \parbox{5cm}{$C_1=\set{\epsi,\, 1,\, 4,\, 4\ct 1,\, 4\ct 3,\, 4\ct 8  }$}};


  \draw (4,0) circle (0.18) ;

  \draw (4.13,0.13) -- (4.57,0.57) ;
  \draw (4.45,0.25) node {$\tiny \red{1}$ };

  \draw (4.7,0.7) circle (0.18) ;
  
  \draw (3.87,0.13) -- (3.43,0.57) ;
  \draw (3.55,0.25) node {$\tiny \red{4}$};
  
  \draw (3.3,0.7) circle (0.18) ;

  \draw (3.17,0.83) -- (2.73,1.27) ;
  \draw (2.85,0.95) node {$\tiny \red{3}$};
  \draw (2.6,1.4) circle (0.18) ;

  
  \draw(2,-0.2) node [below right] {
   \parbox{5cm}{\small $C_2=\set{\epsi,\, 1,\, 4,\, 4\ct 3 }$} };

  \end{tikzpicture}\\[-4ex]
\end{center}
As it turns out, $C_1$ is the support of a type \eg $(4\cdot(8\cdot \tv_3,3\cdot \tv_1)\rew \tv_2)\rew \tv_1$ (figure below). 
By contrast, no type $T$ may satisfy $\supp{T}=C_2$, because a non-terminal node of a type (necessarily an arrow) should have a child on track 1  (by convention, its right-hand side), but $4\in C_2$ and $4\cdot 1 \notin C_2$.
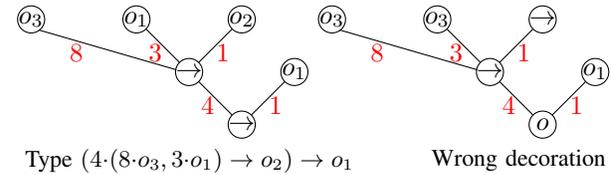
\begin{center}
  \begin{tikzpicture}

  \draw (0,0) node {$\rew$} ;
  \draw (0,0) circle (0.18) ;

  \draw (0.13,0.13) -- (0.57,0.57) ;
  \draw (0.45,0.25) node {$\tiny \red{1}$ };
  
  \draw (0.7,0.7) node {$\tv_1$} ;
  \draw (0.7,0.7) circle (0.18) ;
  
  \draw (-0.13,0.13) -- (-0.57,0.57) ;
  \draw (-0.45,0.25) node {$\tiny \red{4}$};
  \draw (-0.7,0.7) node {$\rew$} ;
  \draw (-0.7,0.7) circle (0.18) ;

  \draw (-0.57,0.83) -- (-0.13,1.27);
  \draw (-0.25,0.95) node { $\tiny \red{1} $ } ;
  \draw (0,1.4) node {$\tv_2$} ;
  \draw (0,1.4) circle (0.18);

  \draw (-0.83,0.83) -- (-1.27,1.27) ;
  \draw (-1.15,0.95) node {$\tiny \red{3}$};
  \draw (-1.4,1.4) node {$\tv_1$} ;
  \draw (-1.4,1.4) circle (0.18) ;

  \draw (-2.8,1.4) circle (0.18) ;
  \draw (-2.8,1.4) node {$\tv_3$} ;
  \draw (-2.2 ,0.95)  node {$\tiny \red{8}$};
  \draw (-0.87,0.73) -- (-2.7,1.26) ;

  \draw (-3,-0.2) node [below right] {
     \parbox{5cm}{\small Type $(4\ct(8\ct \tv_3,3\ct \tv_1)\rew \tv_2)\rew \tv_1$}};


  \draw (4,0) node {$\tv$} ;  
  \draw (4,0) circle (0.18) ;

  \draw (4.13,0.13) -- (4.57,0.57) ;
  \draw (4.45,0.25) node {$\tiny \red{1}$ };
  
  \draw (4.7,0.7) node {$\tv_1$};
  \draw (4.7,0.7) circle (0.18) ;
  
  \draw (3.87,0.13) -- (3.43,0.57) ;
  \draw (3.55,0.25) node {$\tiny \red{4}$};

  \draw (3.3,0.7) node {$\rew$};
  \draw (3.3,0.7) circle (0.18) ;
  
 \draw (3.43,0.83) -- (3.87,1.27);
  \draw (3.75,0.95) node { $\tiny \red{1} $ } ;
  \draw (4,1.4) node {$\rew$} ;
  \draw (4,1.4) circle (0.18);

  \draw (2.6,1.4) node {$\tv_3$};
  \draw (3.17,0.83) -- (2.73,1.27) ;
  \draw (2.85,0.95) node {$\tiny \red{3}$};
  \draw (2.6,1.4) circle (0.18) ;

   \draw (1.2,1.4) circle (0.18) ;
  \draw (1.2,1.4) node {$\tv_3$} ;
  \draw (1.8 ,0.95)  node {$\tiny \red{8}$};
  \draw (3.13,0.73) -- (1.3,1.26) ;
  
  \draw(2.4,-0.2) node [below right] {
      \parbox{5cm}{\small Wrong decoration} };
\end{tikzpicture}\\[-4ex] 
\end{center}
This motivates the following notion: a \textbf{support candidate (s-candidate)} of type is a subset $C\subseteq \bbN^*$ such that there exists a type $T$ satisfying $C=\supp{T}$. Given an s-candidate $C$, it is easy to define a correct type whose support is $C$:  
\begin{itemize}
\item The non-terminal nodes of $C$ should be decorated with arrows\paper{.}\techrep{ and \ldots}
\item \techrep{\ldots}the leaves of $C$ should be decorated with type variables.
\end{itemize}
So was done for the decoration on the left-hand side, representing the type $(4\cdot (8\cdot \tv_3,3\cdot \tv_1)\rew \tv_2)\rew \tv_1$}. 
  \techrep{As a counterexample,}\paper{In contrast,} the decoration on the right-hand side is incorrect: $\epsi$ (non-terminal) is labelled with $o\in \TypeV$ and  $4\cdot 1$ (leaf) with  $\rew$.\\

The observations about $C_1$ and $C_2$ above suggest considering two relations $\rewto$ and $\rewtt$ defined by:
\begin{itemize}
\item For all $c\in \bbN^*$, $k\in \bbN$, $c\cdot k \rewto c$.
\item For all $c\in \bbN^*$, $k\geqslant 2$, $c\cdot k \rewtt c\cdot 1$.
\end{itemize}
A set of positions $C$ is closed under $\rewto$ (\ie $c_1\in C$ and $c_1\rewto c_2$ implies $c_2\in C$) iff it is a tree. Stability under condition $\rewtt$ means that if a node $c$ is not terminal, then it has a child on track 1.
We have:
 
\begin{lemma}
\label{l:s-candidate-types}
  Let $C\subseteq \bbN^*$. Then $C$ is a type support candidate (\ie there exists a type $T$ s.t. $C=\supp{T}$) iff $C$ is non-empty and is closed under $\rewto$ and $\rewtt$.
\end{lemma}

Thus, relations $\rewto$ and $\rewtt$ are enough to characterize s-candidates. We call them \textbf{stability relations} \eg $C_1$ is stable under $\rewto$ and $\rewtt$, whereas $4\cdot 3\in C_2$, $4\cdot 3\rewtt 4\cdot 1$ but $4\cdot 1\notin C_2$, so that $C_2$ is not stable under $\rewtt$ (this example means that $4$ has no child on track 1 whereas it is not terminal and thus cannot be decorated by an arrow nor by a type variable).

When $c_1\rewto c_2$ or $c_1\rewtt c_2$, we say that $c_1$ \textbf{subjugates} $c_2$, because $c_1$ demands $c_2$ to ensure a correct\techrep{formation of the} support.

\subsection{Toward the Characterization of Bisupport Candidates}
\label{ss:toward-candidate-bisupp}

For the remainder of this paper, we fix an injection $\code{\cdot}:\bbN^* \rew \Nmzo$. By Goal\;\ref{goal:code-S-typability}, we want to prove that every term $t$ is $\code{\cdot}$-typable. By analogy with the notion of candidate supports for types (previous section), the idea is to characterize the \textbf{bisupport candidate\techrep{ (b-candidate)}} for the $\code{\cdot}$-derivations typing a given term $t$ \ie sets $B\subseteq \bbN^*\times \bbN^*$ s.t. there exists a $\code{\cdot}$-derivation $P$ typing $t$ satisfying $B=\bisupp{P}$ (Prop.\;\ref{prop:charac-bisupp} to come).

In  the remainder of this section, we define:
\begin{itemize}
\item $\bbB^t$, the set of the potential bipositions of a derivation typing a term $t$ (this Section\;\ref{ss:toward-candidate-bisupp}).
\item On $\bbB^t$, we define a relation $\rewbullet$ (which is actually the union of 7 stability relations). More precisely: 
\begin{itemize}
\item There is a special constant symbol $\bbot$ in $\bbB^t$, that roughly indicates ``untypability'' or ``emptiness''.
\item The term $t$ is typable iff there is a non-empty subset $B$ of $\bbB^t$, such that $B$ is stable under $\rewbullet$ and does \textit{not} contain $\bbot$. Such a $B$ is the support of a derivation typing $t$. This equivalence is given by the ``completeness-like'' statement of Corol.\;\ref{corol:typability-and-thepsi}.\\
\end{itemize}
\end{itemize}

Let us now define $\ttB^ t$ by first noticing  that not every position $a\in \bbN^*$ (or biposition $(a,c)\in \bbN^*\times \bbN^*$) may be in a derivation typing a given term $t$. For instance, we have $\supp{\lx.y\,x}=\set{\epsi,\,0,\,0\cdot 1,\,0\cdot 2 }$, so, if $P$ types $\lx.x\,x$, then $a\in \supp{P}$ implies $\ovl{a}=\epsi,\,0,\,0\cdot 1$ or $0\cdot 2$ \ie $\ovl{\supp{P}}\subseteq \set{\epsi,\,0,\,0\cdot 1,\,0\cdot 2 }$. For instance, $\supp{\Pex}=\set{\epsi,\,0,\,0\cdot 1,\, 0\cdot 2,\,0\cdot 5,\,0\cdot 6}$. 
 More generally, if $t$ is a term, we set $\bbA^t=\set{a\in \bbN^*\,|\,\ovl{a}\in \supp{t}}$ and $\bbB^t=(\bbA^t\times \bbN^*)\cup \set{\bbot}$ (where $\bbot$ is an ``empty biposition'' constant), so that, if $P$ is a derivation typing $t$, then a position (resp. a biposition) of $P$ must be in $\bbA^t$ (resp. in $\bbB^t\setminus\set{\bbot}$) \ie $\supp{t}\subseteq \bbA^t$ and $\bisupp{P}\subset \bbB^t\setminus\set{\bbot}$. 
The constant $\bbot$ roughly materializes emptiness and will be used  to describe how ``relevance related emptiness'' is constrained to occur in $\code{\cdot}$-derivations (see polar inversion in \S\;\ref{ss:asc-pi-c-u}.
\paper{\\}

\techrep{
Not every $a\in \bbA^t$ may be in a derivation $P$ typing $t$ \eg $2\in \supp{u}$ with $u=(\lx.y)z$, but if $P$ types $u$, then subterm $z$ is left untyped -- $\lx.y$ must be typed with $\est \rew T$ (relevance) -- and 2 cannot be in $\supp{P}$.\\}

We drop \paper{from now on }$P$ and $t$ from most notations\techrep{, in which they are implicit}.
 We set $\bbAa(x)=\set{a_0\in \bbA\,|\,a\leqslant a_0,\, t(a_0)=x,\; \nexists a'_0,\,a\leqslant a'_0< a_0,\, t(a'_0)=\lx }$. Thus, if $P$ is a $\code{\cdot}$-derivation, then, with the notation $\ttAx^P_a$ (\S\;\ref{ss:parsing-pointing}), 
 $\AxPa(x)\subset \bbAa(x)$ for all $a\in \supp{P},\, x\in \TermV$ and $\bbAa(x)$ may be considered as the set of position candidates for $\ax$-rules typing   \techrep{the free occurrences of} $x$ above $a$.

Now, remember that the function $\code{\cdot}$ has been fixed to choose axiom tracks: if $x$ a variable and $a_0$ an axiom position candidate for $x$ (\ie $t(a_0)=x$), then a potential $\code{\cdot}$-derivation $P$ containing $a_0$ has an axiom of the form $P(a_0)=\ju{x:(k\cdot T)}{x:T}$ with $k=\code{a_0}$. Thus, if 
$t(a)=\lx$ and we set $\Trl{a}=\set{\code{a_0}\,|\,a_0\in \bbA_{a\cdot 0}(x)}$, then $\Trl{a}$ is the set of axiom tracks dedicated to $x$ above the $\abs$-rule at position $a$ by the function $\code{\cdot}$. It is interesting in that, \eg if $t(a)=\lx$, $8 \notin \Trl{a}$, then we can assert that, if there exists a $\code{\cdot}$-derivation $P$ and $a\in \supp{P}$, then $P(a)=\sSk \rew T$ with $8\notin K$. Indeed, by definition of $\Trl{a}$, there is no axiom position candidate $a_0$ for $x$ above $a$ whose axiom track is 8.
\fublainv{peut-être expliquer ce qu'on entend par potential derivation}

Thus, when a variable $x$ is not at some places in $t$, $\code{\cdot}$ contrains  emptiness to ``occur'' at some particular tracks if we perform an abstraction $\lx$. This give us more fine-grained information about occurrences of emptiness in a derivation typing $t$ than the case where $\lx.u:\est\rew T$ because $x$ does not occur free in $u$: system $\ttS$  will provide us information about emptiness \textit{track by track}. This is precisely what we need to understand typability in the relevant and non-idempotent framework (remember \S\;\ref{ss:intro-infty-relev-typing}): we have to ensure that emptiness does not compromise typability \ie intuitively, emptiness must not propagate everywhere in the derivations typing a given term $t$. If it did, a derivation typing $t$ would be empty (\ie $t$ would not be typable) and we want to show that this does not happen, in the purpose of  proving that every term is typable in $\ttS$

\subsection{Tracking a Type in a Derivation}
\label{ss:asc-pi-c-u}

Let us now express the stability conditions (as in Sec.\;\ref{ss:candidate-supp-types}) that a $\code{\cdot}$-bisupport candidate for a derivation typing $t$ should satisfy. We will need to ensure \techrep{six}\paper{the}  points\paper{ below}:
\begin{itemize}
\item Identification of the components (\ie the bipositions) of a same type $T$ in a derivation from bottom to top (see Fig.\;\ref{fig:asc-cons}): relation of \textbf{ascendance} $\rewa$.
\item Identification of the components of  type given in an $\ax$-rule to a variable $x$ ($S_5$ in Fig.\;\ref{fig:asc-cons}) and its occurrence called by the abstraction $\lx$: relation of \textbf{polar inversion} $\rewp$.
\item Identification of the matching components of the types of $u$ and $v$ in the $\app$-rule typing $u\,v$ (types $S_k$ in the $\app$-rule of Fig.\;\ref{fig:asc-cons}): relation of \textbf{consumption} $\rew$.     
\item Correct type formation, as in Sec.\;\ref{ss:candidate-supp-types}: extensions of relations $\rewto$ and $\rewtt$.
\item The type of a subterm of the form $\lx.u$ is an arrow type (and not a type variable): relation $\rewr$.\paper{\\}
\techrep{
\item For technical reasons, we also need a ``big-step" stability condition, meaning that the support of a derivation is a tree: relation $\rewdown$.\\}
\end{itemize}

Once again, by lack of space, most of the proofs are omitted for the remainder of the paper and we can give only a few illustrations of the concepts that we use. We refer to \fucite{} and Chapers 11 and 12 of \cite{VialPhd} for all the details and more examples and heuristics.

\techrep{We formalize these ideas in the following sections. }\techrep{For the discussion below, let us recall again that, given a type $T$ and a sequence type $\sSk$, a position $c\in \supp{T}$ corresponds to the position $1\cdot c\in \supp{\sSk\rew T}$, since $T$ occurs in this arrow type right-hand side. And if $k\in K$, position $c$ in $S_k$ corresponds to position $k\cdot c$ in $\sSk$.}
In Fig.\;\ref{fig:asc-cons}, we indicate the position of a judgment between angle brackets \eg $\ju{C;x:\sSksh}{t:T}\posPr{\,a\cdot 0\,}$ means that judgment $\ju{C;x:\sSksh}{t:T}$ is at position $a\cdot 0$. We denote by $\ttpos$ the (partial) converse of $\code{\cdot}$ \eg, if $a_0:=\pos{7}$ exists, then $a_0$ is the axiom position candidate whose axiom track is 7: concretely, this just means that, if there exits a $\code{\cdot}$-derivation $P$ typing $t$ s.t. $a_0\in \supp{P}$, then $P(A)=\ju{x:(7\ct S)}{x:S}$ for some type $S$ and  \paper{$x\in \TermV$}\techrep{variable $x$}.
\\
\begin{figure}[h]
  \vspace{-0.7cm}
  \begin{center}
  \begin{tikzpicture}

\draw (0,2.8) node [right] {\small \bf Abstraction rule};
  
  \draw (1.5,2.25) --++ (4.2,0);
  \draw (5.6,2.25) node [right] {\small $\ax$};
  \draw (1.5,2.3) node [below right] {$\ju{x:(5\ct S_5)}{x:S_5}\posPr{\,\pos{5}\,}$};

  \draw [dashed] (3.2,1.8) --++ (0,-0.2) --++(-0.7,-0.5) --++ (0,-0.2);

  \draw (0,1) node [below right] { $\ju{C;\,x:\sSk}{u:T}\!\posPr{\,a\ct 0\,}$}
   ;
  \draw (5.8,0.4) node [right] {\small (with $5\in K,8\notin K$)};
  
  \draw (0.1,0.4) --++(4.4,0);
  \draw (4.4,0.45) node [right] {$\abs$};
  
  \draw (0,0.4) node [below right] {$\ju{C }{\lambda x.u:\sSk\rew T }\!\posPr{a}$ };



\draw (0,-0.4) node [right]{\small \bf Application rule};
  
 \draw (0,-0.7) node [below right] {$\ju{C}{u: \sSk\! \rew\! T}  \posPr{\,a\ct 1\,} \hspace{0.4cm} (\ju{D_k}{v:S_k}  \posPr{\,a\ct k\,} )_{\kK}$
  } ;

  \draw (0,-1.25) --++(7.65,0);
  \draw (7.77,-1.25) node [right] {\small $\app$};
  
  \draw (1.77,-1.2) node [below right] {$\ju{C \uplus(\uplus_{\kK} D_k)}{u\,v:T} \posPr{\,a\,}$ };
\end{tikzpicture}\\[-2ex]
    \caption{Ascendance, Polar Inversion and Consumption}
    \label{fig:asc-cons}
  \end{center}
  \vspace{-0.3cm}
\end{figure}

\noindent \textbullet~ Assume that, in a $\code{\cdot}$-derivation $P$, we find an $\abs$-rule at position $a$ as in Fig\;\ref{fig:asc-cons}: 
the judgment $\ju{C;x:\sSksh}{u:T}$ (pos. $a\cdot 0$) yields $\ju{C}{\lx.u:\sSksh\rew T}$ below (pos. $a$). The occurrence of $T$ in the conclusion of the rule is intuitively the same as that in its premise: we say the former is the \textbf{ascendant} of the latter, since it occurs above in the typing derivation. Likewise, in the $\app$-rule, the occurrence of $T$ in $\ju{C\uplus_{\kK} D_k}{u\,v:T}$ stems from that of premise $\ju{C}{u:\sSksh\rew T}$: the first occurrence of $T$ is also the ascendant of $T$ in the conclusion of the rule. \paper{Ascendance induces a stability relation $\rewa$ on $\ttB^t$, the set of candidate bipositions of $t$, that can be formally defined by:} 
\begin{itemize}
\item $(a,c)\rewa (a\cdot 1,1\cdot c)$ if $t(a)=\arob$.
\item $(a,1\cdot c)\rewa (a\cdot 0,c)$ if $ t(a)=\lx$.
\end{itemize}
\techrep{
 This formalization relies on the correspondence between $c$ and $1\cdot c$ above.
 Relation $\p_1 \rewa  \p_2$ means that $\p_2$ is the ascendant of $\p_1$ \ie $\p_1$ and $\p_2$ are corresponding pointers to the same type symbol in the conclusion and the (left) premise of the rule at some position $a$.
\pierre{EXEMPLE TRAITER}} For instance, in Fig\;\ref{fig:a-thread-in-Pex}, the bred (resp. the blue) occurrences of $\tv'$ are ascendants of one another. They correspond to bipositions $(\epsi,1^2)\rewa (0,1) \rewa (0^2,\epsi) \rewa (0^2\cdot 1,1)$ (resp. $(\epsi,1\cdot 4\cdot 1)\rewa (0,4\cdot 1)$).\\
{
\begin{figure*}[t]
\[\infer[\abs]{\infer[\abs]{\infer[\app]{\infer[\ax]{\phd}{\ju{4}{x: (8\ct \orange{o},3\ct o',2\ct o)\rew \red{o'}}  }\hspace{0.5cm}
      \infer[\ax]{\phd}{\ju{9}{x:\tv}\trck{2}} \hspace{0.2cm}  
\infer[\ax]{\phd}{\ju{2}{x:\tv'}\trck{3}} \hspace{0.2cm}
   \infer[\ax]{\phd}{\ju{5}{x:\purple{\tv}}\trck{8}}
    }{\ju{\ldots}{x\,x:\red{\tv'}}}}{
  \ju{}{\lx.x\,x:(2\cdot \tv',4\cdot (8\cdot \green{\tv},3\cdot \tv',2\cdot \tv)\rew \blue{\tv'},5\cdot \tv,9\cdot \tv)\rew \red{\tv'}}
}}{
\ju{}{\ly x.xx:\est\rew (2\cdot \tv',4\cdot (8\cdot \green{\tv},3\cdot \tv',2\cdot \tv)\rew \blue{\tv'},5\cdot \tv,9\cdot \tv)\rew \red{\tv'}}
  }\]
\caption{Threads, Ascendance and Consumption}
\label{fig:a-thread-in-Pex}
\vspace{-0.4cm}
\end{figure*}}

\ignore{
\begin{figure*}[t]
\[\infer[]{\infer{\infer{\infer{\phd}{\ju{x:\!(4\ct (8\ct \green{o},3\ct o',2\ct o)\rew o')}{x: \!(8\ct o,3\ct o',2\ct o)\rew \red{o'}}  }\hspace{0.5cm}
      \infer{\phd}{\ju{x:\!(9\ct \tv)}{x:\!\tv}}\trck{2} \hspace{0.2cm}  
\infer{\phd}{\ju{x:\!(2\ct \tv')}{x:\!\tv'}}\trck{3} \hspace{0.2cm}
   \infer{\phd}{\ju{x:\!(5\ct \tv)}{x:\!\green{\tv}}}\trck{8}
    }{\ju{x:(2\cdot \tv',4\cdot (8\cdot \tv,3\cdot \tv',2\cdot \tv)\rew \tv',5\cdot \tv,9\cdot \tv)}{x\,x:\red{\tv'}}}}{
  \ju{}{\lx.xx:(2\cdot \tv',4\cdot (8\cdot \tv,3\cdot \tv',2\cdot \tv)\rew \blue{\tv'},5\cdot \tv,9\cdot \tv)\rew \red{\tv'}}
}}{
\ju{}{\ly x.xx:\est\rew (2\cdot \tv',4\cdot (8\cdot \tv,3\cdot \tv',2\cdot \tv)\rew \blue{\tv'},5\cdot \tv,9\cdot \tv)\rew \red{\tv'}}
  }\]
\caption{A Thread composed of a Positive Ascendant Thread (red) and a Negative Ascendent Thread (blue)}
\label{fig:a-thread-in-Pex}
\vspace{-0.4cm}
\end{figure*}}

\noindent \textbullet~ Let us have another look at the $\abs$-rule at position $a$ in Fig.~\ref{fig:asc-cons}. Assume $5\in K$.
Then the occurrence of $S_5$ in $\sSk \rew T$ at pos. $a$ stems from an axiom rule concluding with $\ju{x:(5\cdot S_5)}{x:S_5}$ at pos. $\pos{5}$: we say that the occurrence $S_5$ (in $\sSk \rew T$) is the \textbf{polar inverse} of the occurrence of $S_5$ in the axiom rule.
Assume on the contrary that $8 \notin K$. So $S_8$ does not exist and there is no $\ax$-rule typing $x$ and using axiom track 8 above $a$. Morally, $S_8$ is empty.

More generally, we recall from the end of \S\;\ref{ss:toward-candidate-bisupp} that, if $t(a)=\lx$ and $k_0\geqs 2$, the function $\code{\cdot}$\ldots
\begin{itemize}
\item \ldots either gives an unique axiom position candidate for the type on track $k$ in $\sSk$: this happens when $\exists a_0\geqs a\cdot 0$ s.t. $\code{a_0}=k_0$ (\ie, when $k_0\in \Trl{a}$ by construction of $\Trl{a}$). In that case, $a_0=\pos{k_0}$. 
\item  \ldots or tells us that $k_0\notin K$ \ie there is no (top-level) type on track $k_0$ in $\sSk\rew T$: this happens when no $a_0\geqs a\cdot 0$ satisfies $\code{a_0}$ (\ie $k_0\notin \Trl{a}$). In that case, $\pos{k_0}$ is undefined. We consider the type $S_{k_0}$ (that intuitively does not exist) be the polar inverse of an empty type.
\end{itemize}
Polar inversion also induces a stability relation $\rewp$ on $\bbB^t$, that can be formally defined by:
\begin{itemize}
\item $(a,k\cdot c) \rewp ( \pos{k},c)$ if  $k\in \Trl{a}$\hfill (first case)
\item $(a,k\cdot c) \rewp \bbot$ if  $k\notin \Trl{a}$ \hfill (second case)
\end{itemize}
Now, we may understand the use of constant $\bbot$: it indicates biposition that \textit{cannot} be in any potential $\code{\cdot}$-derivation typing $t$. More precisely, $\bbot$ is here to play the role of the polar inverse of all the bipositions that cannot exist, because of the choices made by the function $\code{\cdot}$.
\techrep{

  \ignore{
For instance, in $\Pex$ (with a suitable choice of $\code{\cdot}$), we have $(\epsi,9)\rewp (0\cdot 2,\epsi)$ (because the $\ax$-rule typing $x$ and using track 9 is at position $0\cdot 2$): indeed, $(\epsi,9)$ and $(0\cdot 2,\epsi)$ both point to the type symbol $o$. Likewise, $(\epsi,4\cdot 1)\rewp (0\cdot 1,1)$ and   in Fig.\;\ref{fig:a-thread-in-Pex}, the occurrence of $\tv'$ at biposition $(\epsi,9)$ is colored in blue whereas that at $(0\cdot 1,1)$ is the top occurrence of $\tv$ colored in red.
Likewise, $(\epsi,4)\rewp (0\cdot 1,\epsi)$ (they both point to the type symbol $\rew$) and $(\epsi, 4\cdot 3)\rewp (0\cdot 1,3)$ (they both point to the type symbol $o'$). Since 3 and 8 are not used as axiom tracks for $x$, we should have $(\epsi,3)\rewp \bbot$ and $(\epsi,8\cdot 3\cdot 2\cdot 1 \cdot 0)\rewp \bbot$ (for a good choice of $\code{\cdot}$). } \pierre{EXEMPLE TRAITER}} For instance, in Fig.\;\ref{fig:a-thread-in-Pex}, the top blue occurrence of $\tv'$ is the polar inverse of the top red one: formally, $(0,4\cdot 1)\rewp (0^2\cdot 1,1)$.

\subsection{Type Formation, Type Destruction}
\label{s:consumption}

In this subsection, we conclude the definitions of the stability relations that characterize the form of $\ttS$-derivations, yielding the notion of \textbf{subjugation}, as in \S\;\ref{ss:candidate-supp-types}. \\

\noindent \textbullet~ The notion of \textbf{consumption} is related to rule $\app$. Assume $t(a)=\arob,\,\tra=u\,v$ with $u:\sSk \rew T$ and $v:S_k$ for all $\kK$ as in Fig.\;\ref{fig:asc-cons} so that $u\,v$ can be typed with $T$. Each type $S_k$ occurs in $\sSk \rew T$ and $v:S_k$. However, it is absent in the type of $u\,v$: we say it has been \textbf{consumed}. Formally, we set, for all $(a,c)\in \bbB^t,\, k\geqslant 2$ s.t. $t(a)=\arob$:
\begin{itemize}
\item $(a\cdot 1,k\cdot c)\rewc{a} (a\cdot k, c)$
\end{itemize}  
\techrep{\pierre{
Indeed, the premise concluding with $u:\sSk \rew T$ (resp. with $v:S_k$) is at position $a\cdot 1$ (resp. $a\cdot k$). Position $c\in \supp{S_k}$ corresponds to position $k \cdot c$ in  $\supp{\sSk\rew T}$. For instance, in $\Pex$, there is an $\app$-rule at position 0 and 
$(0\cdot 1,8)\rewc{0} (0\cdot 8,\epsi)$ (pointing to type symbol $o$) and $(0\cdot 1,3)\rewc{0} (0\cdot 3,\epsi)$ (pointing to $o'$).
}} In Fig.\;\ref{fig:a-thread-in-Pex}, the orange and the purple occurrences of $\tv$ are consumed in the $\app$-rule: formally, $(0^2\cdot 1,8)\rewc{0^2} (0\cdot 8,\epsi)$. 

We set $\rew\,=\cup \set{\rewc{a}\,|\, a\in  \bbA,\, t(a)=\arob}$ and write $\wer$ for the symmetric relation.\\

Let $P$ be a $\code{\cdot}$-derivation typing a term $t$. If $\p_1\rewa \p_2$ or $\p_1 \rewp \p_2$ or $\p_1 \rew \p_2$, then $\p_1 \in \bisupp{P}$ iff $\p_2 \in \bisupp{P}$ (by construction of those relations).\paper{\\}

\techrep{
However, relations $\rewa,\, \rewp$ and $\rew$ are not enough to express the stability conditions characterizing a bisupport candidate for a derivation typing $t$ (as we have for types with Lem.~\ref{l:s-candidate-types}).
We need to ensure that types are correctly formed and that an abstraction is typed with an arrow type (and not with a type variable).
For that, we define additional relations $\rewto,\, \rewtt,\ \rewr$ and $\rewdown$ below.\\
}



\noindent \textbullet~Relations $\rewto$ and $\rewtt$ ensure that the types are correctly defined and are natural extensions of those of Sec\;\ref{ss:candidate-supp-types}:
\begin{itemize}
\item For all $(a,c) \in \bbB^t$ and $k\in \bbN$, $(a,c\cdot k)\rewto (a,c)$.
\item  For all $(a,c) \in \bbB^t$ and $k\geqslant 2$, $(a,c \cdot 1) \rewtt (a,c\cdot k)$.\\
\end{itemize}

\noindent \textbullet~
\techrep{
    Note that, for instance, if $t=I=\lx.x$ and $B=\set{\pepsi}=\set{(\epsi,\epsi)}$, then $B$ is stable under $\rewa,\; \awer,\; \rew,\;\wer,\; \rewto,\; \rewtt$ but obviously, there is no derivation $P$ typing $t$ such that $\bisupp{P}=B$ (indeed, $t$ is necessarily typed with an arrow type \ie $1\in \supp{\ttT^P(\epsi)}$, which means $(\epsi,1)\in \bisupp{P}$). Thus, those relations are not enough to describe a candidate bisupport. \\
}
The relation below $\rewr$ ensures that, if $\lx.u$ is a typed subterm of $t$, then its type $T$ is an arrow type\techrep{ (and not a type variable}. \techrep{When $T$ is an arrow, then $\supp{T}$ must at least contain 1 (besides $\epsi$), the position of the root of right-hand side of the arrow. We set then, for all $a\in \bbA^t$ s.t. $t(a)=\lx$:}
\begin{itemize}
\item $(a,\epsi)\rewr (a,1)$\\
\end{itemize}

\noindent \textbullet~\paper{The ``big-step'' stability relation $\rewdown$ below roughly states that the support of a potential derivation is a tree:}\techrep{Relation $\rewdown$ is included in the reflexive transitive closure of $\awer\cup\wer\cup \rewto \cup \rewr$, but turns out to be useful (to obtain Lemma~\ref{l:ref-red}). When the biposition $(a',c)$ is in a derivation $P$, then every position below $a'$ must be in $P$ (must have a non-empty type). We set then, for all $a,a'\in \bbA$ and $c\in \bbN^*$ s.t. $a\leqslant a'$:}
\begin{itemize}
\item $ (a',\,c)\rewdown (a,\,\epsi)$
\item $\bbot \rewdown (a,\epsi)$
\end{itemize}
The 2nd case is also useful to ensure Lemma~\ref{l:ref-red}.\\

\noindent \textbullet~ We set $\rewbullet =\,\rew \cup \wer \cup \rewto \cup \rewtt \cup \rewr \cup \rewdown$. If $\p_1\rewbullet \p_2$, notice that, by construction, $\p_1 \in \bisupp{P}$ implies $\p_2\in \bisupp{P}$.
We say then that $\p_1$ \textbf{subjugates} $\p_2$, generalizing \S\;\ref{ss:candidate-supp-types}.


\subsection{Threads and Minimal Bisupport Candidate}
\label{s:bbBm}

We prove now that the relations above are indeed enough to express a sufficient condition of typability (Corollary\;\ref{corol:typability-and-thepsi}).

As we have seen, if $P$ is a $\code{\cdot}$-derivation, then $\bisupp{P}$ is closed under $\rewa$, $\awer$, $\rewp$, $\pwer$, $\rew$, $\wer$, $\rew$, $\rewto$, $\rewtt$, $\rewr$ and $\rewdown$. Of course, $\bbot$, the empty biposition, cannot be in $P$. It turns that it is enough to characterize candidate bisupports (Proposition\;\ref{prop:charac-bisupp}). In this statement,  $\iden$ is the reflexive, transitive, symmetric closure of $\rewa\cup \rewp$. We have:

\begin{proposition}
  \label{prop:charac-bisupp}
  Let $B\subseteq \bbB^t$. Then $B$ is a  $\code{\cdot}$-candidate bisupport for a derivation typing $t$ (\ie there exists a $\code{\cdot}$-derivation s.t. $B=\bisupp{P}$) iff (1) $B$ is non-empty, (2) $B$ is closed under $\iden$ and  $\rewbullet$, and (2) $B$ does not contain $\bbot$.
\end{proposition}

If the closure of a set $B$ contains $\bbot$, then intuitively, $B$ needs to use a slot that is constrained (by relevance) to be empty: thus, no derivation can contain $B$. 

\techrep{\begin{proof}}              
\paper{
  \vspace{0.2cm}\begin{proofsketch}}  
The necessity of these conditions has been discussed in the previous subsections.

Conversely, assume that $\emptyset \neq B\subset \bbB^t\setminus\set{\bbot}$ is closed under $\iden$ and $\rewbullet$. We want a derivation $P$ s.t. $\bisupp{P}=B$. For that, we need to suitably decorate the $\p \in B$. Mainly, a non-terminal biposition must be labelled with $\rew$ and a terminal one with a fixed type variable $o$, in order to get correct types (as in \S~\ref{ss:candidate-supp-types}).
\techrep{
  Thus, we set $\Lves{B}=\set{(a,c)\in B\,|\, (a,c\cdot 1)\notin B}$ and we define $P$ on $B$ by $P(\p)=\tv$ if $\p \in \Lves{B}$ and $P(\p)=\rew$ if not.} \techrep{We now verify}\paper{On can check} that $P$ is a correct $\ttS$-derivation using the definition of $\iden$ and $\rewbullet$.\paper{\end{proofsketch}}

\techrep{ 
  Let $A=\set{a\in \bbA^t\,|\, \exists c\in \bbN^*,\, (a,c)\in B}$. Thus, $A\subseteq \bbA$. For all $a\in A$, we set $\ttT(a)(c)=P(a,c)$ whenever $(a,c)\in B$. For all $a\in A$ and $x\in \TermV$, we set $A_a(x)=A\cap \bbAa(x)$ and  $\ttC(a)(x)=\uplus_{a_0\in \bbAa(x)} (\code{a_0}\cdot \ttT(a_0))$. Thus, $\ttT(a)$ and $\ttC(a)(x)$ are functions from $\bbN^*$ to $\TypeV\cup \set{\rew}$. By hypothesis, if $\thr{\p}=\thbot$, then $\p\notin B$. 



For all $a\in \bbA$, if $a\in A$, $\dom{\ttT(a)}$ is a tree and $\ttT(a)$, as a labelled tree, is a correct type and if $a\notin A$, $\dom{\ttT(a)}=\emptyset$.\\
Indeed, if $a\in A$, then $\dom{\ttT(a)}$ is non-empty. The definitions of $\Lves{B}$, $\ttT(a)(c)$ and Lemma\;\ref{l:s-candidate-types} grant then that $\ttT(a)$ is a correct type.

\noindent We consider now $a\in A$ and check that the typing rules are respected.
\begin{itemize}  
\item Assume $t(a)=x$. Then, by definition of $\ttC$, $C(a)(y)=\est$ if $y\neq x$ and $\ttC(a)(x)=(\code{a}\cdot \ttT(a))$, so that $a$ is a correct axiom rule.\\

\item Assume $t(a)=\arob$.
\begin{itemize}
\item  Since $(a,c)\rewa (a\cdot 1,1\cdot c)$, then $(a,c)\in B$ iff $(a\cdot 1,1\cdot c)\in B$ and even $(a,c)\in \LB$ iff $(a\cdot 1,1\cdot c)\in \LB$, so that functions $c\mapsto \ttT(a)(c)$ and $c\mapsto \ttT(a\cdot 1)(1\cdot c)$ are equal. \fublainv{conclure par \ttT(a)=\Hd{\ttT(a\cdot 1)} \ie \ttT(a\cdot 1) is of the form $\ldots \rew \ttT(a)$}
\item  Moreover, by $\rew$ and $\wer$, 
  $(a\cdot 1,k\cdot c)\in B\iff (a\cdot k,c)\in B$, and even $(a\cdot 1,k\cdot c)\in \LB\iff (a\cdot k,c)\in \LB$. Thus, for all $k\geqslant 2,c\in \bbN^*,\, \ttT(a\cdot 1,k\cdot c)=\ttT(a\cdot k,c)$.
\item  Since $a\in A$, $(a,\epsi) \in B$, so $(a\cdot 1,1)\in B$, so $(a\cdot 1,\epsi)\notin \LB$, so $\ttT(a\cdot 1)(\epsi)=\rew$, by definition of $T$.
\end{itemize}
  Thus, $\ttT(a\cdot 1)=(\ttT(a\cdot k))_{k\geqslant 2}\rew \ttT(a)$.\\

  By definition of $\ttC$ and of $A_a(x)$, we easily obtain $\ttC(a)(x)=\cup_{k\geqslant 1,a\cdot k\in A} \ttC(a\cdot k)(x)$ as expected.

  Thus, $a$ is a correct $\app$-rule.\\

\item Assume $t(a)=\lx$. Then, by $\rewr$\,, $(a,1)\in B$ and $(a,\epsi)\notin \LB$, so that $\ttT(a)(\epsi)=\rew$. 
  Since $(a,1\cdot c) \rewa (a\cdot 0,c)$, then $(a,1\cdot c) \in B$ iff $(a\cdot 0,c) \in B$ and even $(a,1\cdot c) \in \LB$ iff $(a\cdot 0,c) \in \LB$, so that functions $c\mapsto \ttT(a,1\cdot c)$ and $(c \mapsto \ttT(a\cdot 0,c))$ are equal.\\
  Moreover, let $k\geqslant 2$:
  \begin{itemize}
  \item If $k\notin \Trl{a}$, then $(a,k\cdot c) \rewp \bbot$, so $\thr{a,k\cdot c}=\thbot$, so, by hypothesis, $(a,k\cdot c) \notin B$, so $k\cdot c\notin \supp{\ttT(a)}$.
  \item If $k\in \Trl{a}$, then $(a,k\cdot c) \rewp (\pos{k},c)$, so $(a,k\cdot c)\iden (\pos{k},c)$. So $(a,k\cdot c)\in B$ iff $(\pos{k},c)\in B$ and even $(a,k\cdot c)\in \LB$ iff $(\pos{k},c)\in L(B)$.
  \end{itemize}
  This shows that, for all $k\geqslant 2$, functions $c\mapsto \ttT(a,k\cdot c)$ and $c\mapsto \ttC(a\cdot 0)(x)(c)$ are equal, by definition of $\ttC$.\\
  Thus, we have $\ttT(a)=\ttC(a\cdot 0)(x)\rew \ttT(a\cdot 0)$ and $a$ is a correct $\abs$-rule.
      \end{itemize}    
\end{proof}

\begin{remark}
We only use one type variable $\tv$ in the above proof. A general method (yielding every $\code{\cdot}$-derivation whose bisupport is $B$) is to label $\p_1,\p_2\in \LB$ with the same type variable whenever $\p_1 \iden_{\arob} \p_2$, where $\iden_{\arob}$ is the reflexive, transitive, symmetric closure of $\rewa\cup \rewp \cup \rew$.
\end{remark}}

From now on, it will be better to reason modulo $\iden$ (it may already be guessed that $\iden$ \textit{should} commute with $\rew,\, \rewto,\ldots$, which is made explicit in \S\;\ref{s:interaction-normal-chains}) and to focus on subjugation.

\ignore{
\begin{figure}[t]
\begin{center}
{
\begin{tikzpicture}

  \inputtypeleftn{0.4}{6.3}{$ \mult{\tv_1,\tv_2}\rewsh \mult{\purple{\tv}} \rewsh \tv \hspace{0.5cm} \phd$}
  \drawlabnode{0.4}{6.3}{$x$}
  
 \inputtypeleftn{1.05}{7.2}{$\mtv \rew \green{\tv_1}$}
  \drawlabnode{1.05}{7.2}{$y$}
  \inputtypeleftn{2.35}{7.2}{$\tv$}
  \drawlabnode{2.35}{7.2}{$z$}
  \blocka{1.7}{6.3}
  \outputtyperight{1.7}{6.3}{$\;\;\green{\tv_1}$}
  
  \inputtyperightn{3.05}{7.2}{$\erewa_2$}
  \drawlabnode{3.05}{7.2}{$y$}
  \drawlabnode{4.35}{7.2}{$z$}
  \blocka{3.7}{6.3}
  \outputtyperight{3.7}{6.3}{$\;\;\tv_2$}
  
   \draw (1.25,5.45) -- (3.5,6.25);
   \blocka{1.05}{5.4}
   \outputtypeleft{1.05}{5.4}{$ \mult{\purple{\tv}}\rew \tv$}
   
  \drawlabnode{4.35}{5.4}{$y$}
  \inputtyperightn{4.35}{5.4}{$\tv$}
  
  \draw (1.9,4.55) -- (4.15,5.35);
  \blockappl{1.7}{4.5}
  \outputtypeleft{1.7}{4.5}{$\tv$}
  \blockunary{1.7}{3.6}{$\ly$}
  \outputtypeleft{1.7}{3.6}{$\mult{\mtv \rewsh \green{\tv_1},\emul \rewsh \tv_2,\tv} \rew \tv$}
  \blockunary{1.7}{2.7}{$\lx$}
  \outputtypeleft{1.7}{2.7}{$\mult{ \mult{\tv_1,\tv_2}\rewsh \mult{\purple{\tv}} \rewsh \tv}\blue{\rew} \mult{\mtv\rewsh \green{\tv_1},\emul \rewsh \tv_2,\tv} \rew \tv$}
  \blockunary{1.7}{1.8}{$\lz$}
  \outputtypeleft{1.7}{1.8}{$\mult{\tv}\rew \mult{\mult{\tv_1,\tv_2}\rewsh \mult{\purple{\tv}} \rewsh \tv}\blue{\rew} \mult{\mtv\rewsh \green{\tv_1},\emul  \rewsh \tv_2,\tv} \rew \tv$}

  \inputtyperightn{2.85}{2.7}{$\mult{\tv_3}\rew \red{\tv}$}
  \drawlabnode{2.85}{2.7}{$f$}
  
  \drawtri{4.6}{2.7}{$u$}
  \outputtyperight{4.6}{2.7}{$\tv_3$}
  \draw (3.78,1.95) -- (4.6,2.7);

  \blockappl{3.6}{1.8}
  \outputtyperight{3.6}{1.8}{$\;\;\red{\tv}$}

  \draw (2.53,1.05) -- (3.43,1.65);
  
  \blockappl{2.35}{0.9}
   \outputtypeleft{2.35}{0.9}{$\mult{ \mult{\tv_1,\tv_2}\rewsh \mult{\purple{\tv}} \rewsh \tv}\blue{\rew} \mult{\mtv\rewsh \green{\tv_1},\emul \rewsh \tv_2,\tv} \rew \tv$}

  \drawtri{4.9}{0.9}{$v$}
  \outputtyperight{4.9}{0.9}{$\mult{\tv_1,\tv_2}\rewsh \mtv \rewsh \tv $}
  \draw (3.18,0.15) -- (4.9,0.9);
  
  \blockappl{3}{0}
  \outputtypeleft{3}{0}{$\mult{\mtv\rewsh \green{\tv_1},\emul \rewsh \tv_2,\tv} \rew \tv$}


   \draw [->,>=stealth,dotted] (2.05,6.2) --++ (0.13,0.28) --++ (-1.14,1.44) --++ (-0.25,0);

  \draw [->,>=stealth,loosely dashed] (-1.1,3.5) --++ (0,0.5) --++ (-1.3,1.3) --++ (0,3.1) --++ (2.6,0) --++ (0.3,-0.3)  ;
  \draw (-3.2,6) node {\parbox{2cm}{\begin{center}\fnsz \texttt{polar\\ inversion}\end{center}}};

  \ascarrowup{-1.1}{2.53}
  \ascarrowup{-1.1}{1.63}
    \ascarrowleft{-0.55}{0.73}
    \ascarrowleft{0.1}{-0.17}


  \ascarrowup{-2.25}{1.63}
    \ascarrowleft{-1.8}{0.73}

\draw [->,>=stealth,dotted] (3.95,1.65) --++ (0.15,0.15) --++ (0,1.35) --++ (-0.15,0.15) ;



  \draw [->,>=stealth,dotted] (-0.05,5.2)--++ (-0.3,0.4) --++ (0,1.2) ;
 
  \draw [->,>=stealth,loosely dashed] (-3.35,2.4) --++ (0,2.4) --++ (-0.7,0.7) --++ (0,2.1)   --++ (3.3,0) --++ (0.4,-0.4) ;
  
  \ascarrowup{-3.35}{1.63}
    \ascarrowleft{-2.85}{0.73}

\end{tikzpicture}

  }
\caption{Ascendance, Polar Inversion and Threads (Informal)}
\label{fig:asc-pol-inv-informal-1}
\end{center}
\end{figure}

}

\begin{definition} 
  \label{def:threads-comp-unsound}
Let $t$ be a term and $\code{\cdot}:\bbN^* \rightarrow \bbN\setminus \set{0,1}$ an injection, and $\rewa$, $\rewp$ the relations of ascendance and polar inversion in $\bbB^t$ defined \wrt $\code{\cdot}$.
\begin{itemize}
 \item An \textbf{ascendant thread} is an equivalence class of relation $\idena$, the reflexive, transitive, symmetric closure of $\rewa$.
 \item A \textbf{thread} (metavariable $\theta$) is an equivalence class of relation $\iden$ (see Fig.\;\ref{fig:a-thread-in-Pex} \techrep{ and \;\ref{fig:threads}}).
\item  The \textbf{quotient set} $\bbB^t\!/\!\iden$ is denoted $\ttThr$.
\end{itemize}
\end{definition}

In Fig.\;\ref{fig:a-thread-in-Pex}, the red occurrences of $\tv'$  correspond to an ascendant thread and the blue one to another. Their union constitute a (full) thread, that we denote $\theta_a$.\techrep{ \pierre{In Fig.\;\ref{fig:a-thread-in-Pex}, the red occurrences of $\tv'$  correspond to the ascendant thread $\set{(\epsi,1),(1,\epsi),(0\cdot 1,1)}$ and the blue occurrence of $\tv'$ to another ascendant thread (with only one element). The four colored occurrences of $\tv'$ correspond to a thread.}} Likewise, the green and the orange occurrence of $\tv$ respectively correspond to the negative and the positive part of a thread $\theta_b$. The unique purple occurrence of $\tv$ correspond to a singleton thread $\theta_c$.

The notation $\ttThr$ implicitly depends on $t$ and $\code{\cdot}$. The thread of $(a,c)\in \bbB$ is written $\thr{a,c}$ and we set:
\begin{center}
$\begin{array}{c@{\hspace{1.6cm}}c}
    \thepsi=\thr{\epsi,\epsi} & \thbot =\thr{\bbot} \\
    \text{``root thread''} & \text{``thread of emptiness''}
    \end{array}$
\end{center}
If $\thr{\p}=\theta$, we say that $\theta$ \textbf{occurs at biposition $\p$}, also written $\theta:\p$ or $\p:\theta$ \eg $\theta_a:(\epsi,1^2)$ or $\theta_a:(0,4\ct 1)$.
\techrep{ We will consider the case $\p_1:\th:\p_2$, where $\p_1$ and $\p_2$ are different bipositions.}


We consider now the extension of every other relation modulo $\iden$. Namely,
we write $\theta_1 \rrewc{a} \theta_2$ if $\exists \p_1,\p_2,~ \theta_1=\thr{\p_1},\ \theta_2=\thr{\p_2},\ \p_1\rewc{a} \p_2$. Thus, $\theta_1\rrewc{a} \theta_2$ iff $\theta_1:\p_1\rewc{a} \p_2:\theta_2$ for some $\p_1,\p_2$. In that case, we say that $\theta_1$ (resp. $\theta_2$) has been \textbf{left-consumed} (resp. \textbf{right-consumed}) at biposition $\p_1$ (resp. $\p_2$) \techrep{and that $\theta_1$ and $\theta_2$ are facing each other a pos. $a$ }\eg in Fig.\;\ref{fig:a-thread-in-Pex},
\begin{center}
  \vspace{-0.2cm}
  $\theta_b:(0^1\ct 1,8)\rrewc{0^2} (0^2\ct 8,\epsi): \theta_c$
\end{center}
  \techrep{Thus, $\theta_b$ and $\theta_c$ face each other at pos. $0^2$}.
We proceed likewise for $\rewto$, $\rewtt$, $\rewr$, $\rewdown$, $\rewbullet$, thus defining $\rrewto$, $\rrewtt$, $\rrewr$, $\rrewdown$, $\rrewb$. Notation $\rrewb^*$ denotes the reflexive transitive closure of relation $\rrewb$. \pierre{Remember again that these relations $\rrewto,\ldots$ concern for now \textit{candidate} derivations typing a given term $t$ }

\begin{corollary}
\label{corol:typability-and-thepsi}
  If $\thbot$ is not in the transitive closure of $\set{\thepsi}$ by $\rrewb$, then $t$ is typable in $\ttS$ (by means of a $\code{\cdot}$-derivation).
\end{corollary}

\begin{proof}
  Let $\bbBm=\set{\p\in \bbB\,|\, \thepsi\, \rrewb^*\, \thr{\p}}$ \ie $\bbBm$ is the union of the reflexive transitive closure of $\thr{\epsi}$ under $\rrewb$. If $\thepsi \rewbullet^* \thbot$ does not hold,  then $\bbBm$ satisfies the hypotheses of Proposition~\ref{prop:charac-bisupp}. So there exists a derivation $P$ s.t. $\bisupp{P}=\bbBm$ and thus, $t$ is typable.
\end{proof}

Since a $\ttS$-derivation contains $\pepsi$, by Proposition\;\ref{prop:charac-bisupp}, any $\code{\cdot}$-derivation typing $t$ will satisfy $\bbBm\subseteq \bisupp{P}$ and thus, $\bbBm$ is the \textbf{minimal bisupport candidate} for a $\code{\cdot}$-derivation typing $t$.\\

\noindent \textbf{Analogies with first order model theory:} 
Given $t\in \Lam$ and $\code{\cdot}$ and keeping in mind the intuition of bisupport candidates, let $\calT$ be the first order theory whose  set of constants is $\Thr{P}$, that one unary predicate symbol $\inBis$ (standing for ``is in bisupport'') and whose set of axioms is $\set{\inBis(\theta_1)\equiv \inBis(\theta_2)\;|\; \theta_1,\theta_2\in \Thr{P},\; \theta_1\rrewb \theta_1}\cup \set{\neg \inBis(\thbot)}$. Then Corollary\;\ref{corol:typability-and-thepsi} states that there \textit{exists} a $\code{\cdot}$-derivation $P$ typing $t$ iff $\calT$ is not contradictory: this is a sort of completeness result. Of course, it remains to be proved that $\calT_t$ is not contradictory (given any $t$). And this will be done using a technique closely associated to the $\lam$-calculus: a finite reduction strategy (presented in \S\;\ref{s:normalizing-chains}).




\section{Nihilating Chains}
\label{s:nihilating-chains}

We begin \S\;\ref{s:nihilating-chains} with a global description of the key steps leading to the fulfilment of Goal\;\ref{goal:code-S-typability} (every term is $\code{\cdot}$-typable) the final result (every term is $\scrR$-typable) and a presentation of the central notion of nihilating chain.\\

In the purpose of proving that every term is typable, we want to prove that, for all term $t$ and injection $\code{\cdot}:\bbN^*\rew \bbN\setminus\set{0,1}$, there is a $\code{\cdot}$-derivation typing $t$. According to Corollary\;\ref{corol:typability-and-thepsi}, we must show that $\thbot$ is not in the reflexive transitive closure of $\thepsi$ by $\rrewb$. A proof of $\thepsi \rrewb^* \thbot$ would involve a nihilating chain: 

\begin{definition}
  \label{def:nihilating-chain}\mbox{}
  \begin{itemize}
\item  A \textbf{chain} is a \textit{finite} sequence of the form $\theta_0 \rrewb \theta_1 \rrewb \ldots \rrewb \theta_{m}$. 
\item When $\theta_0=\thepsi$ and $\theta_{m}=\thbot$, the chain is said to be \textbf{nihilating}.
  \end{itemize}
\end{definition}

In order to apply Corollary~\ref{corol:typability-and-thepsi}, we must then prove that there is no nihilating chain. In other words, this corollary implies:

\begin{proposition}
  \label{prop:if-no-nihil-ch-then-comp-unsound}
If the nihilating chains do not exist, then every term is $\code{\cdot}$-typable, and thus, also $\scrR$-typable.
\end{proposition}

We proceed \textit{ad absurdum} and consider $\theta_0 \rrewb \theta_1 \rrewb \ldots \rrewb \theta_{m}$ with $\theta_0=\thepsi$ and $\theta_{m}=\thbot$. 
However, $\rrewb$ can be $\rrew,\, \werr,\, \rrewto,\, \rrewtt,~ \rrewr$ or $\rrewdown$. The structure of the proof is the following:

\begin{itemize}
\item We define (Definition\;\ref{d:syntactic-polarity}) the notion of \textit{polarity} for bipositions:  a biposition  is negative when it is created by an $\abs$-rule (modulo $\rewa$) and positive if not.
\item  The termination of a finite collapsing strategy (Sec.\;\ref{s:collapsing-strategy}) guarantees that positivity can be assumed to only occur at suitable places in the chain without loss of generality. In that case, we say that the\techrep{ nihilating } chain is \textit{normal} (Definition\;\ref{def:normal-chain}).
\item In normal chains, the different cases of subjugation interact well (\S\;\ref{s:interaction-normal-chains}), so that, from any normal chain, we may build another that begins with $\thepsi \rrewb \theta_1$\techrep{ (\S\;\ref{s:comp-unsound-almost-at-hand})}. This is  easily shown to be impossible,\techrep{ allowing us to conclude }\paper{ which entails } that nihilating chains do not exist and that every term is $\ttS$-typable.
\end{itemize}  

\ignore{   
\begin{itemize}
\item We define (Def.~\ref{d:syntactic-polarity}) the notion of \textit{syntactic polarity} for bipositions: a biposition  is negative when it was created by an $\abs$-rule (modulo $\rewa$) and positive if not.
\item  A collapsing strategy grants that positivity can occur at suitable places in the chain (collapsing hypothesis). In that case, we say that the nihilating chain is normal.
\item We can rule out the occurrences of $\rrew,~ \rrewr$ and $\rrewdown$ in a normal  chain.
\item \pierre{Relation $\werr$ ``commutes'' with $\rrewto$ and $\rrewtt$ in normal nihilating chains. This will allow us to assume that the chain is of the form:
\begin{center}
  $\theta_0 \rrewt \theta_1 \rrewt \ldots \rrewt \theta_{m'} \werr \theta_{m'+1}\werr \ldots \werr \theta_m$\end{center}}
\item \pierre{This implies that $\theta_{m'}=\thbot$. We can then assume that $m'=0$.}
\item \pierre{Definition of $\rewa$ and $\rewp$ shows then that $\theta_m$ cannot be $\thepsi$, which concludes the proof.}
\end{itemize}}

\subsection{Polarity and Threads}
\label{s:syntactic-polarity}

In this section, we\techrep{ describe the form of threads and ascendant threads (Definition\;\ref{def:threads-comp-unsound}) and } define the key notion of syntactic polarity of a biposition.

\techrep{
Notice that $\rewa$ is functional: if $\p_1 \rewa \p_2$, we write $\p_2=\asc(\p_1)$. Notice also that $\asc$ is injective. Thus, $\p_1 \idena \p_2$ iff $\exists i\geqslant 0,~ \p_2=\asc^i(\p_1)$ or $\p_1=\asc^i(\p_2)$.

Given a biposition $\p=(a,c)\in \bbB$, we call $a$ the \textbf{outer} and $c$ the \textbf{inner position} of $\p$. Since $\asc$ may only add the prefix 0 or 1 to $a$ and  add/remove the prefix 1 to $c$, by induction:

\begin{lemma}
\label{l:form-idena} 
If $(a_1,c_1) \idena (a_2,c_2)$ then $\exists a_3\in \set{0,1}^*,\,(a_2=a_1\ct a_3~ \text{or}~ a_1=a_2\ct a_3)$ and  $\exists i\geqslant 0,\,(c_2=1^i\cdot c_1~ \text{or}~ c_1=1^i\cdot c_2)$.
\end{lemma}}

We set, for all $\p\in \bbB$, $\Asc{\p}=\asc^i(\p)$, where $i$ is maximal (\ie $\asc^i(\p)$ is defined, but not $\asc^{i+1}(\p)$). Thus, $\Asc{\p}$ is the \textbf{top ascendant} of $\p$ \eg in Fig.\;\ref{fig:a-thread-in-Pex},\techrep{ $\Asc{\epsi,1^2}=(0^2\cdot 1,1)$ (\resp $\Asc{\epsi,1\cdot 4\cdot 1}=(0,4\cdot 1)$), that is,} the top red (\resp blue) occurrence of $\tv'$ is the top ascendant of the other ones (\resp one). A top ascendant is either located in an $\ax$-node (\eg the top red ascendant in Fig.\;\ref{fig:a-thread-in-Pex}) or in an $\abs$-node (\eg the blue one)\techrep{, since $\asc$ is total on $\app$-nodes}, motivating the notion of (syntactic) polarity\techrep{  for bipositions:}

\begin{definition}
\label{d:syntactic-polarity} \mbox{}
  \begin{itemize}
\item  Let $\p \in \bbB^t\setminus\set{\bbot}$ and $(a_0,c_0)=\Asc{\p}$. We define the \textbf{polarity} of $\p$ as follows: if $t(a_0)=x$ for some $x\in \TermV$, then we set $\Pol{\p}=\oplus$ and if $t(a_0)=\lx$, then we set $\Pol{\p}=\ominus$. We also set $\Pol{\bbot}=\omin$.
\item If $\thr{\p}=\theta$ and $\Pol{\p}=\oplus/\ominus$, we say that $\theta$ occurs positively/negatively at biposition $\p$.
\item If $\theta$ is left/right-consumed at $\p$ and $\Pol{\p}=\oplus$ (resp. $\Pol{\p}=\ominus$), we say that $\theta$ is left/right-consumed positively (resp. negatively) at biposition $\p$.  
\end{itemize}
\end{definition}

Then, we write for instance $\theta_1 \loplus\rrewc{a} \romin \theta_2$
to mean that $\theta_1$  is left-consumed positively and $\theta_2$ is right-consumed negatively in the $\app$-rule at position $a$.
In Fig.\;\ref{fig:a-thread-in-Pex}, the blue occurrence of $\tv'$  is negative and the red ones are positive.

\techrep{Since $\rewp$ also defines an injective function (out of $\thbot$) and $\p_1\rewp \p_2$ implies that $\p_1$ (in a $\lx$) and $\p_2$ (in an axiom or $\bbot$) do not have ascendants:

\begin{lemma}
  \label{l:form-iden} \mbox{}
\begin{itemize}
\item  For all $\p_1,\p_2\in \bbB^t$, $\p_1\iden \p_2,\, \Pol{a_1,c_1}=\oplus$ and $\Pol{a_2,c_2}=\omin$ iff $\Asc{a_1,c_1}\rewp \Asc{a_2,c_2}$.
\item For all $\p\in \bbB^t$, $\thr{\p}=\thbot$ iff $\Asc{\p}=(a_0,k\ct c_0)$ with $t(a_0)=\lx$ and $k\notin \Trl{a_0}$. 
\end{itemize}
\end{lemma}

Lemmas~\ref{l:form-idena} and \ref{l:form-iden} may be illustrated by Fig.\;\ref{fig:threads} (see also the more detailed Fig.\;\ref{fig:asc-pol-inv-informal-1}, p.\;\pageref{fig:asc-pol-inv-informal-1}), in which we only represent (with thick lines) the outer positions  in the occurrences of the threads.
\begin{figure}[h]
\begin{center}
  \begin{tikzpicture}

  \draw (0.5,4.2) node [right] {\small $\ovl{\ju{x:(5\ct S_5)}{x:S_5}}$};
  \draw [dotted] (2.75,4.1) --++ (0,-0.4) --++(0.4,-0.4)  --++ (0,-0.4) --++ (-1.1,-0.2) --++ (0,-0.2); 
\red{  \draw  [very thick] (2.9,4.1) --++ (0,-0.4) --++(0.4,-0.4);}
  \draw (3.2,3.65) node [right] {\parbox{1.5cm}{\fnsz positive occ. (red)}};
  
  
 \draw (1.6,3) node [right] {\fnsz polar inv.};
  
  \draw (0,2.3) node [right] {\small $\ju{C}{\lx.u:\sSksh \rew T}$};
\blue{  \draw [very thick] (1.85,2.25) --++(0,-0.4) --++ (0.2,-0.2) --++ (0,-0.3)--++ (0.3,-0.3) ;}

  \draw (0,1.3) node [right] {\fnsz \parbox{1.7cm}{negative occ. of the thread (blue)}};

      \draw (0.1,0.2) node [right] {\textbf{Thread with $+/-$ occ.}};

\transh{1}{
  
  \draw (7.3,3.7) node [above right] {\small  $\bbot$};
  \draw (7.3,3.7) node {$\bullet$};

  \draw [dotted] (5.2,3.1)  --++ (0,0.4) -- (7.3,3.7) ;
  \draw (6,3.8) node {\fnsz polar inv.};
  
  \draw (3.5,2.9) node [right] {\small $\ju{C'\!}{\!\!\ly.v\!:\!(\!S'_k\!)_{\kK'}\!\rew \!T'$}};
\blue{  \draw [very thick] (5.2,2.8) --++ (0,-0.4) --++ (0.3,-0.3) --++ (0,-0.2); }

\blue{  \draw [very thick] (6.5,2.3) --++ (0,-0.3) --++ (0.2,-0.2) --++ (0,-0.3) --++ (0.15,-0.15);}
  \draw [dotted] (7.3,3.7) --++(-0.3,-0.3) --++ (0,-0.7) --(6.5,2.3) ; 

\blue{  \draw [very thick] (7.1,2.2) --++ (0,-0.4) --++ (0.2,-0.2) --++ (0,-0.2); }
  \draw [dotted] (7.1,2.2) --++ (0.2,0.2) -- (7.3,3.7);

  \draw (3.5,1.8) node [below right] {\fnsz negative occ. (blue)};
  


    \draw (3.6,0.2) node [right] {\textbf{Thread of emptiness $\thbot$}};
}
\end{tikzpicture}
  \caption{Threads}
  \label{fig:threads}
\end{center}
\end{figure}

Thus, except the empty thread $\thbot$, Lemma~\ref{l:form-iden} means a thread can have \textit{at most} two ``connex components" (a series of positive ascendants and/or a series of negative ascendants respectively) \eg the green or the purple threads in Fig.\;\ref{fig:asc-pol-inv-informal-1}. All bipositions of $\thbot$ are negative.}
\techrep{\begin{remark} \mbox{}
\begin{itemize}
\item If the top occurrence of an ascendant thread is in an $\ax$-rule typing a variable $x$ s.t. $x$ is not bound in $t$, then the thread has one (positive) ``connex component" \eg the red thread in Fig.\;\ref{fig:asc-pol-inv-informal-1}. 
\item If a series of ascendants ends at the \textit{root} of an $\abs$-rule introduction $\lx$ (\ie $\Asc{a_0,c_0}=(a,\epsi)$ with $t(a)=\lx$), then the thread  has one (negative) ``connex component" \eg the blue thread in Fig.\;\ref{fig:asc-pol-inv-informal-1}. Notice that such a thread does \textit{not} contain $\bbot$.
  \end{itemize}
\end{remark}}

\techrep{
  Lemma \ref{l:form-iden} will be refined into Lemmas~\ref{l:form-iden-pos} and \ref{l:form-iden-neg}.\\}
\techrep{Since a consumed biposition does not have a descendant, Lemma~\ref{l:form-idena} and \ref{l:form-iden} imply the following lemma (which is the formal version of Observation\;\ref{obs:thread-at-most-2-consumed-informal}, p.\;\ref{obs:thread-at-most-2-consumed-informal}):

\begin{lemma}[Uniqueness of Consumption]
  \label{l:uniqueness-consumption}
  Let $\oast \in \set{\oplus,\omin}$ and $\theta \in \ttThr,~ \theta\neq \thbot$. Then, there is a most one $\theta'$ s.t. ($\theta \loast  \rrew \theta'$ or $\theta \loast\werr \theta'$).
\end{lemma}

In the figure above, consumption may occur at the bottom of each ``connex component".}

\subsection{Interactions in Normal Chains}
\label{s:interaction-normal-chains}


In \S\;\ref{s:interaction-normal-chains}, we present the notion of normal chain and explicit some interesting interaction properties that allow us to simplify/rewrite them.

As it has been discussed in \S\;\ref{ss:intro-infty-relev-typing} and \ref{ss:system-R-LICS18}, the possibility for a variable $x$ (of a redex or of a redex to be created later)  to be substituted in a reduction sequence is problematic. Intuitively, a biposition is negative when it was ``created" in an abstraction $\lx$ and that left-consumption is associated to left-hand sides of application. Thus, a negative left-consumption hints at the presence of redex (this intuition will be made more explicit in Sec.\;\ref{s:collapsing-strategy}). More precisely, it indicates the presence of what we will call a \textit{redex tower}. This suggests the following notion:

\begin{definition}
  \label{def:normal-chain}
  A nihilating chain is \textbf{normal} if no thread is left-consumed negatively in it (the chain does not contain a link of the form $\theta_i\lomin \rrew \theta_{i+1}$ or $\theta_i \werr\romin \theta_{i+1}$).
\end{definition}

Normal chains can be handled! We state some \textbf{interaction lemmas} below, that notably describe some commutations between stability relations.

\techrep{
Some relations of subjugation may roughly be related to the negative hand side of an arrow (\eg left-consumption destroys $\sSk$ in $t:\sSk\rew T$) or positive hand side (\eg ascendance) or the root of a type (\eg $\rewr$, from inner position $\epsi$ to 1). On the other hand, in a type, the track values give  a concise way to indicate nesting inside arrows \ie an argument track means a nesting inside the source of an arrow and track 1 means nesting in its target. This will help us now to better understand the possible forms of normal chains, yielding the series of \textbf{Interaction Lemmas} below.

Let $\p=(a,c)\in \bbB^t$.
According to Lemma~\ref{l:form-idena}, ascendance/descendance can only add or remove the prefix 1 to $c$. From that, we deduce that the only way to remove an argument track while visiting a thread is to pass from the negative hand side to the positive hand side of $\rewp$ (indeed, track $k$ is absent from the inner position of $\p_2$ in  $\p_1=(a,k\cdot c)\rewp (\pos{k},c)=\p_2$). Moreover, a biposition that is left-consumed must start with an argument track (a track $\geqslant 2$) by definition of $\rew$. Thus, if a thread is left-consumed \textit{positively}, the inner position of all its occurrences will contain an argument track (this is the meaning of Lemma~\ref{l:form-iden-pos}).

By definition of $\rewr$ and $\rewdown$, the target of $\rewr$ (resp. of $\rewdown$) must have an inner position equal to $1$ (resp. $\epsi$). This, with the previous observation (related to Lemma~\ref{l:form-iden-pos}), shows that a \textit{thread} that is consumed positively cannot be the target of $\rrewr$ or $\rrewdown$ (Lemma~\ref{l:rrew-rrewr-rrewdown} below).

It is easy to see  that $\rrewto$ and $\rrew$ commute (Lemma~\ref{l:rrew-rrewto} below). Using a similar reasoning as the one above (presence of an argument track in the inner position), we may also prove that $\rewtt$ and $\rrew$ commute when the involved thread is left-consumed positively in $\rrew$ (Lemma~\ref{l:rrew-rrewtt}). 
 
}



\techrep{\paragraph*{Formal Proofs of the Intersection Lemmas} 
 The end of \S\;\ref{s:interaction-normal-chains} is technical and is dedicated to the proof of the interaction lemmas described above. A summary of these lemmas can be found at the beginning of \S\;\ref{s:comp-unsound-almost-at-hand}.
  
\begin{lemma}  
  \label{l:rew-rewto}
If $(a_1,c_1)\iden (a_2,c_2)$, then $(a_1,c_1\ct k)\iden (a_2,c_2\ct k)$ for all $k$.
\end{lemma}

\begin{proof}
This holds when we replace $\iden$ by $\rewa$ and $\rewp$. The lemma follows by induction.
\end{proof}

\noindent Lemma~\ref{l:rew-rewto} is useful to prove Lemmas\;\ref{l:rrew-rrewto}, \ref{l:form-iden-pos} and \ref{l:form-iden-neg}.
} 
 
\begin{lemma}[Exchange of $\rrew$ and $\rrewto$]
  \label{l:rrew-rrewto}
If $\theta_1 \rrewto \theta_2$ and $\theta_2\rrew \theta_4$, then, $\exists \theta_3,\, \theta_1\rew \theta_3 $ and $\theta_3\rrewto \theta_4$.
\end{lemma}


\techrep{
\begin{proof}

Say $\theta_1:(a,c\cdot \ell)\rewto (a,c):\theta_2$ and $\theta_2:(a'\cdot 1,k\cdot c')\rew (a'\cdot k, c'):\theta_4$. By Lemma \ref{l:rew-rewto}, $(a,c\cdot \ell)\iden (a\cdot 1,k\cdot c'\cdot \ell)$. So we set $\theta_3=\thr{a'\cdot 1,k\cdot c'\cdot \ell}$, so that $\theta_1:(a'\cdot 1,k\cdot c'\cdot \ell)\rew (a'\cdot 1,k\cdot c'\cdot \ell):\theta_3$ and $\theta_3:(a'\cdot 1,k\cdot c'\cdot \ell)\rewto (a'\cdot 1,k\cdot c'):\theta_4$ as expected.
\end{proof}
}


\techrep{

 \begin{lemma}
  \label{l:form-iden-pos}
  Assume $\Pol{a\ct 1,k}=\oplus$ and $(a\ct 1,k)\iden (a_2,c_2)$.
  \begin{itemize}
  \item If $\Pol{a_2,c_2}=\oplus$, then $\exists i,\, c_2=1^i\ct k$.\\
    Moreover, for all $c\in\bbN^*,\, (a\ct 1,k\ct c)\iden (a_2,c_2\ct c)$.
  \item If $\Pol{a_2,c_2}=\omin$, let $(a_0,1^j\ct k\ct c)$ and $\ell=\code{a_0}$. Then $\exists j,\, c_2=1^j\ct \ell \ct 1^i \, k$.\\Moreover, for all $c\in \bbN^*,\, (a\ct 1,k\ct c)\iden (a_2,c_2\ct c)$. 
  \end{itemize}
  If $\Pol{a\ct 1,k\ct c}=\oplus$ for some $c$, then $\Pol{a\ct 1,k}=\oplus$ and the lemma can be applied.
\end{lemma}
}

\techrep{
  \begin{proof}
  The two first points come from Lemmas~\ref{l:form-idena},\;\ref{l:form-iden} and \ref{l:rew-rewto}.\\
Regarding the end of the statement: by induction on $i$, we also prove that, for all $k\geqslant 2,a,c\in \bbN$, $(a^i,c^i):=\asc^i(a,k)$ is defined iff $\asc^i(a,k\ct c)$ is and in that case, $\asc^i(a,k\ct c)=(a^i,c^i\ct c)$. Thus, if $(a_0,c_0)=\Asc{a,k}$, then $\Asc{a,k\ct c}=(a_0,c_0\ct c)$. This implies that  $\Pol{a,k}=\Pol{a,k\ct c}$.
\end{proof}


\noindent Lemma~\ref{l:form-iden-pos} is useful to prove Lemmas~\ref{l:rrew-rrewr-rrewdown} and \ref{l:rrew-rrewtt}: 
}

\begin{lemma}[Elimination of $\rrewr$ and $\rrewdown$]
\label{l:rrew-rrewr-rrewdown}
If $\theta \loplus \rrew \theta'$, there is no $\theta_0$ s.t. $\theta_0\rrewr \theta$ or $\theta_0 \rrewdown \theta$.
\end{lemma}

   \techrep{
     \begin{proof}
       Say $\theta:(a\cdot 1,k\ct c)\rew (a\cdot k, c):\theta'$ (with necessarily $k\geqslant 2$). We assume \textit{ad absurdum} that $\theta_0\rrewr \theta$ or $\theta_0\rrewdown \theta$.

       The first case implies that $\theta=\thr{a',\epsi}$ and the second one that $\theta=\thr{a',1}$ for some $a'$. But this is impossible according to Lemma~\ref{l:form-iden-pos}.
\end{proof}
}

\begin{lemma}[Exchange of $\loplus \rrew$ and $\rrewtt$]
\label{l:rrew-rrewtt}
If $\theta_1\rrewtt \theta_2$ and $\theta_2\loplus\rrew \theta_4$, then, $\exists \theta_3,\, \theta_1\loplus \rrew \theta_2 $ and $\theta_3\rrewtt \theta_4$.
\end{lemma}

  \techrep{
    \begin{proof}
      Say $\theta_1:(a, c\cdot \ell)\rewtt (a, c \cdot 1):\theta_2$ and $\theta_2:(a'\cdot 1,k\cdot c')\loplus\rew (a'\cdot k, c'):\theta_4$.
           
  By Lemma~\ref{l:form-iden-pos}, $c\cdot 1=1^i\cdot k\cdot c'$ or $c\cdot 1=1^j\cdot \ell' \cdot 1^i\cdot k\cdot  c'$ for some $i,j,\ell'$. Thus, $c'=c'_0\cdot 1$ for some $c'_0$.

  Also by Lemma~\ref{l:form-iden-pos} (last statement), we have
$(a,c \cdot \ell) \iden (a'\cdot 1,k\cdot c'_0\cdot \ell)$.

  We set then $\theta_3=\thr{a'\cdot k,c'_0\cdot \ell}$, so that
  $\theta_1:(a'\cdot 1,k\cdot c'_0\cdot \ell)\rew (a'\cdot k,c'_0\cdot \ell):\theta_3$ and $\theta_3:(a'\cdot k,c'_0\cdot \ell)\rewtt (a'\cdot k,c'_0\cdot 1):\theta_4$, as expected.
\end{proof}
}

  \techrep{
  We recall from Lemma~\ref{l:form-iden} that, if the thread of a biposition $\p=(a,c)$ is $\thbot$, then $\p$ is negative and was ``created" in an $\abs$-rule (thus, $c$ must have an argument track).
  In that case, it is easy to see, that when we postfix anything to the inner position $c$ (\ie we use relation $\rewto$), we get a biposition $\p'=(a,c\ct c')$ whose thread is also $\thbot$. Using the presence of an argument track $k$ (as in the positive case discussed above), we can also prove  that $\thbot$ is stable under $\rrewtt$ and cannot be the target of $\rrewr$ and $\rrewdown$. All this is captured by Lemma~\ref{l:rbot}. 
  }

\techrep{
  \begin{lemma}
  \label{l:form-iden-neg}
  Assume $(a,1^i\ct k)\iden \bbot$ with $k\geqslant 2$. Then, for all $c\in \bbN^*,\,  (a,1^i\ct k\ct c)\iden \bbot$.\\
  Moreover, if $(a,1^i\ct k\ct c)\iden \bbot$ with $k\geqslant 2$, then $(a,1^i\ct k) \iden \bbot$, and we can apply the lemma.
\end{lemma}
}

 \techrep{
   \begin{proof}
Let $(a_0,1^{i_0} \cdot k)=\Asc{a,1^i \cdot k}$ and for all $j$ such that it is defined,  $(a^i,1^{\tti(j)}\cdot k)=\asc^j(a,1^i\cdot k)$.

By induction on $j$ and Lemma~\ref{l:form-idena}, for all $c\in \bbN^*$, $\asc^j(a,1^i\cdot k\cdot c)$ is defined iff $\asc^j(a,1^i\cdot k \cdot c)$ is, and in that case, $\asc^j (a,1^i\cdot k\cdot c)=(a^j,1^{\tti(j)} \cdot k\cdot c)$. 

Thus, $\Asc{a,1^i\cdot k\cdot c}=(a_0,1^{i_0}\cdot k\cdot c)$.

Assume moreover that $(a,1^i\cdot k)\iden \bbot$. Then $\Asc{a,1^i\cdot k}\rewp \bbot$. This implies that $t(a_0)=\lx,~ i_0=0$ and $k\notin \Trl{a_0}$. Thus, $\Asc{a,1^i\cdot k\cdot c}=(a_0,k\cdot c)\rewp \bbot$. So $(a,1^i\cdot k\cdot c)\iden \bbot$.

Now, assume instead that $(a,1^i\cdot k \cdot c) \iden \bbot$. Then  $\Asc{a,1^i\cdot k\cdot c}\rewp \bbot$. This implies that $t(a_0)=\lx,~ i_0=0$ and $k\notin \Trl{a_0}$. By the above induction, $\Asc{a,1^i\cdot k}=(a_0,k)\rewp \bbot$. So $(a,1^i\cdot k)\iden \bbot$.
\end{proof}

\noindent Lemma~\ref{l:form-iden-neg} is useful to prove Lemma~\ref{l:rbot}:}

\begin{lemma}[The Thread of Emptiness in Action]
\label{l:rbot} \mbox{}
  \begin{itemize}
  \item If $\thr{\p}=\thbot$, then $\Pol{\p}=\ominus$.  
  \item If $\theta \rrewto \thbot$ or $\theta \rewtt \thbot$, then $\theta=\thbot$.
  \item We cannot have $\theta \rrewr \thbot$ or $\theta \rrewdown \thbot$.
  \end{itemize}  
\end{lemma}

  \techrep{
\begin{proof}\mbox{}
\begin{itemize}
\item  The implication $\thr{\p}=\thbot \Rightarrow \Pol{\p}=\omin$ comes from Lemma\;\ref{l:form-iden}.
\item  Assume $\theta\rrewto \thbot$ \eg say $\theta:(a,c\cdot \ell)\rewto (a,c):\theta$. By Lemma\;\ref{l:form-iden}, we have $(a,c)\iden (a_0,k\cdot c_0)$ with $t(a_0)=\lx,\, k\notin \Trl{a_0},\, k\geqslant 2$. By Lemma\;\ref{l:rew-rewto}, $(a,c\cdot \ell)\iden (a_0,k\cdot c_0\cdot \ell)$. But $(a_0,k\cdot c_0\cdot \ell)\rewp \bbot$. So $\thr{a,c\cdot \ell}=\bbot$ \ie $\theta=\thbot$.
\item 
  Assume $\theta \rrewtt \thbot$ \eg say $\theta:(a,c\cdot \ell)\rewtt (a,c\cdot 1):\thbot$.
  
  By Lemma~\ref{l:form-iden}, we have $(a,c\cdot 1)\idena (a_0,k\cdot c_0)$ with $t(a_0)=\lx,\, k\notin \Trl{a_0},\, k\geqslant 2$.

  By Lemma\;\ref{l:form-idena}, $c\cdot 1=1^i\cdot k\cdot c_0'\cdot 1$ for some $i,c_0'$.

  By Lemma\;\ref{l:form-iden-neg}, $(a,1^i\cdot k)\iden \bbot$ and then $(a,c\cdot \ell)=(a,1^i\cdot k\cdot c_0' \cdot \ell)\iden \bbot$.
  Thus, $\theta=\thbot$.
\item If $\p:\thbot$, by Lemmas\;\ref{l:form-idena} and \ref{l:form-iden}, $\p=(a,1^i\cdot k\cdot c)$. Thus, $\p \rewr \p'=(a',1)$ or $\p \rewdown \p'=(a',\epsi)$ is impossible.
\end{itemize}
\end{proof}
}

  When the considered nihilating chain is not normal, the arguments involving the presence of argument tracks fail and it is not difficult to find counter-examples to the commutation of $\rrew$ and $\rrewtt$ or to Lemma~\ref{l:rrew-rrewr-rrewdown}\techrep{.}\paper{: the behaviour of normal chains is far better understandable.\\}

  \ignore{
for instance, if $(\epsi,\epsi):\thepsi:(2,\c$
\\}

\noindent \textbf{Goal 4 (almost) at Hand}
Using Lemmas~\ref{l:rrew-rrewto},\,\ref{l:rrew-rrewr-rrewdown},\,\ref{l:rrew-rrewtt},\,\ref{l:rbot}, it is not difficult to define an algorithm taking a normal nihilating chain as input and outputting a nihilating chain of the form $\thepsi=\theta_0 \loplus \rrew \theta_1\loplus \rrew ... \loplus\rrew \theta_{\ell}=\thbot$.\\
Then, one proves that there is not $\theta$ such that $\thepsi \loplus \rrew \theta$. This implies that there is no chain of the form $\thepsi=\theta_0 \loplus \rrew \theta_1\loplus \rrew ... \loplus\rrew \theta_{\ell}=\thbot$, and then, that there is no normal nihilating chain:

\begin{proposition}
\label{prop:no-normal-chain}
There is no \textit{normal} nihilating chain.
\end{proposition}

Proposition\;\ref{prop:no-normal-chain} \textit{almost} proves that every term  is $\code{\cdot}$-typable (by Proposition\;\ref{prop:if-no-nihil-ch-then-comp-unsound}). Almost, because only the non-existence of normal nihilating chains is ensured for now (and not that of nihilating chains in general). The only point that will remain to be verified is that normal nihilating chains can be considered without loss of generality (which is the object of \S\;\ref{s:normalizing-chains}).

\techrep{
\subsection{Complete Unsoundness (almost) at Hand}
\label{s:comp-unsound-almost-at-hand}

Now, using the Interaction Lemmas~\ref{l:rrew-rrewto},\,\ref{l:rrew-rrewr-rrewdown},\,\ref{l:rrew-rrewtt},\,\ref{l:rbot}, we may build  (by the Claim below), from any \textit{normal} chain, a nihilating chain of the form $\thepsi=\theta_0 \loplus \rrew \theta_1\loplus \rrew ... \loplus\rrew \theta_{\ell}=\thbot$. Then we prove that a chain of this latter form cannot exist. More concretely, this \textit{almost} proves that system $\ttS$ is \pierre{completely unsound (by Proposition\;\ref{prop:if-no-nihil-ch-then-comp-unsound}).} Almost, because only the non-existence of normal nihilating chains will be proved in \S\;\ref{s:comp-unsound-almost-at-hand}. The only point that will remain to be verified is that normal nihilating chains can be considered without loss of generality (which is the object of \S\;\ref{s:normalizing-chains}).

The Interaction Lemmas can be summarized as follows: \label{summary-interaction-lemmas}
\begin{itemize}
\item $\theta_1 \rrewto \theta_2 \rrew \theta_3$ can be replaced by $\theta_1 \rrewb \theta'_2 \rrewto \theta_3$  (Lemma\;\ref{l:rrew-rrewto}).
\item $\theta_1  \rrewto  \theta'_2 \loplus \rrew \theta_3$  can be replaced by
  .$\theta_1 \loplus \rrew  \theta_2 \rrew \theta_3$ (Lemma\;\ref{l:rrew-rrewtt}).
\item $\theta_1 \rrewdown \theta_2 \loplus \rrew \theta_3$ and $\theta_1  \rrewr \theta_2 \loplus \rrew \theta_3$  is impossible (Lemma\;\ref{l:rrew-rrewr-rrewdown}).
\item $\theta \rrew_{\mathtt{t}1/2} \thbot$ imply $\theta=\thbot$ and $\theta {\rrew}_{\mathtt{abs}/\mathtt{down}} \thbot$ is impossible (Lemma\;\ref{l:rbot}).
\end{itemize}

Moreover, the Consumption Lemma (Lemma\;\ref{l:uniqueness-consumption}) can also be seen as an interaction lemma: it entails that $\theta_1\werr \roplus \theta_2 \loplus \rrew \theta_3$ implies $\theta_1=\theta_3$ \ie nothing happens when $\werr \roplus$ is followed by $\rrew$. These observations are enough to prove that normal chains do not exist.

In particular, $\thepsi:(a\ct 1,k\ct c)\rew (a\ct k,c):\theta_1$ for some $(a,c)\in \bbB,k\in \bbN$. We prove now that $\thepsi \rew \ldots$ is impossible: by Lemma~\ref{l:form-idena}, $\asc^i(\epsi,\epsi)=(a_0,1^{i_0})$ for some $a_0,i_0$.
\begin{itemize}
\item Assume $\Pol{\epsi,\epsi}=\oplus$: since $\Pol{a\ct 1,k\ct c}=\oplus$, by Lemma~\ref{l:form-idena}, $(\epsi,\epsi)\iden (a\ct 1,k\ct c)$ is impossible. 
\item Assume $\Pol{\epsi,\epsi}=\omin$: $\exists \p,\, (a_0,1^{i_0})\rewp \p$ is impossible by definition of $\rewp$. Thus, $(\epsi,\epsi)\iden \p$ and $\Pol{\p}=\oplus$ is impossible.
\end{itemize}
So, there are no nihilating chains containing only $\loplus \rrew$. And thus, there are no normal nihilating chains:

\begin{proposition}
\label{prop:no-normal-chain}
There is no \textit{normal} nihilating chain.
\end{proposition}

In the next section, we describe a method to build, from any given nihilating chain (if such a chain existed), a normal chain. By the above proposition, it will \textit{ad absurdum} entail that there are no nihilating chains. But before that, in order to be valid, the reasoning and the Proposition above need the following Claim to be proved.\\ 

\noindent \textbf{Claim:} From any \textit{normal} chain, we can build a nihilating chain of the form $\thepsi=\theta_0 \loplus \rrew \theta_1\loplus \rrew ... \rrew \theta_{\ell}=\thbot$.\\

\begin{proof}
Let $\calC$ be a \textit{normal} nihilating chain. Thus, $\calC$ is of the form $\theta_0 \rrewb \theta_1 \rrewb \ldots \rrewb \theta_{m}$ with $\theta_0=\thepsi$ and $\theta_m=\thbot$. We are going to apply the algorithm below to $\calC$. \\

\noindent \textit{Description of the algorithm:}  
  We first apply commutations to $\calC$. During this process, we implicitly use Lemma~\ref{l:uniqueness-consumption} every time that some relation $\theta_{i-1}  \werr \roplus \theta_i \loplus \rrew \theta_{i+1}$ pops up somewhere in the chain: since this implies $\theta_{i-1}=\theta_{i+1}$, we just remove $\theta_i$ from the chain.\\
  
The first stage of the algorithm consists in using Lemmas\;\ref{l:rrew-rrewto} and \ref{l:rrew-rrewtt} as many times as possible, so that $\theta_{i-1}\rewto \theta_i \loplus \rrew \theta_{i+1}$ or  $\theta_{i-1}\rewtt \theta_i \loplus \rrew \theta_{i+1}$ does not occur anymore in the chain.\\

Then, we are interested in the last relation $\theta_{m-1}\rrewb \theta_m=\thbot$ of the chain.

By Lemma~\ref{l:rbot}, if it is $\rrewto$ or $\rrewtt$, then $\theta_{m-1}=\thbot$ and we discard $\theta_m$ (second stage of the algorithm).

Also by Lemma~\ref{l:rbot}, the last relation can neither be $\rrewr$ nor $\rrewdown$.

Since $\thbot$ only occurs negatively (Lemma~\ref{l:rbot}) and the chain is normal, then the last relation cannot be $\werr$.

So the last relation of the chain after the two stages of the algorithm is $\theta_{\ell-1}\loplus \rrew \theta_\ell=\thbot$ (with $\ell \geqslant m$).
\\

Assume now \textit{ad absurdum} that the chain contains a relation that is not $\rrew$.

We investigate the maximal $k$ such that relation $\theta_{k-1}\rrewb \theta_k$ is \textit{not} $\rrew$.

We have just proved that $k<\ell$, so that we have $\theta_{k-1}\rrewb \theta_k \loplus \rrew \theta_{k+1}$. By Lemma\;\ref{l:rrew-rrewr-rrewdown}, this relation cannot be $\rrewr$ or $\rrewdown$.

By the first stage of the algorithm, it can neither be $\rrewto$ nor $\rrewtt$. So we must have $\theta_{k-1} \werr \roplus \theta_k \loplus \rrew \theta_{k+1}$.

So we must have $\theta_{k-1} \werr \theta_k$. Since the chain is normal, we have $\theta_{k-1}\werr \roplus \theta_k \loplus \rrew \theta_{k+1}$.

By the procedure described above, it is impossible (we could remove $\theta_k$ from $\calC$), so when the algorithm is completed, we have a chain of the form $\theta_0\rrewb \theta_1\rrewb \ldots \rrewb \theta_\ell=\thbot$.
\end{proof}
}

\section{Normalizing Nihilating Chains}
\label{s:normalizing-chains}

\paper{
  \subsection{Residuation and Subjugation}
  \label{ss:residuation-subjugation}
In this section, we prove that negative left-consumption in a nihilating chain can be avoided (without loss of generality).
  By Proposition\;\ref{prop:if-no-nihil-ch-then-comp-unsound}, this will allow us to prove that every term is typable. The fact that system $\ttS$ is relevant, non-idempotent, rigid and syntax-directed entails that if $P \tri \juCtt$, then, there is a unique derivation $P'\tri \juCtpt$ obtained from $P$ by subject reduction (thus, subject reduction is \textit{deterministic} in system $\ttS$). Moreover, intuitively, every part of $P'$ comes from a part of $P$ and so, every position and right biposition of $P'$ can be thought as the \textbf{residual} of a unique position or (right) biposition of $P'$. We do not give details (that be found in \S\;IV and Fig.\;1  in \cite{VialLICS17}), but this induces a partial injective function $\Res_b$ from the right bisupport of $P$ to that of $P'$. The function $\Res_b$ turns out to be compatible with subjugation:
}
\techrep{
In this section, we prove that negative left-consumption in a nihilating chain can be avoided (without loss of generality). By Proposition\;\ref{prop:if-no-nihil-ch-then-comp-unsound}, this will allow us to prove the theorem of complete unsoundness. For that, the notions of quasi-residual of a biposition and of residuals of threads are fundamental and we present them. The main technique to be implemented here was presented from a perspective in \S\;\ref{s:collapsing-technique-informal}.


\begin{figure*}
  \input{subjRedTreeFig.tex}
  \caption{Residuals and Quasi-Residuals (copy of Fig.\;\ref{fig:SR-residuals}, p.\;\pageref{fig:SR-residuals})}
  \label{fig:SR-residuals-c-unsound}
  $\phd$\\[-7ex]
\end{figure*}

Remember from \S\;\ref{s:one-step-sr-se}  that
system $\ttS$ enjoys both subject reduction and expansion, meaning typing is  invariant under (anti)reduction in system $\ttS$, and that moreover, in system $\ttS$, subject reduction is processed in a deterministic way thanks to the tracks (contrary to system\;$\scrRo$, \S\;\ref{s:non-conf-Ro}).

The determinism of reduction allowed us to define in \S\;\ref{s:sr-proof} residuation for some positions and bipositions. Residuation (\S\;\ref{s:res-lam}) is a way to describe how positions move or are destroyed during reduction, and the residuation relation is  functional for system $\ttS$ (in contrast with the $\lam$-calculus) because system $\ttS$ is linear and forbids duplication (Remark\;\ref{rmk:res-in-S-vs-res-in-lam}).

We extend now residuation to quasi-residuation for system $\ttS$, as it was done for the $\lam$-calculus in \S\;\ref{s:res-lam}. Quasi-residuation for positions does not preserve the labelling.
For that, we reuse Fig.\;\ref{fig:SR-residuals} which is copied into Fig.\;\ref{fig:SR-residuals-c-unsound}.
We also refer to \S\;\ref{s:sr-proof} for a description of the moves of the symbols.

\subsection{Quasi-Residuals}
\label{s:residuals-c-unsound}



In this section, we define residuation for threads and state some of its most important properties.

From \S\;\ref{s:comp-unsound-almost-at-hand}, we know that we only have to escape the case of negative left consumption, to grant that nihilating chains do not exist and thus ensuring that every term is typable (discussion beginning Sec.\;\ref{s:nihilating-chains}). To achieve this purpose, the main tool is a normalizing reduction strategy, called the \textbf{collapsing strategy}, which destroys threads that are left-consumed negatively and allows us to build a normal chain from any nihilating chain. This reduction strategy is finite, despite the fact that no form of normalization is not ensured in system $\ttS$. But to be formulated, it needs the notion of residual for threads, which can be ensured only if it is well-behaved  for bipositions.





Residuals were defined for most positions, but not for all of them \eg $a\cdot 1$, that corresponds to the abstraction of the redex, is destroyed during reduction and does not have a residual. When $(a,c)\in \bisupp{P}$ and $a'=\Res_b(a)$ is defined, we set $\Res_b(a,c)=(a',c)$, thus defining the residuals of bipositions.

The partial function $\Res_b$ on bipositions may be quite naturally extended to a total function $\QRes_b$ (standing for \textbf{quasi-residual}) from $\bbB^t$ to $\bbB':=\bbB^{t'}$.
For instance, the abstraction of the redex, typed with a type of the form $\sSksh\rew T$, is destroyed, but the types $S_k$ are still present in the typing of $\rsx$ (they occur as types of subterm $s$).\\

}
\techrep{
  \paragraph*{Formal Definitions of Quasi-Residuals}
  As in \S\;\ref{s:sr-proof}, or what concerns (quasi-)residuals, metavariable $a$ will now denote only positions in $\bbA$ s.t. $\ovl{a}=b$. Metavariables $\al$ and $c$ range over $\bbN^*$. For instance, $\al\neq a$ means that $\ovl{a}\neq b$.
If $k\in \Trl{a}$, $a_k$ is the unique position s.t. $\pos{k}=a\cdot 10 \cdot a_k$ (see $a_2,\, a_3,\, a_3$ in Fig.\;\ref{fig:SR-residuals-c-unsound}).
%
Let us have a look at the positions in $t$: positions $a$ and $a\cdot 1$ point to the root and the abstraction of the redex, position $a\cdot 10$ points to the root of $r$, position $a\cdot k$ (with $k\geqslant 2$) points to the root of $s$, position $a\cdot 10 \cdot a_k$ (with $k\in \Trl{a}$) points to an occurrence of $x$. 
Then, thoses positions $a,\, a\cdot 1$ and $a\cdot 10\cdot a_k$, that respectively point to the root, the abstraction or the variable of the redex are considered to be \textbf{destroyed} when the redex is fired. The associated bipositions will not have a proper residual, but we define below their \textit{quasi-residuals}. This is an adaptation of the notion of quasi-reduation for the positions of $\lam$-terms (\S\;\ref{s:res-lam}).


First, let us remember from \S\;\ref{s:sr-proof} how the proper residual for positions and bipositions are defined:
\begin{itemize}  
\item \textit{Out of the redex: } If $\al \ngeqslant a$, then $\al$ is not in the redex. We set $\Res_b(\al)=\al$.  
\item \textit{Inside $r$:}
  Position $a\cdot 10 \cdot \al\in \bbB$ (paradigm $\heartsuit$) has a residual (except when $\al=a_k$ for some $k$) and should become $a\cdot \al $ after reduction: we set $\Res_b(a\cdot 10\cdot \al)=a\cdot \al$ for $\al \neq a_k$.
\item \textit{Inside some argument derivations:} Assume $k\in \Trl{a}$.
  Argument derivation at $a\cdot k$ will replace $\ax$-rule typing at position $a\cdot 10\cdot a_k$ (which is destroyed).  So its position after reduction will be $a\cdot a_k$. More generally, the $a\cdot k\cdot \al\in \bbB$ (paradigm $\clubsuit$) will be found at $a \cdot a_k\cdot \al$ after reduction. We set then $\Res_b(a\cdot k\cdot \al)=a\cdot a_k\cdot \al$ when $k\in \Trl{a}$.
\item  \textit{Some bipositions:} We set $\Res_b(\al,c)=(\Res_b(\al),c)$ when $\Res_b(\al)$ is defined. 
\end{itemize}

\noindent Now, we define the quasi-residuals for all bipositions of $\bbB$ and some positions of $\bbA$.
\begin{itemize}
\item \textit{Extension of $\Res_b$:} $\QRes_b(a)=\Res_b(a)$ and $\QRes_b(a,c)=(\Res_b(a),c)$ whenever $\Res_b(a)$ is defined.
\item \textit{Root of the Redex:} We have $\ttT(a)=\ttT(a\cdot 10)$ and, for all $c\in \bbN^*$, $(a,c)\idena (a\cdot 10,c)$, so we set $\QRes_b(a)=\QRes_b(a\cdot 10)=a$ and  $\QRes_b(a,c)=\QRes_b(a\cdot 10,c)=(a,c)$.
\item \textit{Variable of the redex:} Let $k\in \Trl{a}$. We have $\ttT(a\cdot 10 \cdot a_k)=\ttT(a\cdot k)$ (\eg $\ttT(a\cdot 10\cdot a_2)=\ttT(a\cdot 2)=S_2$ in Fig.\;\ref{fig:SR-residuals-c-unsound}). So we set $\QRes_b(a\cdot 10 \cdot a_k)=\Res_b(a\cdot k)=a\cdot a_k$
\item \textit{Abstraction of the redex:}
  \begin{itemize}  
  \item 
    \textit{Target of the Arrow Type:} We have $\ttT(a\cdot 1)\rstr{1}=\ttT(a\cdot 10)$ and actually, $(a\cdot 1,c)\rewa (a\cdot 10,c)$ for all $c\in \bbN^*$, so we set $\QRes_b(a\cdot 1,c)=\QRes_b(a\cdot 10,c)$.
  \item \textit{Source of the Arrow type (1):} If $k\in \Trl{a}$, then $\ttT(a\cdot 1)\rstr{k}=\ttT(a\cdot k)$ (\eg $\ttT(a\cdot 1)\rstr{7}=\ttT(a\cdot 7)=S_7$ in Fig.\;\ref{fig:SR-residuals-c-unsound}) and actually, $(a\cdot 1,k\cdot c)\rewb{a} (a\cdot k,c)$ for all $c\in \bbN^*$. So we set $\QRes_b(a\cdot 1,k\cdot c)=\Res_b(a\cdot k,c)=(a\cdot a_k,c)$.
  \item \textit{Source of the Arrow type (2):} If $k\geqslant 2$ and $k\notin \Trl{a}$, then  $(a\cdot 1,k\cdot c)\rewp \bbot$ and there is no $\ax$-rule $\geqslant a$ typing $x$ using track $k$.
    So we say that $(a\cdot 1,k\cdot c)$ is \textbf{nihilated} after reduction. We set $\QRes_b(a\cdot 1,k\cdot c)=\bbot$.
    \item \textit{Root of the type:} To ensure Lemma~\ref{l:ref-red}, it is convenient to set $\QRes_b(a\cdot 1,\epsi)=\QRes_b(a\cdot 1,1)=(a,\epsi)$.
  \end{itemize}
\item \textit{Nihilated argument derivations:}
  Assume $k\notin \Trl{a},\,k\geqslant 2$, then there is no $\ax$-rule typing $x$ using axiom track $k$. So argument $P_k$ is not moved inside $r$ but nihilated after reduction. We then set $\QRes_b(a\cdot k\cdot \al,c)=\bbot$ and also say that $(a\cdot k\cdot \al,c)$ has been \textbf{nihilated} after reduction.
\end{itemize}

\begin{remark}
  \mbox{} \label{rmk:qres-not-preserve-labels}
\begin{itemize}
\item  We refer the quasi-residual of $a\cdot 10$ as $\QRes_b(a\cdot 10)$ and not as $\Res_b(a\cdot 10)$ because, if $t(a\cdot 10)=x$, then $a\cdot 10$ has no proper residual.
\item Thus, $\QRes_b$ is a total function on bipositions and $\Res_b$ is a \textit{partial injective} function. Moreover, $t'(\QRes_b(\al))=t(\al)$ is not true in general \ie quasi-residuation does not preserve labelling (contrary to $\Res_b$, Lemma\;\ref{l:res-S-preserves-label}), as in the case of the pure $\lam$-calculus (\S\;\ref{s:res-lam}).
  \end{itemize}
\end{remark}
}
\techrep{
We also consider the relations $\rewa,\, \rewp,\, \rewbullet,\ldots$ defined w.r.t. $\bbB'$ (we do not distinguish them graphically from those of $\bbB$) and $\code{\cdot}'$, an injective function from the set of leaves of $\bbA^{t'}$ to $\bbN\setminus \set{0,1}$ naturally defined from $\code{\cdot}$: we set $\code{\al'}'=k$ if there is $\al\in \bbA$ s.t. $t(\al)=y\neq x$ and $\Res_b(\al)=\al'$. We check that $\al'\mapsto \code{\al'}'$ is still an injective function. We set $\posp{k}=\al'$ if there is $\al'$ s.t. $\codep{\al'}=k$. Thus, $\posp{k}=\Res_b(\pos{k})$.  Notice that $\Res_b(\al)\in \bbA'$ and $\QRes_b(\p)\in \bbB'$. We also define $\asc',\ \ttThr'$ and $\thrp{\cdot}$.

Relations $\rewa$ and $\rewp$ and thus $\iden$ are compatible with reduction.
}\techrep{
 By case analysis (still guided by Fig.\;\ref{fig:SR-residuals-c-unsound}):
\begin{itemize}
\item Assume $\p_1\rewa \p_2$ or $\p_1\rewp \p_2$ and $\p_1$ is nihilated. Then $\QRes_b(\p_1)=\QRes_b(\p_2)=\bbot$. We assume below that $\p_1$ is not nihilated.
\item Assume $\p_1=(\al,c) \rewa \p_2$.
  If $\al \neq a,a\cdot 1$, then $\QRes_b(\p_1)\rewa \QRes_b(\p_2)$. If $\al=a,a\cdot 1$, then $\QRes_b(\p_1)=\QRes_b(\p_2)$.
\item If $\p_1=(\al,k \cdot c) \rewp \p_2$ with $\p_2\neq \bbot$. If $\al \neq a\cdot 1$, then $\Res_b(\p_1)\rewp \Res_b(\p_2)$. If $\al =a\cdot 1$, then $\QRes_b(\p_1)=\QRes_b(\p_2)$.
\item Assume $\p=(\al,k\cdot c)\rewp \bbot$. If $t(\al)=\lx$, then $\QRes_b(\p)=\bbot$. If $t(\al)=\ly\neq \lx$, then $\QRes_b(\p)\rewp\bbot$.
\end{itemize}
This entails, by induction on $\iden$:}

\begin{lemma}
\label{l:iden-res}
  If $\p_1\iden \p_2$, then $\QRes_b(\p_1)\iden \QRes_b(\p_2)$.
\end{lemma}

\noindent 
This Lemma allows us to define (quasi-)residuals for \textit{threads}. We set $\QRes_b(\theta)=\thrp{\QRes_b(\p)}$ for any $\p:\theta$\paper{ (where $\thrp{\cdot}$ denotes threads in $\bbB'$)}. By case analysis, we have:
  \techrep{
\begin{itemize} 
\item If $\p_1\rewbullet \p_2$ and $\p_{2}$ is nihilated, then $\QRes_b(\p_1)=\QRes_b(\p_2)=\bbot$. We assume below that $\p_1$ is not nihilated.
\item Assume $\p_1=(\al\cdot 1,k\cdot c) \rew (\al\cdot k,c )=\p_2$. If $\al \neq a$, then $\QRes_b(\p_1)\rew \QRes_b(p_2)$ and if $\al =a$, then $\QRes_b(\p_1)=\QRes_b(\p_2)$.
\item Assume $\p_1=(\al\cdot k, c) \wer (\al \cdot 1,k\cdot c )=\p_2$. If $\al \neq a$, then $\QRes_b(\p_1)\wer \QRes_b(p_2)$ and if $\al =a$, then $\QRes_b(\p_1)=\QRes_b(\p_2)$.
\item Assume $p_1\rewto \p_2=(\al,c)$. If $\p_2\neq (a\cdot 1,\epsi)$, $\QRes_b(\p_1)\rewto \QRes_b(\p_2)$. If $\p_2=(a\cdot 1,\epsi)$, $\QRes_b(\p_1)=\QRes_b(\p_2)=(a,\epsi)$.
\item Assume $\p_1\rewtt \p_2=(\al,c\cdot 1)$. If $\p_2\neq (a\cdot 1,1)$, then $\QRes_b(\p_1)\rewtt \QRes_b(\p_2)$. If $\p_2=(a\cdot 1,1)$, then $\QRes_b(\p_1)\rewdown \QRes_b(\p_2)=(a,\epsi)$ 
\item Assume $\p_1=(\al,\epsi) \rewr (\al,1)=\p_2$. If $\al \neq a\cdot 1$, then $\QRes_b(\p_1)\rewr \QRes_b(\p_2)$. If $\al=a\cdot 1$, then $\QRes_b(\p_1)=\QRes_b(\p_2)=(a,\epsi)$.
  \item Assume $\p_1\rewdown \p_2$. Then $\QRes_b(\p_1)\rewdown \QRes_b(\p_2)$.
\end{itemize}
This yields:}
\begin{lemma}
  \label{l:ref-red} \mbox{}
Let $\theta_1,\theta_2\in \ttThr$. We set $\theta'_i=\QRes_b(\theta_i)$.
  \begin{itemize}
  \item If $\theta_1\rrew \theta_2$, then $\theta_1'\rrew \theta_2$ or $\theta'_1=\theta'_2$.
  \item If $\theta_1\rrewto \theta_2$, then $\theta_1'\rrewto \theta'_2$ or $\theta'_1=\theta'_2$.
  \item If $\theta_1\rrewtt \theta_2$, then $\theta_1'\rrewtt \theta'_2,~ \theta'_1\rrewdown\theta'_2$ or $\theta'_1=\theta'_2$.
  \item If $\theta_1\rrewr \theta_2$, then $\theta_1'\rrewr \theta'_2$ or $\theta'_1=\theta'_2$.
  \item If $\theta_1\rrewdown \theta_2$, then $\theta_1'\rrewdown \theta'_2$ or $\theta'_1=\theta'_2$.
\end{itemize}
  \end{lemma}

Besides, $\Res_b(\thepsi)=\thepsi$ and $\Res_b(\thbot)=\thbot$ as expected, 
so \paper{L.}\techrep{Lemma}\;\ref{l:ref-red} implies that, if there is a nihilating chain for $t$ of length $m$, then there is one for $t'$ of length $\leqslant m$ (with $t\rew^* t'$).

\subsection{The Collapsing Strategy}
\label{s:collapsing-strategy}

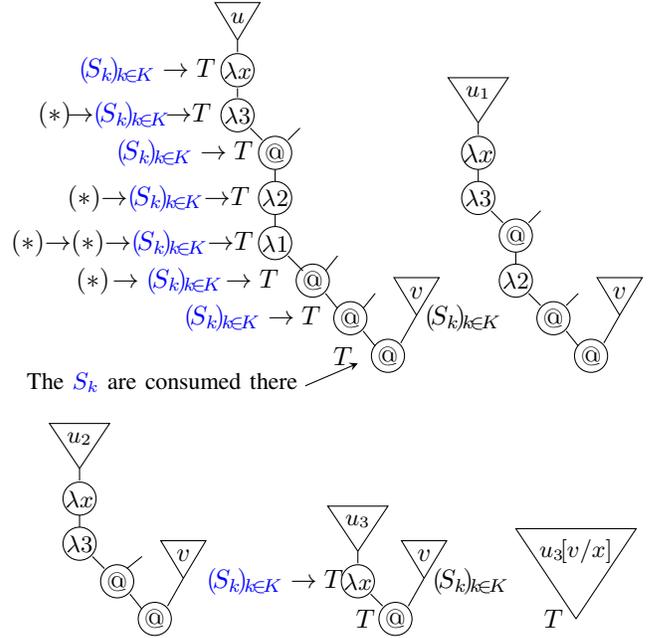
\begin{figure}
\begin{center}
  \begin{tikzpicture}

  
  \draw (0,4.5) --++ (-0.3,0.5) --++ (0.6,0) --++ (-0.3,-0.5);
  \draw (0,4.8) node {\small $u$};

  \draw (0,4.34) -- (0,4.5);
  
  \draw (0,4.1) node {\small $\lx$};
  \draw (0,4.1) circle (0.22);
  \draw (-1.2,4.1) node {$\blue{\sSksh}\rew T$};

  \draw (0,3.74)-- (0,3.86);
  
  \draw (0,3.5) node {\small $\lambda 3$};
  \draw (0,3.5) circle (0.22);
  \draw (-1.5,3.5) node {$(*)\!\rewsh \blue{\sSksh}\! \rewsh T$};
  
  \draw (0.17,3.33) -- (0.33,3.17);
  \draw (0.83,3.33)--  (0.67,3.17);
  
  \draw (0.5,3) node {$\arob$};
  \draw (0.5,3) circle (0.22);
  \draw (-0.7,3) node {$\blue{\sSksh}\rew T$};
  
  \draw (0.5,2.62) -- (0.5,2.78);

  \draw (0.5,2.4) node {\small $\lambda 2$};
  \draw (0.5,2.4) circle (0.22);
  \draw (-1.05,2.4) node {$(*)\rewsh \blue{\sSksh} \rewsh T$};

  \draw (0.5,2.02) -- (0.5,2.18);

  \draw (0.5,1.8) node {\small $\lambda 1$};
  \draw (0.5,1.8) circle (0.22); 
  \draw (-1.4,1.8) node {$(*)\rewsh(*)\rewsh \blue{\sSksh} \rewsh T$};
  
  \draw (0.67,1.63) -- (0.87,1.43);
  \draw (1.33,1.63) -- (1.17,1.43);
  
  \draw (1,1.3) node {$\arob$};
  \draw (1,1.3) circle (0.22);
  \draw (-0.8,1.3) node {$(*)\!\rew \blue{\sSksh}\! \rew T$};
  
  \draw (1.17,1.13) -- (1.33,0.93);
  \draw (1.83,1.13) -- (1.67,0.93);
  
  \draw (1.5,0.8) node {$\arob$};
  \draw (1.5,0.8) circle (0.22);
  \draw (0.2,0.8) node {$\blue{\sSksh}\rew T$};
  
  \draw (1.67,0.63) -- (1.83,0.43);

  \draw (2,0.3) node {$\arob$};
  \draw (2,0.3) circle (0.22);
  \draw (1.4,0.3) node {$T$};
  \draw (-1,0.2) node [below] {\small The $\blue{S_k}$ are consumed there}; 
  \draw [>=stealth,->] (0.9,-0.1) -- (1.6,0.2);
  
  \draw (2.17,0.47) --++ (0.21,0.37);
  \draw (2.38,0.84) --++ (-0.3,0.5) --++ (0.6,0) --++ (-0.3,-0.5);
  \draw (2.38,1.14) node {$v$};
  \draw (3,0.8) node {$\sSksh$};


  \draw (3.2,3.4) --++ (-0.4,0.6) --++ (0.8,0) --++ (-0.4,-0.6);
  \draw (3.2,3.8) node {\small $u_1$};

  \draw (3.2,3.24) -- (3.2,3.4);

  \draw (3.2,3) node {$\lx$};
  \draw (3.2,3) circle (0.22);

  \draw (3.2,2.62) -- (3.2,2.78);
  
  \draw (3.2,2.4) node {\small$\lambda 3$};
  \draw (3.2,2.4) circle (0.22);

  \draw (3.53,2.07) -- (3.37,2.23);
  \draw (3.87,2.07) -- (4.03,2.23);
  
  \draw (3.7,1.9) node {$\arob$};
  \draw (3.7,1.9) circle (0.22);

  \draw (3.7,1.52)-- (3.7,1.68);
  
  \draw (3.7,1.3) node {\small $\lambda 2$};
  \draw (3.7,1.3) circle (0.22); 

  \draw (4.03,0.97) -- (3.87,1.13);
  \draw (4.37,0.97) -- (4.53,1.13);
   
  \draw (4.2,0.8) node {$\arob$};
  \draw (4.2,0.8) circle (0.22);

  \draw (4.37,0.63) -- (4.53,0.43);

  \draw (4.7,0.3) node {$\arob$};
  \draw (4.7,0.3) circle (0.22);
   
  \draw (4.87,0.47) --++ (0.21,0.37);
  \draw (5.08,0.84) --++ (-0.3,0.5) --++ (0.6,0) --++ (-0.3,-0.5);
  \draw (5.08,1.14) node {$v$};

\trans{-8.3}{-3.5}{  
  \draw (6.2,2.3) --++ (-0.4,0.6) --++ (0.8,0) --++ (-0.4,-0.6);
  \draw (6.2,2.7) node {\small $u_2$};

  \draw (6.2,2.14) -- (6.2,2.3);

  \draw (6.2,1.9) node {\small $\lx$};
  \draw (6.2,1.9) circle (0.22);

  \draw (6.2,1.52) -- (6.2,1.68);
  
  \draw (6.2,1.3) node {\small $\lambda 3$};
  \draw (6.2,1.3) circle (0.22);
  
  \draw (6.53,0.97) -- (6.37,1.13); 
  \draw (6.83,0.97) -- (7.03,1.13);
  
  \draw (6.7,0.8) node {$\arob$};
  \draw (6.7,0.8) circle (0.22); 
  
  \draw (6.87,0.63) -- (7.03,0.43);
  
  \draw (7.2,0.3) node {$\arob$};
  \draw (7.2,0.3) circle (0.22);

  \draw (7.37,0.47) --++ (0.21,0.37);
  \draw (7.58,0.84) --++ (-0.3,0.5) --++ (0.6,0) --++ (-0.3,-0.5);
  \draw (7.58,1.14) node {$v$};
}

\trans{-8.5}{-3.5}{
  \draw (10.1,1.2) --++ (-0.4,0.6) --++ (0.8,0) --++ (-0.4,-0.6);
  \draw (10.1,1.6) node {\small $u_3$};

  \draw (10.1,1.04) -- (10.1,1.2);

  \draw (9,0.8) node {$\blue{\sSksh}\rew T$};
  \draw (10.1,0.8) node {\small $\lx$};
  \draw (10.1,0.8) circle (0.22);

  \draw (10.43,0.47) -- (10.27,0.63);
  
  \draw (10.6,0.3) node {$\arob$};
  \draw (10.6,0.3) circle (0.22);
  \draw (10.2,0.3) node {$T$};

 \draw (10.77,0.47) --++ (0.21,0.37);
  \draw (10.98,0.84) --++ (-0.3,0.5) --++ (0.6,0) --++ (-0.3,-0.5);
  \draw (10.98,1.14) node {$v$};
  \draw (11.6,0.8) node {$\sSksh$};


  \draw (13,0.3) --++ (-0.8,1.2) --++ (1.6,0) --++ (-0.8,-1.2);
  \draw (13,1.2) node {\small $u_3\![v/x]$};
  \draw (12.7,0.3) node {$T$};}
  \end{tikzpicture}
  \caption{Collapsing a Redex Tower}
  \label{fig:towered-redex}
\end{center}
\vspace{-0.7cm}
\end{figure}
We explain now how to \textit{normalize} a chain \ie discard negative left-consumption. 
 This will allow us to use Proposition\;\ref{prop:no-normal-chain} to finally conclude that nihilating chains do not exist.

 The idea is that if $\theta_1:\p_1 \lomin \rrewc{a} \theta_2$, then either $\tra$ is a redex and, by definition of $\Res_b$, we have $\Res_b(\theta_1)=\Res_b(\theta_2)$ (\ie $\theta_1$ and $\theta_2$ are \textbf{collapsed} by the reduction step) or there is a redex between $\p_1$ and $a$. When we reduce it, the relative height of $\p_1$ will decrease. More precisely, the 2nd case is associated to the notion of \textbf{redex tower}\paper{, which is more or less a \textit{finite} nesting of redexes, that can -- more importantly -- be collapsed in a \textit{finite} number of steps. A case of negative left-consumption of a sequence type $\sSk$ (which is the domain of the abstraction $\lx.u$), coming along with a redex tower, is represented  in Fig.\;\ref {fig:towered-redex} (by lack of space, we write $\lam 1,\lam 2,\ldots$ instead of $\lx_1,\lx_2,\ldots$ and $(*)$ stands for matterless sequence types).
   The sequence type $\sSk$ of negative polarity is ``called'' by the node $\lx$ at the top of the figure and consumed at the bottom $\app$-rule. 
  The initial redex tower is reduced in 4 steps, so that its height decreases and finally, the types $S_k$, that were left-consumed negatively, are destroyed in the final term $u_3[v/x]$.  }. 
\techrep{In Fig.\;\ref{fig:towered-redex}, we have represented a redex tower of \textit{height} 7, meaning that the sequence type $\sSksh$ called by the abstraction $\lx$ is left-consumed on its prefix of rank 7. It is very similar to that of Fig.\;\ref{fig:neg-left-cons-and-tow-redex-informal} p.\;\pageref{fig:neg-left-cons-and-tow-redex-informal}. 
The types of some subterms are indicated on their left (except for $v$) \eg subterm $\lx.u$ has type $\sSksh\rew T$.
We write $\lam 1,\, \lam 2,\, \lam 3$ for $\lx_1,\,\lx_2,\,\lx_3$ and $(*)$ for sequence types of no matter.

At each step of reduction, we reduce the height of consumption by 2.
Thus, in a finite number of steps, we get a proper redex and we collapse $\theta_1$ on $\theta_2$. This is illustrated by Fig\;\ref{fig:towered-redex} or Fig.\;\ref{fig:coll-tow-redex-informal}. We  now formalize this argument.\\
\paragraph*{Formalizing the Collapsing Strategy}  
Assume then $\theta_1:\p_1=(\al\cdot 1,k\cdot c) \lomin \rrew \theta_2$. Then $\Asc{\theta_1}=(\al_*,k\cdot c)$ and $t(a_*)=\lx$ for some $\al_*,x$. We set $h=|a_*|-|a|$ and we call $h$ the \textbf{height} of the consumption. By Lemma~\ref{l:form-idena}, for $1\leqslant i \leqslant h$, we may write $\p_i=(\al_i,1^{j(i-1)}\cdot k\cdot c)$ for $\asc^{i-1}(\theta_1)$ where $\al_{i+1}=\al_i\cdot k_i$ for a $k_i\in \set{0,1}$.

If $h=1$, then $\al_*=\al \cdot 1$ and we set $b=\ovl{\al}$ so that $t\rstr{b}$ is a redex and $\QRes_b(\p_1)=\QRes_b(\p_2)$. Thus, $\Res_b(\theta_1)=\Res_b(\theta_2)$.

Assume now $h>1$. Then let $1\leqslant i_0\leqslant h-1$ be maximal s.t. $t(a_{i_0})=\arob$. Actually, $i_0\leqslant h-2$ (if $i_0=h-1$, $t\rstr{\al_{h-1}}$ is a redex, so $\al_{h-1}=\al$ \ie $h=1$).

We set $b=\ovl{\al^{i_0}}$ so that $t\rstr{b}$ is a redex. We set $\p'_i=(\al'_i,c'_i)=\QRes_b(\p_i)$ (so that $\p'_{i_0}=\p'_{i_0+1}=\p'_{i_0+2}$).
By induction on $i$, if $1\leqslant i\leqslant i_0$, then $(\asc')^{i-1}(\p_1)=\p'_i$ and $|\al'_i|-|\al|=i$ and if $i_0<i\leqslant h-2$, then $(\asc')^{i-1}(\p_1)=\p'_{i+2}$ and $|\al'_{i+2}|-|\al|=i$. Thus, $(\asc')^{h-2}(\p_1)=(\al'_h,k\cdot c)$ and $|\al'_h|-|\al|=h-2$. Since $t'(\al'_{h-2})=t(\al_h)$ (proper residual), $\Ascp{\p'_1}=\p'_h$ and  $\Polp{\p'_1}=\omin$. Thus, $\theta'_1=\Res_b(\theta_1)\lomin \rew \theta'_2=\Res_b(\theta_2)$, but the height has decreased by 2. In a finite number of steps, we equate then $\theta_1$ and $\theta_2$.

We extend the notation $\Res_b$ for finite reduction sequences, so that we get this lemma (whose informal version is Observation\;\ref{obs:neg-left-cons-tow-red-and-coll-strat}, p.\;\pageref{obs:neg-left-cons-tow-red-and-coll-strat}): }

\begin{lemma}
\label{l:collapsing-strategy}
If $\theta_1\lomin \rrew \theta_2$, then there is a reduction strategy $\rs$ such that $\Res_{\rs}(\theta_1)=\Res_{\rs}(\theta_2)$.
\end{lemma}

\noindent This Lemma, along with\techrep{ the observation concluding Sec.\;\ref{s:residuals-c-unsound}}\paper{ the conclusion of \S\;\ref{ss:residuation-subjugation}}, yields:

\begin{proposition}
\label{prop:collapsing-strategy}
There is a reduction strategy, that we call the \textbf{collapsing strategy}, producing a normal chain from any nihilating chain.
\end{proposition}

\techrep{
  \subsection{Redex Towers}
\label{s:towered-redex-very-formal}
  
  The notion of redex tower, illustrated by Fig.\;\ref{fig:towered-redex} and \ref{fig:neg-left-cons-and-tow-redex-informal}, can be formalized. We will use it again in the course of proving Theorem\;\ref{th:order}.\\

  First, we can easily extend the notion of ascendance to positions of $\bbA$.
  So we set, for all $a\in \bbA$, $\asc(a)=a\cdot 1$ if $t(a)=\arob$, $\asc(a)=a\cdot 0$ if $t(a)=\lx$ ($\asc(a)$ is undefined when $t(a)=x$). We also define $\Asc{a}$ as $\asc^i(a)$ where $i$ is maximal such that this expression is defined.

  Now, let $a\in \bbA$. We define first the \textbf{consumption degree} $\cdeg_a(a')$ of some extensions $a'$ of $a$. Intuitively, to compute $\cdeg_a(a')$, we visit the ascendants of $a$, one after another. The postfix 1 means that we visit an application left-hand side and we increment $\cdeg_a(a')$, and the postfix 0 means that we visit an abstraction and we decrement $\cdeg_a(a')$. 

  For instance, in Fig.\;\ref{fig:towered-redex}, if $a$ is the adress of the root of the redex tower, we have $\cdeg_a(a\cdot 1)=1,\; \cdeg_a(a\cdot 1^2)=2,\; \cdeg_a(a\cdot 1^3)= 3,\; \cdeg_a(a\cdot 1^3\cdot 0)=2,\; \cdeg_a(a\cdot 1^3\cdot 0^2)=1,\; \cdeg_a(a\cdot 1^3\cdot 0^2 \cdot 1)=2,\; \cdeg_a(a\cdot 1^3\cdot 0^2 \cdot 1\cdot 0)=1,\;
  \cdeg_a(a\cdot 1^3\cdot 0^2 \cdot 1\cdot 0^2)=0$ where $t(a\cdot 1^3\cdot 0^2\cdot 1 \cdot 0^2)=\lx$. Thus, the consumption degree reaches 0 for the first time when the abstraction of the redex tower is reached.

So, we define by induction:
  \begin{itemize}
  \item $\cdeg_a(a)=0$.
    \item If $\cdeg_a(a')\geqslant 0$, then $\cdeg_a(a'\cdot 1)=\cdeg_a(a')+1$ (resp. $\cdeg_a(a'\cdot 0)=\cdeg_a(a')-1$) when $t(a)=\arob$ (resp. $t(a)=\lx$).
    \item If $\cdeg_a(a')=-1$, then $\cdeg_a(a'\cdot 1)$ (resp. $\cdeg_a(a\cdot 0)$) is undefined when $t(a')=\arob$ (resp. $t(a')=\lx$).
  \end{itemize}

 As it can also be seen with Fig.\;\ref{fig:towered-redex}, if $d=\cdeg_a(a')\geqslant a$, then $\ttT(a')$ is an arrow type ending with $\ttT(a)$ and, more precisely, $\ttT(a)$ is preceded in $\ttT(a')$ by exactly $d$ arrows.   

  \begin{definition}
    Let $t$ a term and $b\in \supp{t}$ such that $t(b)=\arob$. We say there is a \textbf{redex tower} of height $h\geqslant 1$ if
    \begin{itemize}
    \item $\cdeg_b(\asc^h(b))=0$
    \item For all $1\leqslant i \leqslant h-1$, $\cdeg_b(\asc^i(b))>0$.
    \end{itemize}
  \end{definition}

  \noindent The extension below will also be useful to prove Theorem\;\ref{th:order}. Inductively, a redex tower sequence is a redex tower whose left-hand side (the part over the abstraction) is itself a redex tower sequence. We can then collapse one redex tower after the other. See \S\;\ref{s:formalizing-asc-threads-etc} for a high-level presentation and in particular, Fig.\;\ref{fig:tow-red-seq}.

  \begin{definition}
    \label{def:towered-redex-sequence}
    Let $t$ a term and $b\in \supp{t}$ such that $t(b)=\arob$. We say there is a \textbf{redex tower sequence} of height $h\geqslant 1$ at position if
$\cdeg_b(\asc^h(b))=0$.
  \end{definition}

For the proof of Theorem\;\ref{th:order} (order discrimination),  we need to describe the thread $\thepsi$. So, for $n\geqslant 0$, we write $\epsi(n)=\asc^n(\epsi),\ \pepsi(n)=\asc^n(\epsi,\epsi)$ and $\phdadd{\tti(n)}$ the naturel number such that $\pepsi(n)=(a(n),1^{\tti(n)})$ when they are defined (if $\pepsi(n)$ is defined, then $a(n)$ is defined; the form $1^{\tti(n)}$ is justified by Lemma\;\ref{l:form-idena}).\\  
The following lemma is just the formalization of Observation\;\ref{obs:tow-red-seq-and-order}, p.\;\pageref{obs:tow-red-seq-and-order}:

\begin{lemma}
\label{l:epsi-epsi-neg}
Let $t$ be a term and $\code{\cdot}:\bbN^*\rew \argN$ an injection and all the associated notations ($\asc$, $\ttPol$, $\rrewb$ etc).\\ 
If $\Pol{\epsi,\epsi}=\omin$, then there exist $a\in \bbA$ and a finite segment $\rs$ of the head reduction strategy such that $\asc(\epsi,\epsi)=(a,\epsi)$, $t\stackrel{\rs}{\bred} \lx.t'$ and $\Res_\rs(a,\epsi)=(\epsi,\epsi)$. 
\end{lemma}
\begin{proof}
   Assume $\Pol{\epsi,\epsi}=\omin$.
  By Lemma\;\ref{l:form-idena}, $\Pol{\epsi,\epsi}=\omin$ means that $\thepsi$ only occurs  negatively (an ascendant of $(\epsi,\epsi)$ is of the form $(a',1^i)$ so cannot be on the left-hand side of $\rewp$). In that case, if $(a,c)=\Asc{\epsi,\epsi}$, we then have $t(a)=\lx$ (negativity) and $c=\epsi$ (if we had $c=1^i$ with $i\geqslant 1$, then $(a,c)$ would have an ascendant).\\

 The claim about the head reduction strategy is proved using exactly the same technique as in Sec.\;\ref{s:collapsing-strategy}: more precisely, we destroy the highest redex tower sequence (Definition\;\ref{def:towered-redex-sequence}) rooted at $\epsi$, which yields a reduct of $t$ that is an abstraction $\lx.t'$.\fublainv{preuve peu claire...}
  \end{proof}

}

\section{Applications}
\label{s:apps-comp-unsound}

We can now prove that every term is $\code{\cdot}$-typable (and thus, also $\scrR$-typable), using the residuation of threads, the collapsing strategy and the non-existence of normal threads, which is ensured by the Interaction Lemmas. 
The same methods can be used to prove order-discrimination in systems $\scrR$ and $\ttS$.

\begin{theorem}
  \label{th:comp-unsoundness-R-S}  
    Every $\lam$-term is typable in the relevant and non-idempotent intersection type system $\scrR$.
\end{theorem}

\begin{proof}
We just need to achieve Goal\;\ref{goal:code-S-typability}.
 \techrep{ By Propositions \ref{prop:if-no-nihil-ch-then-comp-unsound},~\ref{prop:no-normal-chain} and \ref{prop:collapsing-strategy}:}
  \begin{itemize}
  \item By Proposition\;\ref{prop:no-normal-chain}, there is no normal nihilating chain.
  \item By using the collapsing strategy (Proposition\;\ref{prop:collapsing-strategy}), if  nihilating chains existed, so would the normal nihilating chains. Thus, nihilating chains do not exist.
  \item By Proposition\;\ref{prop:collapsing-strategy}, the non-existence of nihilating chains entails that every term is $\code{\cdot}$-typable\paper{ and $\scrR$-typable}.
\techrep{  \item By Goal\;\;\ref{goal:code-S-typability} and the afferent discussion, this is enough to prove that every term is $\scrR$-typable.}
  \end{itemize}  
\end{proof}

\paper{
  \noindent System $\scrR$  discriminates terms \wrt their orders\techrep{ (Definition\;\ref{def:order})}, as claimed:}\techrep{Now that we know that every term may be typed, we can study some semantical aspects of system $\scrR$, namely, how it may be order-sensitive.

  \subsection*{Order-Discrimination}
}
  
\paper{
\begin{lemma}
  \label{l:zero-term-o}
Let $t$ be a zero term and $\tv$ a type variable, then there is context $C$ such that $\ju{C}{t:\tv}$ is $\ttS$-derivable.
\end{lemma}

\begin{proofsketch} \fublainv{à changer avec \scrR ou Dw au lieu de \scrD}
  Let $t$ be a term s.t. $\thepsi \rrewb^* \thr{\epsi,1}$ \ie s.t.
  $(\epsi,1)\in \bbBm$ (see Corollary\;\ref{corol:typability-and-thepsi}), which implies that the type of $t$ cannot be a type variable by the proof of this same corollary. We prove that $t$ is of order $\geqslant 1$, which is enough to conclude.

  For that, we consider a \textit{$\lam$-chain} \ie a chain of the from $\thepsi=\theta_0\rrewb \ldots \rrewb \theta_m=\thr{\epsi,1}$, of minimal length. The notion of normal chains extends to $\lam$-chains and by the collapsing strategy, we can replace $t$ by a reduct $t'$ s.t. the considered chain is normal.

Using \textit{ad hoc} interaction lemmas, we prove that the normality of the chain entails $\thepsi\rrewr \thr{\epsi,\epsi}$. Collapsing then redex towers, we may reduce $t'$ to an abstraction $\lx.t^{\prime\prime}$, thus proving that $t$ is of order $\geqslant 1$.  
\end{proofsketch}

\begin{theorem}
  \label{th:order}
  Let $t$ be a term of order $n$. Then there is a context $\Gam$ and a type $\tau$ of order $n$  (see Sec.\,\ref{ss:typing-examples}) such that $\juGtt$ is derivable in system $\scrR$.
\end{theorem}

\begin{proofsketch}
When $n=\infty$, this comes from Theorem~\ref{th:comp-unsoundness-R-S}, subject reduction and the $\abs$-rule. When $n\in \bbN$, we use Lemma~\ref{l:zero-term-o} and (finite) subject expansion (Proposition\;\ref{s:S-sr-se}).
\end{proofsketch}}
\techrep{
\noindent The goal of this section is to prove:
  
\begin{proposition}
  \label{prop:zero-term-o}
Let $t$ be a zero term and $o$ a type variable, then there is context $C$ such that $\ju{C}{t:o}$ is derivable in system $\ttS$.
\end{proposition}

\noindent Exactly as in \S\;\ref{s:order-discr-Ro}, this proposition easily entails:

\begin{theorem}
  \label{th:order}
  Let $t$ be a term of order $n$.
 Then there is a context $\Gam$ and a type $\tau$ of order $n$  (see Sec.\,\ref{s:typing-examples}) such that $\juGtt$ is derivable in system $\scrR$.
\end{theorem}

\begin{proof}
  Let $t$ be a term of order $\geqslant n$ and (by Theorem\;\ref{th:comp-unsoundness-R-S}) $\Pi$ a $\scrR$-derivation concluding with $\ju{\Gam}{t:\tau}$.

  There is a term $t'=\lx_1\ldots \lx_n.t'_0$ such that $t\bred^* t'$. By subject reduction, there is a derivation $\Pi'$ concluding with $\ju{\Gam}{\lx_1\ldots \lx_n.t'_0:\tau}$. By the $\abs$-rule, $\tau$ must be of order $\geqslant n$. This implies that the statement is true for terms of infinite order.\\
  
  Now, assume that $t$ is of order $n<\infty$.  There is a term $t'=\lx_1\ldots \lx_n.t'_0$ such that $t\bred^* t'$ and $t'_0$ is of order 0.

  Let $\tv \in \TypeV$. By Propositions\;\ref{prop:zero-term-o} and \ref{prop:collapse-S-D}, there is a $\scrR$-derivation $\Pi'_0$ concluding with $\ju{\Gam_0}{t'_0:\tv}$. Using $n$ $\abs$-rules, we get a derivation $\Pi'$ concluding with a judgment of the form $\ju{\Gam}{\lx_1\ldots \lx_n.t'_0:B}$, with $B$ of order $n$.

  By subject expansion, there is a derivation $\Pi$ concluding with 
$\ju{\Gam}{t:B}$. The statement is thus proved.
\end{proof}

Now, let us prove Proposition\;\ref{prop:zero-term-o}. For that, we consider $t$ be a term such that $\thepsi \rrewb^* \thr{\epsi,1}$ \ie s.t.
  $(\epsi,1)\in \bbBm$ (see Corollary\;\ref{corol:typability-and-thepsi}), which implies that the type of $t$ cannot be a type variable by the proof of this same corollary. We prove that $t$ is of order $\geqslant 1$, which is enough to conclude.

  We then consider a \textit{$\lam$-chain} \ie a chain of the form $\thepsi=\theta_0\rrewb \ldots \rrewb \theta_m=\thr{\epsi,1}$, of minimal length. The notion of normal chains extends to $\lam$-chains and by the collapsing strategy, we can replace $t$ by a reduct $t'$ such that the considered chain is normal.

  Assume $\Pol{\epsi,\epsi}=\ominus$. Then, by Lemma\;\ref{l:epsi-epsi-neg}, there is a reduct of $t'$ of the form $\lx.t\secu$. So $t\bred^* \lx.t\secu$ and we may conclude.\\

 We prove now \textit{ad absurdum} that $\Pol{\epsi,\epsi}=\oplus$ is impossible. So we assume $\Pol{\epsi,\epsi}=\oplus$ for the remainder of the proof.

 We need to consider $B_0=\set{(a,c)\in \bbA\;|\; \thepsi \rrewto^* \thr{a,c} }$. To obtain a contradiction, we prove that $B_0$ is closed under normal chains but that $(\epsi,1)$ is not in $B_0$.

  We have to consider two different subcases: when $\thepsi$ also has some negative occurrences and when it does not. In both cases, we should prove that no relation allows exiting $B_0$ (except $\lomin\!\rew$), reaching thus a contradiction. 

  However, the second case may be see as a particular case of the first one and may be skipped. So we assume now that there is $\p \iden (\epsi,\epsi)$ such that $\Pol{\p}=\omin$.

This implies, with $h$ such that $\Asc{\epsi,\epsi}=\asc^h(\epsi,\epsi)$, that $t(\epsi(h))=x$ and  there is $0\leqslant \hlam < h$ such that $t(\epsi(\hlam))=\lx$.

Thus,
$(\epsi(\hlam),k \cdot 1^{\tti ( h )})\rewp \p_h=\Asc{\epsi,\epsi}$ with $k=\code{\epsi(h)}$.
For all $0\leqslant \hlam$ such that it is defined, we write $\pepsi'(n)=(\epsi(n),1^{\ttj (n) \cdot k \cdot \tti (h) })$ for the unique $\pepsi'(n)\in \bbB$ such that $\asc^{\hlam-n}(\pepsi'(n))=(\epsi(\hlam),k\cdot 1^{\tti(h)})$. The form $(\epsi(n),1^{\ttj(n) \cdot k\cdot \tti(h)})$ is justified by Lemma\;\ref{l:form-idena}.
We write $h_0$ for the minimal integer such that $\pepsi'(h_0)$ is defined.

Notice that $\ttj(\hlam)=0<\tti(\hlam)$: if we had $\tti(\hlam)=0)$, we could prove that $\Pol{\epsi,\epsi}=\ominus$, since this implies
 $\asc^{\hlam}\epsi,\epsi)=(\epsi(\hlam),\epsi)$. Moreover, $t(\epsi(\hlam))=\lx$, so that $\asc^{n+1}(\epsi,\epsi)$ is not defined \ie $\Asc{\epsi,\epsi}=(\epsi(\hlam),\epsi)$ and $\Pol{\epsi,\epsi}=\ominus$, which contradicts our assumption.


From $\ttj(\hlam)<\tti(\hlam)$, by an easy induction, we deduce that, for all $h_0\leqslant n \leqslant \hlam$, $\ttj(n)<\tti(n)$.

\begin{lemma}
  Assume that $\Pol{\epsi,\phdadd{\epsi}}=\oplus$. Then $B_0=\set{(\epsi(n),1^i)\,|\, n\leqslant h,\, i \leqslant \tti(n)}\cup \set{(\epsi(n),1^{\ttj(n)\cdot k \cdot i}\,|\, \hlam \leqslant n\leqslant h,\, i\leqslant \tti(h)}$. 
  \\
  Moreover, $B_0$ is closed under $\rewdown,\, \rewr,\ \rewtt$ and $B_0$ has an empty intersection with $\dom{\loplus\!\rew} \cup \codom{\rew}$
\end{lemma}

\begin{proof}
  Let $B_0'=\set{(\epsi(n),1^i)\,|\, n\leqslant h,\, i \leqslant i}$ and
  $\set{(\epsi(n),1^{\ttj(n)\cdot k \cdot i}\,|\, \hlam \leqslant n\leqslant h,\, i\leqslant \tti(h)}$

  Clearly, $B'_0\cup B\secu_0\subseteq B_0$. We study $B'_0\cup B\secu_0$ and prove that $B'_0\cup B\secu_0=B_0$.

  We notice first that if $\p=(a,c) \in B'_0\cup B\secu_0$ is such that $c$ contains an argument track, then $\p\in B\secu_0$ and $\Pol{\p}=\omin$.
\begin{itemize}
\item $B'_0\cup B\secu_0$ is clearly closed under $\iden$ and $\rewdown$.  
\item Closure under $\rewto$: the only problematic case is $\p_0=(\epsi(n),1^{\ttj(n)\cdot j}\cdot k)\rewto (\epsi(n),1^{\ttj(n)\cdot j}\cdot k)$. Since $\p_0\in B\secu_0$, we must show that $\p\in B'_0$. This is guaranteed by $\ttj(n)<\tti(n)$.
\item $\p_0=(a_0,c_0) \loplus\rew \p$ with $\p_0\in B'_0\cup B\secu_0$, because $\p_0\rew \p$ implies that $c_0$ holds an argument track, and thus, as notice above, that $\Pol{\p_0}=\ominus$.
\item $\p\rew (\epsi(n),c_0)\in B'_0\cup B\secu_0$ is impossible, because $\epsi(n)\in \set{0,1}^*$ (no argument track in $\epsi(n)$).
\item Closure under $\rewtt$: assume that $\p_0=(\epsi(n),1^{\ttj(n)}\cdot k)\rewtt (\epsi(n),1^{\ttj(n)+1})$. Since $\p_0\in B\secu_0$, we must show that $\p\in B'_0$. This is guaranteed by $\ttj(n)<\tti(n)$, so that $\ttj(n)+1\leqslant \tti(n)$.
\item Assume $\p_0=(\epsi(n),\epsi)\rewr (\epsi(n),1)=\p$ so that $\p_0$ in $B'_0$. In order to have $\p\in B'_0$, it is enough  to prove that $\tti(n)\geqslant 1$. If we had $\tti(n)=0$, then we could prove that $\Pol{\epsi,\epsi}=\ominus$ as is the proof of $\tti(\hlam)>0$ above.
\end{itemize}
By point 1, $B'_0\cup B\secu_0$ is stable under $\iden$. Moreover, $B'_0\cup B\secu_0$ contain $\thepsi$ and is stable under $\rewto$ by point 2. So $B'_0\cup B\secu_0\supseteq B_0$. Thus, $B_0=B'_0\cup B\secu_0$ and the Lemma is proven.
\end{proof}

By the Lemma, $(\epsi,1)\notin B_0$, so let $k$ be minimal such that $\theta_k\notin B_0$.
We are now interested in the relation $\theta_k\rrewb \theta_{k+1}$ in the chain. According to the Lemma, this cannot be $\rrewtt,\ \rrewdown,\ \loplus\rrew,\ \werr$ or $\rrewr$. Contradiction since the chain is normal.\\

Thus, $\Pol{\epsi,\epsi}=\ominus$, and we have proven above that this implied that $t$ was of order $\geqslant 1$. This concludes the proof of Proposition\;\ref{prop:zero-term-o}.}


\noindent \textit{Conclusion:} 
\label{ccl-comp-unsound}
 We proved that every term is typable in a reasonable relevant intersection type system (Theorem\;\ref{th:comp-unsoundness-R-S}).
 If we take the typing rules of $\ttS$ \textit{coinductively} (and not only the type grammar), we can also type every infinitary $\lam$-term~\cite{KennawayKSV97}. 


 The techniques that we have developped here build, to the best of our knownledge, the first bridge  between first-order model theory and the study of models of the pure $\lam$-calculus. They are actually modular: we also used them, in a companion paper, to prove that every multiset-based derivation is the collapse of a sequential derivation~\fucite{}. This suggests that these techniques could be used to study the coinductive version of finitary models of the $\lam$-calculs and extend some of their semantical properties to all $\lam$-terms.

\paper{By setting, for all term $t$,
 $\denot{t}_{\mathtt{rel}\infty}=\set{\tri_{\scrR} \juGtt\;|\; \Gam,\tau}$ (\cf \S\;\ref{ss:intro-search-inf-denot}), one defines the \textbf{infinitary version of the relation model}. Theorem\;\ref{th:comp-unsoundness-R-S}  in which, by Theorem \ref{th:comp-unsoundness-R-S}, no term has a trivial denotation, including the mute terms. This model is thus \textbf{non-sensible}\;\cite{Berarducci96} since it does not equate all the non-head normalizing terms (\eg $\Om$ and $\lx.\Om$ of respective order 0 and 1) by Theorem\;\ref{th:order}.
 }

We presented a first semantical result about this model (Theorem\;\ref{th:order}), but its equational theory has yet to be studied.
According to the same theorem, this model equates all the closed zero terms. It then differs both from the non-sensible model of Berarducci trees and that of L\'evy-Longo trees, respectively related to $\Lamuuu$ and $\Lamzzu$ in \cite{KennawayKSV97}. This work may suggest a new notion of tree, that could shed some light on Open Problem \#\;18 of TLCA (the problem of finding trees related to various contextual equivalences).

The study of infinitary models (beyond infinite tree models) is at its early stages, but it already provides descriptions of the infinitary behaviors of $\lam$-terms (\cf Grellois-Melli\`es' infinitary model of Linear Logic in \cite{GrelloisM15,GrelloisM15a}). 
The semantical implications of the main theorem (every term is $\scrR$-typable) remain to be understood and the proof techniques presented here can certainly be used to study infinitary models or coinductive/recursive type systems \textit{before}  they are endowed with some validity or guard condition, or maybe to build other models of pure $\lam$-calculus, for instance, to get some semantical proof of the \textit{easiness}~\cite{Jacopini75} of sets of mute terms, as in \cite{BucciarelliCFS16}.

\bibliographystyle{abbrv}
\bibliography{../../../these-pvial/biblioThese}

\end{document}